\documentclass[11pt]{article}

\usepackage{color}

\usepackage{graphicx,tabularx,array,geometry,amsmath,amsthm,thmtools}
\usepackage{dsfont,graphicx,color,amsfonts,epsf,epsfig,verbatim,amssymb,amsmath,array,cite,amsthm,subfigure,multicol,multirow}

\usepackage[toc,page]{appendix}
\usepackage{caption}
\usepackage{amsfonts}
\usepackage{bbm}
\usepackage{makecell}
\usepackage{multirow}
 \usepackage{amssymb}
\usepackage{amsthm}
\usepackage{txfonts}
\usepackage[T1]{fontenc}
\usepackage{tikz}
\usepackage[scr=dutchcal]{mathalfa}
\let\mathscr\mathbscr
\newcolumntype{x}[1]{>{\centering\arraybackslash}p{#1}}

\newcommand\diag[4]{%
  \multicolumn{1}{p{#2}|}{\hskip-\tabcolsep
  $\vcenter{\begin{tikzpicture}[baseline=0,anchor=south west,inner sep=#1]
  \path[use as bounding box] (0,0) rectangle (#2+2\tabcolsep,\baselineskip);
  \node[minimum width={#2+2\tabcolsep-\pgflinewidth},
        minimum  height=\baselineskip+\extrarowheight-\pgflinewidth] (box) {};
  \draw[line cap=round] (box.north west) -- (box.south east);
  \node[anchor=south west] at (box.south west) {#3};
  \node[anchor=north east] at (box.north east) {#4};
 \end{tikzpicture}}$\hskip-\tabcolsep}}

\declaretheoremstyle[headfont=\bfseries, 
    bodyfont=\normalfont]{normalhead}
\declaretheorem[style=normalhead]{Example}

\newtheorem{Theorem}{Theorem}
\newtheorem{Lemma}{Lemma}
\newtheorem{Remark}{Remark}
\newtheorem{Definition}{Definition}

\allowdisplaybreaks 

\begin{document}

\title{An Achievable Rate-Distortion Region for Multiple
  Descriptions Source Coding Based on Coset Codes}
\author{Farhad Shirani  and S. Sandeep Pradhan\thanks{This work was supported by NSF grant
      CCF-1111061 and CCF-1422284. This work was presented in part at IEEE
      International Symposium on Information Theory (ISIT), July
      2014. } \\
Dept. of Electrical Engineering and Computer Science \\
 Univ. of  Michigan, Ann Arbor, MI.  \\}
\date{}

\maketitle \thispagestyle{empty} \pagestyle{plain}

\vspace{-0.5in}

\begin{abstract}
We consider the problem of multiple descriptions (MD) source coding
 and propose new coding strategies involving both unstructured and
structured coding layers. Previously, the most general
achievable rate-distortion (RD) region for the $l$-descriptions problem was the Combinatorial
Message Sharing with Binning (CMSB) region. The CMSB scheme utilizes
unstructured quantizers and unstructured binning. In the first part of
the paper, we show that this strategy can be improved upon using
more general unstructured quantizers and a more general unstructured binning method. In the
second part, structured coding strategies are considered. First,
structured coding strategies are developed by considering specific MD
examples involving three or more descriptions. We show that
application of structured quantizers results in strict RD improvements
when there are more than two descriptions. Furthermore, we show that
structured binning also yields improvements. These improvements are in
addition to the ones derived in the first part of the paper. This
suggests that structured coding is essential when coding over more
than two descriptions. Using the ideas developed through these
examples we provide a new unified coding strategy by considering several structured
coding layers. Finally, we characterize its performance
in the form of an inner bound to the optimal rate-distortion region 
using computable single-letter information quantities.
The new RD region strictly contains all of the previous known
achievable regions.      

\end{abstract}


\section{Introduction}
The Multiple-Descriptions (MD) source coding problem arises naturally
in a number of applications such as transmission of video, audio and
speech over packet networks and fading channels
\cite{goyal}\cite{videoaudio}. 
The multiple-descriptions (MD) source coding setup describes a
communications setting consisting of one encoder and several
decoders. The encoder receives a discrete memoryless source and
wishes to compress it into several \textit{descriptions}. Each decoder
receives a specific subset of these descriptions through noiseless
links, and produces a reconstruction of the source vector with
respect to its own distortion criterion. The parameters of interest
are the rates required for transmitting the description and the
resulting distortions at the decoders. The objective is to design 
communications schemes which result in the optimal asymptotic trade-off between
these two groups of parameters.  
The problem has been studied extensively
\cite{Ozarow}\cite{Ahlswede}\cite{EGC}\cite{ZB}\cite{VKG}\cite{CMS},
however, the optimal asymptotically achievable rate-distortion (RD) is not known even
for the most elementary case when only two descriptions are
considered. The two-descriptions setup is depicted in Figure
\ref{2Desc}. Evidently, for the individual decoders (which receive
only one description) to perform
optimally the encoder must transmit the two-descriptions according to
the optimal Point-to-Point (PtP) source coding schemes. This may
require the two-descriptions to be similar to each other. On the other hand, if the
descriptions are similar, one of them would be redundant at the
central decoder (which receive two descriptions). In fact, this
decoder requires the two-descriptions
to be different from one another in order to yield a better
reconstruction. The main challenge in the MD problem is to strike a
balance between these two situations. The best known achievable region
for the this communications setting is due to Zhang and Berger
\cite{ZB}. In the Zhang-Berger (ZB) strategy, the encoder
in the first step sends a common and  coarsely quantized version of the source 
on both descriptions, then
in the next step, the encoder sends individual refinements for each
decoder on the corresponding descriptions. The ZB coding strategy was
generalized in \cite{VKG} for the case where there are more than two
descriptions. In this strategy, first,
a common coarsely quantized version of the source is sent to all the decoders, then in the next
step, several refinement layers are transmitted. For the
symmetric $l-$descriptions problem,  a coding scheme based on random
binning was considered in \cite{symL1} which outperforms the VKG
scheme.   This involves generation of independent codebooks followed
by random binning.  Although the MD problem has a centralized encoder,
the strategy involving random binning was proved to be useful.  This
was further improved upon by a new coding scheme in \cite{symL2} based
on certain parity-check codes. However all the three schemes do not fully
exploit the common-information among every subset of individual
descriptions. For example in the three-descriptions problem, there can
be common-information between the first and second descriptions which
is not common with the third description.  A new coding scheme called
Combinatorial Message Sharing with Binning (CMSB) was considered in
\cite{CMS,CMS3} which provided a unified achievable RD region for the
general $l$-descriptions problem. This scheme provided a grand
unification of the schemes based on conditional codebooks and the
schemes based on random binning, which in turn results in the largest
achievable RD region for the problem and subsumes all previous coding
schemes.  The name is due to the combinatorial number of
common-component codebooks present.   It can be noted that CMSB scheme
is based on a construction of random codes where the codewords are
mutually independent, and where the codebooks do not have any algebraic
structure. 

In this paper, we provide a new coding strategy for the general
$l$-descriptions problem which strictly subsumes CMSB strategy which
is the best known in the literature till now. 
The coding strategy is based on the common-information
perspective. Taking a cue from the two-descriptions ZB strategy, we
propose that for the general $l$-descriptions problem the encoder
constructs a common constituent codebook for each subset of the 
$2^l-1$ decoders.  So,  for each subset of the decoders there 
is one common component  in the overall coding scheme.
This implies that the number of constituent codebooks grows double-exponentially in $l$. 
However, we prove that only an asymptotically exponential number of the codebooks are
necessary in terms of contributing to the rate-distortion region, and
the rest are redundant. This significantly simplifies the coding strategy. As an
example, for the $l=3$ case, there are $2^{2^3-1}=128$ possible
common code components,  but only $17$ of the
corresponding codebooks are non-redundant. It turns out that one can identify all of the
non-redundant codebooks by associating them with the Sperner families
of sets \cite{sperner}. As a result, we call the new scheme the Sperner Set Coding
(SSC) scheme. The CMSB scheme utilizes $14$ codebooks for the
$3$-descriptions problem. We prove analytically that the addition of the $3$ new
codebooks in the SSC scheme results in an improved achievable RD
region. In other words, we show analytically that the CMSB scheme is
not complete.  Additionally, we propose a generalized binning approach which
improves upon the CMSB scheme and further enhances the SSC
scheme. We characterize the asymptotic performance of this coding
scheme using computable single-letter information quantities. 
This forms the first part of the paper. Similar to the coding scheme
of CMSB, the SSC scheme uses random unstructured codes. 

It has been observed in several other multi-terminal communications settings
such as the Broadcast Channel (BC) \cite{BC}, Interference Channel
(IC) \cite{IC}, variations of the MAC channel \cite{Nazer}\cite{Zamir}
and the Distributed Source Coding (DSC) problem  \cite{DSC}, that the
application of algebraic structured codes results in improvements over random
unstructured codes in the asymptotic performance limits.
 Based on the inherent dualities between the
multi-terminal communication problems and the corresponding coding
schemes, these observations suggest that one may get such gains in
performance even in the MD problem. 

In the second part of the paper we show that SSC coding scheme which
is based on unstructured codes as mentioned above is not complete. 
We provide several specific examples of $3$- and $4$-description
problems and example-specific  coding schemes based on
random linear codes that perform strictly better than the above SSC coding
scheme. Subsequently, we supplement the above SSC scheme with new
coding layers which have algebraic structure. We restrict our
attention to the algebraic structure associated with finite fields. 
We present a unified coding scheme which works for arbitrary sources
and distortion measures. We characterize the asymptotic performance of
this coding scheme using computable single-letter information
quantities. We interpret the SSC coding as capturing the common
information components among $2^{l}-1$ decoders using univariate
functions, and the algebraic coding supplement as capturing common
information  among $2^{l}-1$ decoders using bivariate and multivariate
functions. 

The rest of the paper is organized as follows. Section
\ref{sec:notation} explains the notation used in the paper. Section
\ref{sec:prelim} provides an overview of the ideas developed in
previous works and provides the groundwork for the next sections. In
Section \ref{sec:Random}, we present a new unstructured coding
strategy which improves upon the CMSB scheme. We show that there are
two different types of gains compared to the previous scheme:  the
first is due to the addition of several common-component codebook
layers, the second is due to a more generalized binning method. In
Section \ref{sec:linexamples}, we identify examples where
improvements due to structured coding materialize in the MD setup. In
this section, we investigate three different examples. In two of the
examples the achievable RD region is improved via using linear
quantizers, and in the other example the gains are due to linear
binning. In Section \ref{sec:RDregion} we generalize the ideas in the
previous section and provide an achievable RD region for the general
$l$-descriptions problem. Since the characterization of RD region is
involved and complicated we provide the final RD region through
several steps, adding new coding layers in each step. Section
\ref{sec:conclusion} concludes the paper. 

 \begin{figure}[t]

\centering
\includegraphics[width=2.5in]{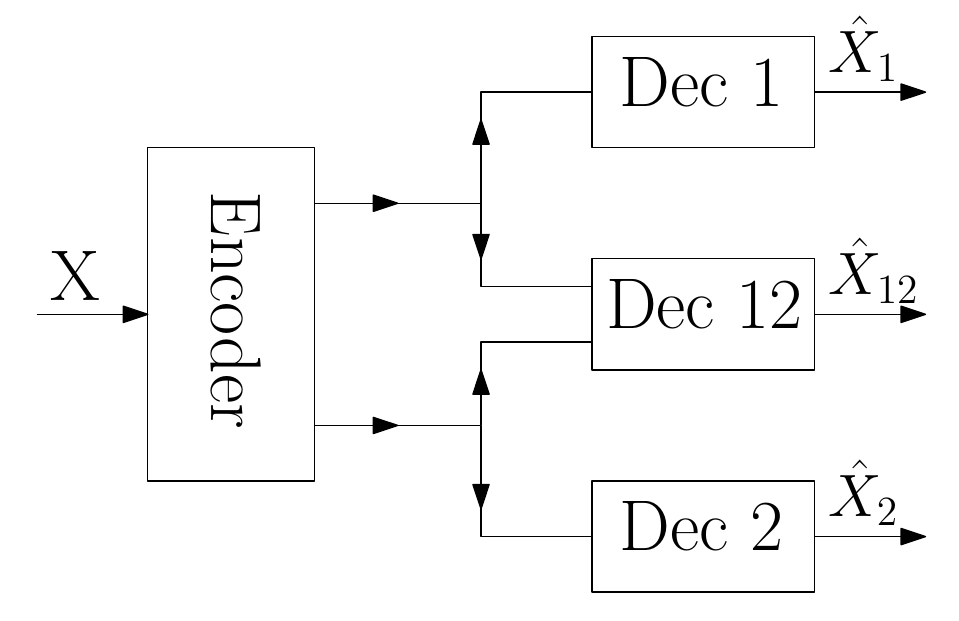}

\caption{The Two-Descriptions Setup}

\label{2Desc}
\end{figure}

\section{Definitions and Notation}
\label{sec:notation}
In this section we introduce the notation used in the paper. We restrict ourselves to finite alphabet random variables. 
We denote random variables by capital letters such as $X, U$ and their
corresponding alphabets (finite) by sans-serif typeface $\mathsf{X}$,
$\mathsf{U}$, respectively. Numbers are denoted by small letters such
as $l, k$. Sets of numbers are also denoted by the sans-serif typeface
such as $\mathsf{M}, \mathsf{N}$. Specifically,  
we denote the set of natural numbers by $\mathbb{N}$, and the field of size $q$ by $\mathbb{F}_q$. The set of numbers $\{1,2,\ldots,m\}$ is also denoted by $[1,m]$. $\alpha_\mathsf{M}$ is used to express the vector $(\alpha_1,\alpha_2,..., \alpha_m)$ where $\mathsf{M}=\{1,2,\ldots,m\}$.
 A collection whose elements are sets is called a family of sets and is denoted by the calligraphic typeface $\mathcal{M}$. For a given family of sets $\mathcal{M}$ we define a set $\widetilde{\mathcal{M}}=\bigcup_{\mathsf{M}\in \mathcal{M}}\mathsf{M}$ as the set of numbers which are the elements of the sets in $\mathcal{M}$.   The family of sets containing all subsets of $\mathsf{M}$ is denoted by $2^\mathsf{M}$. A collection whose elements are families of sets is denoted by the bold typeface $\mathbf{M}$. The collection of families of sets $\{\mathcal{A}_1,\mathcal{A}_2,\ldots \mathcal{A}_m\}$ is also represented by $\mathcal{A}_\mathsf{M}$.
 Random variables are indexed by families of sets as in $U_{\mathcal{M}}$. For the purposes of brevity we will write $U_{\mathsf{M}_1,\mathsf{M}_2,\ldots,\mathsf{M}_n}$ instead of $U_{\mathcal{M}}$ where $\mathcal{M}={\{\mathsf{M}_1,\mathsf{M}_2,\ldots,\mathsf{M}_n\}}$ wherever the notation doesn't cause ambiguity. $U_{\mathcal{M}}^n$ denotes a vector of length $n$ of random variables, each distributed according to the distribution $P_{U_{\mathcal{M}}}$. For $\epsilon>0$ and $n\in\mathbb{N}$, we denote the set of $n$-length vectors which are $\epsilon$-typical with respect to $P_{U_\mathcal{M}}$ by $A_{\epsilon}^n(U_{\mathcal{M}})$. We use the definition of frequency typicality as given in \cite{csiszarbook} in this paper.

We denote a set of random variables as follows
$U_{\mathbf{M}}=\{U_{\mathcal{M}}|\mathcal{M}\in\mathbf{M}\}$. For two
collections of families $\mathbf{M}_1$ and $\mathbf{M}_2$, we write
$[U,V]_{(\mathbf{M}_1,\mathbf{M}_2)}$ to denote the unordered
collection of random variables
$\{U_{\mathbf{M}_1},V_{\mathbf{M}_2}\}$. Let
$\mathbf{N}_i\subset\mathbf{M}_i, i=1,2$, and define
$\overline{\mathbf{N}}=(\mathbf{N}_1, \mathbf{N}_2)$. We express this
as $\overline{\mathbf{N}}\subset
({\mathbf{M}_1,\mathbf{M}_2})$. Unions, intersections and complements
are defined for $(\mathbf{M}_1, \mathbf{M}_2)$ in the same manner. 
A family of sets is called a Sperner family of sets if none of its
elements is a subset of another element. In other words a family of
sets $\mathcal{S}$ is a Sperner family if $\nexists \mathsf{N},
\mathsf{N}'\in \mathcal{S}, \mathsf{N}\subsetneq \mathsf{N}'$. For any
given set $\mathsf{M}$, the three families $\phi, \{\phi\}$ and $\{
\mathsf{M}\}$ are all Sperner families. For a set $\mathsf{M}$, we
define the collection of families of sets $\mathbf{S}_\mathsf{M}$ as
the set of all Sperner families whose elements are subsets of
$\mathsf{M}$ except for the three trivial Sperner families mentioned
above. So we have $\mathbf{S}_\mathsf{M}=\{\mathcal{S}|\nexists
\mathsf{N}, \mathsf{N}'\in \mathcal{S}, \mathsf{N}\subsetneq
\mathsf{N}'\}\backslash\{\phi, \{\phi\}, \{ \mathsf{M}\}\}$.

For the general $l$-descriptions problem, we define the set
$\mathsf{L}\triangleq  [1,l]$, and this set represents the set of all
descriptions. Each decoder receives a subset of these
descriptions. Let $l_i\in \mathsf{L}, i\in [1,n]$ for some $n$. We
denote the decoder which receives descriptions $l_1,l_2,\ldots,l_n$ by
the set $\mathsf{\mathsf{N}}=\{l_1,l_2,..,l_n\}$. Define the family of
sets $\mathcal{L}\triangleq 2^{\mathsf{L}}-\{\phi\}$. This family of
sets corresponds to the set of all possible decoders. 
We further explain the notation through an example. Consider the
three-descriptions problem. In this case we have $l=3$, the set of
descriptions are $\mathsf{L}=\{1,2,3\}$. There are seven possible
decoders. The set of all decoders is $\mathcal{L}=
\Big\{\{1\},\{2\},\{3\},\{1,2\},\{1,3\},\{2,3\},\{1,2,3\}\Big\}$. Consider
the two families of sets $\mathcal{M}_1=\Big\{\{1,2\},\{1,3\}\Big\}$
and $\mathcal{M}_2=\Big\{\{1\},\{3\},\{1,2\}\Big\}$. In this case,
$\widetilde{\mathcal{M}_i}=\{1,2,3\}, i\in \{1,2\}$. Define the set
$\mathbf{M}=\{\mathcal{M}_1,\mathcal{M}_2\}$. The set of random
variables $\{U_{\mathcal{M}_1},U_{\mathcal{M}_2}\}$ is denoted by
$U_{\mathbf{M}}=U_{\mathcal{M}_1,\mathcal{M}_2}$. Here $\mathcal{M}_1$
is a Sperner family, but $\mathcal{M}_2$ is not a Sperner family since
$\{1\},\{1,2\}\in \mathcal{M}_2$ and $\{1\}\subsetneq\{1,2\}$,
furthermore $\mathcal{M}_1\in \mathbf{S}_\mathsf{L}$ but
$\mathcal{M}_2\notin\mathbf{S}_\mathsf{L}$. 
The second part of the paper is involves application of linear codes and their cosets. The following gives a formal definition for such codes,
\begin{Definition} 
 Let $q$ be a prime number. A  $(k,n)$ linear code $\mathcal{C}$ is
 characterized by its generator matrix $G_{k\times n}$ defined on
 $\mathbb{F}_q$. $\mathcal{C}$ is defined as follows:
 $\mathcal{C}\triangleq\{\mathbf{u}G|\mathbf{u}\in \mathbb{F}_q^k\}$. 
 A  coset code $\mathcal{C}'$ is a shifted version of a linear code
 and is characterized by a generator matrix $G_{k\times n}$ and a
 dither $\mathbf{b}^n$ defined on $\mathbb{F}_q$. $\mathcal{C}'$ is
 defined as follows: 
$\mathcal{C}'\triangleq\{\mathbf{u}G+\mathbf{b}|\mathbf{u}\in \mathbb{F}_q^k\}$. 
\end{Definition}
We will make frequent use of nested linear codes. A pair of nested linear codes is defined as follows, 
\begin{Definition}
\label{def:PNLC}
For natural numbers $k_i<k_o<n$, let $G_{k_i\times n}$, and $\Delta{G}_{(k_o-k_i)\times n}$  be matrices on $\mathbb{F}_q$. Define $\mathcal{C}_i, \mathcal{C}_o$ as the linear codes generated by $G$, $[G|\Delta{G}]$, respectively.  ($\mathcal{C}_i, \mathcal{C}_o)$ is called a pair of nested linear codes with the inner code $\mathcal{C}_i$ and the outer code $\mathcal{C}_o$.  Nested coset codes are defined as shifted versions of nested linear codes. 
 \end{Definition}




\section{Preliminaries}
\label{sec:prelim}
\subsection{Problem Statement}
The general $l$-descriptions problem is described in this section. The setup is characterized by a discrete memoryless source with probability distribution $P_X(x), x\in \mathsf{X}$, where $\mathsf{X}$ is a finite set, and the distortion functions $d_{\mathsf{N}}: \mathsf{X}\times\hat{\mathsf{X}}_{\mathsf{N}}\to \mathbb{R}^+, \mathsf{N}\in \mathcal{L}$, where $\hat{\mathsf{X}}_{\mathsf{N}}$ is the reconstruction alphabet. We assume that the distortion functions are bounded, and that the distortion for the $n$-length sequence $(x^n,\hat{x}^n)$ is given by the average distortion of the components $(x_i,\hat{x}_i)$. The discrete, memoryless source $X$ is fed into an encoder.  The encoder upon receiving a block of length $n$ of source symbols produces $l$ different indices called \textit{descriptions} of the source. These descriptions are sent to the decoders. Each decoder receives a specific subset of the descriptions. Decoder $\mathsf{N}, \mathsf{N}\in \mathcal{L}$ receives 
description $i$ for all $i\in \mathsf{N}$. Based on the descriptions it has received, the decoder produces a reconstruction of the source vector.
\begin{Definition}
 An $(n,\Theta_{1},\Theta_2, \ldots, \Theta_l)$ multiple-descriptions code consist of an encoder and $|\mathcal{L}|$ decoders: 
\begin{align*}
 & e_i:\mathsf{X}^n\to [1,\Theta_i], i \in\mathsf{L},\\
 &f_\mathsf{N}:\prod_{i\in \mathsf{N}} [1,\Theta_i]\to \hat{\mathsf{X}}_\mathsf{N}^n, \mathsf{N}\in \mathcal{L}.
\end{align*}
 
\end{Definition}

  The achievable rate-distortion (RD) region is defined as follows,

\begin{Definition}
The RD vector $(R_i, D_\mathsf{N})_{i\in \mathsf{L}, \mathsf{N}\in \mathcal{L}}$ is said to be achievable if for all $\epsilon>0$ and sufficiently large $n$, there exists an $(n,\Theta_{1},\Theta_2, \ldots, \Theta_l)$ multiple-descriptions code such that the following constraints are satisfied:
\begin{enumerate}
\item{$\frac{\log \Theta_i}{n}\leq R_i+\epsilon, \forall i\in \mathsf{L}$,}
\item{$E_{X^n}\Big[d_{\mathsf{N}}\big(f_\mathsf{N}((e_i(X^n))_{i\in \mathsf{N}}), X^n\big)\Big]\leq D_\mathsf{N}+\epsilon, \forall \mathsf{N}\in \mathcal{L}$. }
\end{enumerate}

\noindent The achievable RD region for the $l-$descriptions problem is the set of all achievable RD vectors.

\end{Definition}

\begin{Remark}
 Although the reconstruction alphabet can be different from the source alphabet, throughout this paper we assume that the two alphabets are the same for the ease of notation. The results hold for the general case. 
\end{Remark}

\subsection{Prior Works}

In this section we present a brief description of some of the previous known schemes, and state the corresponding inner bounds developed for the achievable RD region. One of the early strategies for coding over two descriptions was the El Gamal - Cover (EGC) strategy \cite{EGC}. Similar to all the other strategies explained in this section, the EGC scheme relies on random, unstructured codebook generation. The following theorem describes the corresponding inner bound to the achievable RD region which results from the EGC scheme. Note that this is an alternative way to characterize the inner bound described in \cite{EGC}.

\begin{Definition}
    For a joint distribution $P$ on random variables $(U_{\{1\}},U_{\{2\}},U_{\{1,2\}}, X, Q)$ and a set of reconstruction functions $g_{\mathcal{L}}=\{g_\mathsf{N}:\mathsf{U}_{\mathsf{N}}\to\mathsf{X}, \mathsf{N}\in\mathcal{L}\}$, the set $\mathcal{RD}_{EGC}(P, g_\mathcal{L}) $ is defined as the set of RD vectors satisfying the following bounds:
\begin{align}
&R_1\geq I(U_{\{1\}};X|Q),\ \  R_2\geq I(U_{\{2\}};X|Q),\\&  R_1+R_2\geq I(U_{\{1\}},U_{\{2\}};X|Q)+I(U_{\{1\}};U_{\{2\}}|Q)+I(U_{\{1,2\}};X|U_{\{1\}},U_{\{2\}},Q), \label{RL} \\
&D_{\mathsf{N}}\geq E(d_\mathsf{N}(g_{\mathsf{N}}(U_{\mathsf{N}},Q),X)),\mathsf{N}\in \mathcal{L}.
\end{align}
\end{Definition}

\begin{Theorem}[EGC] 
The RD vector $(R_1,R_2,D_{\{1\}},D_{\{2\}},D_{\{1,2\}})$ is achievable for the two descriptions problem, if there exists a  distribution $P$ and reconstruction functions $g_\mathcal{L}$ such that $(R_1,R_2,D_{\{1\}},D_{\{2\}},D_{\{1,2\}})\in\mathcal{RD}_{EGC}(P, g_\mathcal{L})$. 
\end{Theorem}

In the EGC scheme, two codebooks $C_{\{1\}}$ and $C_{\{2\}}$  are generated independently based on the marginals  $P_{U_{\{1\}}|Q}$ and $P_{U_{\{2\}}|Q}$. The two codebooks  should be large enough so that the encoder can find a pair of jointly typical codevectors in the two codebooks. If the codebooks were generated jointly based on the joint distribution $P_{U_{\{1\}},U_{\{2\}}|Q}$,  $R_1+R_2\geq I(U_{\{1\}},U_{\{2\}};X|Q)+I(U_{\{1,2\}};X|U_{\{1\}},U_{\{2\}},Q)$ would ensure the existence of such jointly typical codevectors, however in the EGC scheme, since the codebooks are generated independently, a \textit{rate-penalty} is inflicted on the encoder. The term $I(U_{\{1\}};U_{\{2\}}|Q)$ in \eqref{RL} is a manifestation of this \textit{rate-penalty}. 
Towards reducing the \textit{rate-penalty} a new coding strategy was introduced. The resulting achievable RD region is called the Zhang-Berger [ZB] region. The region is given in the following theorem:

\begin{Definition}
   For a joint distribution $P$ on random variables $(U_{\{1\},\{2\}},U_{\{1\}},U_{\{2\}},U_{\{1,2\}}, X)$ and set of reconstruction functions $g_{\mathcal{L}}=\{g_\mathsf{N}:\mathsf{U}_{\mathsf{N}}\to\mathsf{X}, \mathsf{N}\in\mathcal{L}\}$, the set $\mathcal{RD}_{ZB}(P, g_\mathcal{L}) $ is defined as the set of RD vectors satisfying the following bounds:
\begin{align*}
&R_1\geq I(U_{\{1\},\{2\}},U_{\{1\}};X),\ \ \  R_2\geq I(U_{\{1\},\{2\}},U_{\{2\}};X),\\& R_1+R_2\geq I(U_{\{1\},\{2\}};X)+I(U_{\{1\},\{2\}},U_{\{1,2\}},U_{\{1\}},U_{\{2\}};X) +I(U_{\{1\}};U_{\{2\}}|U_{\{1\},\{2\}}),  \\
&D_\mathsf{N}\geq E(d_\mathsf{N}(g_\mathsf{N}(U_\mathsf{N}),X)),\mathsf{N}\in\mathcal{L}.
\end{align*}
 \end{Definition}

\begin{Theorem}[ZB] \label{thm:ZB} The RD vector $(R_1,R_2,D_{\{1\}},D_{\{2\}},D_{\{1,2\}})$ is achievable for the two descriptions problem, if there exists a distribution $P$ and reconstruction functions $g_\mathcal{L}$ such that $(R_1,R_2,D_{\{1\}},D_{\{2\}},D_{\{1,2\}})\in\mathcal{RD}_{ZB}(P, g_\mathcal{L})$. 
\end{Theorem}

\noindent The closure of the union of all the achievable vectors is called the ZB rate-distortion region and is denoted by $\mathcal{RD}_{ZB}$:
\begin{align*}
 \mathcal{RD}_{ZB}=cl\left(\underset{{P,g_\mathcal{L}}}\bigcup \ \mathcal{RD}_{ZB}(P, g_{\mathcal{L}})\right).
\end{align*}

The scheme differs from the EGC strategy in the introduction of the random variable $U_{\{1\},\{2\}}$.  The random variable $U_{\{1\},\{2\}}$ is called the \textit{common-component} between  the two descriptions. In the EGC scheme, in order to send $U_{\{1\}}=(\widetilde{U}_{\{1\}},U_{\{1\},\{2\}})$ and $U_{\{2\}}=(\widetilde{U}_{\{2\}},U_{\{1\},\{2\}})$, one has to pay the following \textit{rate-penalty}:
\begin{equation*}
I(U_{\{1\}};U_{\{2\}})=H(U_{\{1\},\{2\}})+I(\widetilde{U}_{\{1\}};\widetilde{U}_{\{2\}}|U_{\{1\},\{2\}}).
\end{equation*}
But in the ZB  scheme the \textit{rate-penalty} is reduced to:
\begin{equation*}
I(U_{\{1\},\{2\}};X)+I(U_{\{1\}};U_{\{2\}}|U_{\{1\},\{2\}})=I(U_{\{1\},\{2\}};X)+I(\widetilde{U}_{\{1\}};\widetilde{U}_{\{2\}}|U_{\{1\},\{2\}}).
\end{equation*}


The following definition provides a characterization of the common-component between two random variables,

\begin{Definition}
 Let $X_{\{1\}}$ and $X_{\{2\}}$ be two random variables. $W$ is called a common-component between $X_{\{1\}}$ and $X_{\{2\}}$, if there exist functions $h_i:\mathsf{X}_{\{i\}}\to \mathsf{W}, i=1,2$ such that $W=h_1(X_{\{1\}})=h_2(X_{\{2\}})$ with probability one, and the entropy of $W$ is positive.
\end{Definition}

It was shown in \cite{ZB} that in a certain two-descriptions setup, the addition of $U_{\{1\},\{2\}}$ enlarges the RD region. We call such a random variable non-redundant. The following definition gives a formal description of a non-redundant random variable:
\begin{Definition}
 In a given achievable RD region for the $l-$descriptions setup, characterized by a collection of auxiliary random variables, an auxiliary random variable $U$ is called non-redundant if the RD region strictly reduces when $U$ is set as constant. 
\end{Definition}

\begin{Example}
\label{ex:ZB}
We provide an overview of the example in \cite{ZB} where the ZB rate-distortion region is strictly better than EGC rate-distortion region, since it is used extensively in the following sections. Consider the two-descriptions setting. Here $X$ is a binary symmetric source (BSS), and the side decoders intend to reconstruct $X$ with Hamming distortion. The central decoder needs a lossless reconstruction of the source. In \cite{ZB}, it is shown that the rate distortion vector $(R_1, R_2, D_{\{1\}}, D_{\{2\}},D_{\{1,2\}})= (0.629, 0.629, 0.11, 0.11,0)$ is achievable using the ZB scheme but not the EGC scheme. 

Typically, in a given RD region, a codebook is associated with each random variable. 
We call a codebook non-redundant  if it is associated with a non-redundant random variable.  
In ZB coding scheme, the codebook corresponding to $U_{\{1\},\{2\}}$ is non-redundant.

\end{Example}
The idea of constructing a codebook carrying the common-component between the two random variables is the foundation of most of the schemes proposed for the general $l-$descriptions problem. One can even interpret the main difference between these schemes to be the way the common-component between different random variables are exploited.

As explained in the introduction, the best known achievable RD region for the $l-$descriptions problem is the CMS with binning (CMSB) strategy. In this strategy a combinatorial number of \textit{common-component} random variables are considered. We explain the coding scheme for the three-descriptions case. The codebook structure is shown in Figure \ref{fig:CMS}. There are two layers of codebooks, a layer of Maximum-Distance Separable (MDS) codes and a layer of Source Channel Erasure Codes (SCEC's). The codebook $C_{\mathcal{M}}$ is decoded at decoder $\mathsf{N}$ if $\exists \mathsf{N}'\in \mathcal{M}, \mathsf{N}'\subset \mathsf{N}$. The codebooks are binned independently, and the bin numbers for the MDS code $C_{\mathcal{M}}$ are carried by description $i $ if $i\in \underset{N\in \mathcal{M}}\bigcup N$. Whereas the bin number for each SCEC is carried by only one description i where $i\in  \underset{N\in \mathcal{M}}\bigcap N$.  Let $\mathcal{RD}_{CMSB}$ denote the resulting  
RD region achievable using CMSB strategy (see \cite{CMS,CMS2}).

\begin{figure}[!h]
\centering
\includegraphics[width=3.5in]{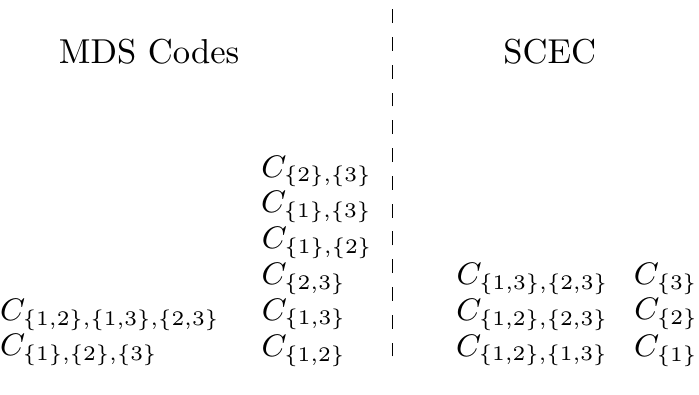}
\caption{The structure of CMSB codebooks in the three-descriptions problem}
\label{fig:CMS}
\end{figure}

\section{Improvements Using Unstructured Codes}
\label{sec:Random}
Our objective is to provide a new achievable RD region for the $l$-descriptions problem, which improves upon the RD region given by the CMSB strategy. This is based on a new coding scheme involving both unstructured and structured codes. The achievable RD region and the corresponding coding scheme is presented pedagogically in two steps. In the first step, presented in this section, we provide an RD region achievable using unstructured codes. This region is strictly better than the CMSB region. In other words this is an improvement upon the CMSB region using only unstructured codes. In the second step, presented in the next two sections, this is enhanced with a structured coding layer which improves the performance even further. In other words we show that the codebooks associated with the structured coding layer are non-redundant. 
\subsection{Main Results}


We describe the key ideas for the case $l=3$.  There are $7$ distinct decoders, one associated with every non-empty 
subset of  $\mathsf{L}=\{1,2,3\}$.  That is, we identify the set of decoders with 
$\mathcal{L}=2^{\mathsf{L}} \backslash \phi$. The new achievable RD region that we provide improves upon the 
CMSB rate-distortion region on two factors. The first comes by adding extra codebooks, and the second comes by a more general binning method. Using the common-component perspective, we associate with every non-empty subset $\mathcal{M}$ of these $7$ decoders an auxiliary random variable and a corresponding codebook. That is, we identify the collection of auxiliary variables (and their codebooks) with $2^{\mathcal{L}} \backslash \phi$.  Each codebook is binned multiple times. 
If a description is received by at least one decoder in $\mathcal{M}$, then a bin 
index of the codebook associated with $\mathcal{M}$ is sent on that description.

Although it appears that  the strategy involves the generation of a doubly-exponential number of codebooks (in $l$), we show that most of these codebooks are redundant, leaving only an asymptotically exponential number of non-redundant codebooks. While the remaining codebooks are generally non-redundant, only a small number of them are such in most of the examples we consider in this paper. 


\begin{figure}[!h]
\centering
\includegraphics[width=4in]{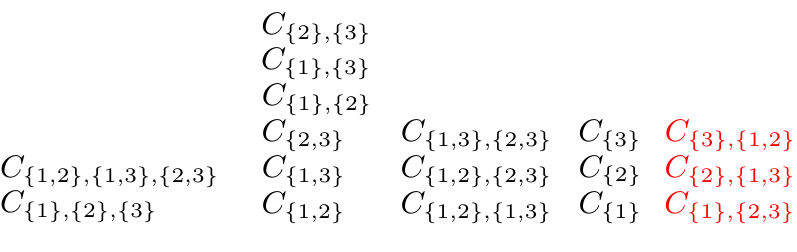}
\caption{The SSC codebooks present in the three-descriptions problem}
\label{codebooks}
\end{figure}
It turns out that a codebook is non-redundant if and only if it is associated with a a family of sets in $\mathbf{S}_{\mathsf{L}}$.
So, instead of 
$63$ codebooks, we have just $17$.  Since the  indices of the codebooks are associated with the Sperner families of sets, we call the scheme the Sperner Set Coding (SSC) scheme.
A schematic of the codebook collection is shown in Figure \ref{codebooks}.  We start from the left and from the top. The first two codebooks can be identified as $(3,2)$ and $(3,1)$ MDS codes. 
The next six codebooks can be identified as three $(2,1)$ MDS codes, and three
$(2,2)$ MDS codes associated with decoders which get two descriptions. 
 The next three can be identified as $(3,2)$ source-channel erasure codes (SCEC). 
The next three can be identified as $(3,1)$ SCEC's
(similar to the codebooks used in the EGC rate region). All these $14$ codebooks are considered in deriving the CMSB rate region. The final set of codebooks are new. They can be identified as three $(2,1)$ MDS codes associated with decoders that receive disjoint subsets of descriptions. The following theorem characterizes the achievable RD region for the SSC scheme:

\begin{Definition}
 For a joint distribution $P$ on random variables $U_{\mathcal{M}},\mathcal{M}\in \mathbf{S}_{\mathsf{L}}$ and $X$ and a set of reconstruction functions $g_{\mathcal{L}}=\{g_\mathsf{N}:\mathsf{U}_{\mathsf{N}}\to\mathsf{X}, \mathsf{N}\in\mathcal{L}\}$, the set $\mathcal{RD}_{SSC}(P, g_\mathcal{L}) $ is defined as the set of RD vectors satisfying the following bounds for some non-negative real numbers $(\rho_{\mathcal{M},i},r_{\mathcal{M}})_{i\in \widetilde{\mathcal{M}},\mathcal{M}\in \mathbf{S}_\mathsf{L} }$ :  
\begin{align}
&\label{cov11} H(U_{\mathbf{M}}|X)\geq  \!\!\!  \sum_{\mathcal{M}\in      \mathbf{M}}{  \!\!\!\!   (H(U_{{\mathcal{M}}})  \!  -   \!   r_{\mathcal{M}})},\forall \mathbf{M}\subset \mathbf{S}_\mathsf{L},
\\&\label{pack11} H(U_{\mathbf{M}_\mathsf{N}}|U_{\mathbf{L}\cup \widetilde{\mathbf{M}}_\mathsf{N}})\leq\!\!\!\!\!\! \sum_{\mathcal{M}\in \mathbf{M}_\mathsf{N}\backslash({\mathbf{L}\cup \widetilde{\mathbf{M}}_\mathsf{N}})} \!\!\!\!(H(U_\mathcal{M})+\!\!\sum_{i\in \widetilde{\mathcal{M}}}\!\rho_{\mathcal{M},i}-r_\mathcal{M}), \forall \mathbf{L}\subset \mathbf{M}_\mathsf{N}, \forall \mathsf{N}\in \mathcal{L},\\
&\nonumber r_\mathcal{M}\leq H(U_\mathcal{M}), \forall \mathcal{M}\in \mathbf{S}_\mathsf{L},\\
&R_i=\sum_{\mathcal{M}} \rho_{\mathcal{M},i},\quad D_{\mathsf{N}}=E\big\{d_{\mathsf{N}}(g_{\mathsf{N}}(U_{\mathsf{N}}),X)\big\}, \label{eq:RDSSC}
\end{align}
 where ${\mathbf{M}}_\mathsf{N}$ is the set of all codebooks decoded at decoder $\mathsf{N}$, that is ${\mathbf{M}}_\mathsf{N} \triangleq  \{\mathcal{M}\in \mathbf{S}_\mathsf{L}| \exists \mathsf{N}'\subset \mathsf{N},  \mathsf{N}'\in \mathcal{M}\}$, and $\widetilde{\mathbf{M}}_\mathsf{N}$ denotes the set of all codebooks decoded at decoders $\mathsf{N}_p\subsetneq \mathsf{N}$ which receive subsets of descriptions received by $\mathsf{N}$, that is  
$\widetilde{\mathbf{M}}_\mathsf{N} \triangleq  \bigcup_{\mathsf{N}_p \subsetneq \mathsf{N}}{\mathbf{M}}_{\mathsf{N}_p}$. 
\end{Definition}

\begin{Theorem} \label{thm:SSC}
The RD vector $(R_i,D_{\mathsf{N}})_{i\in \mathsf{L},\mathsf{N}\in \mathcal{L}}$ is achievable for the $l-$descriptions problem, if there exists a distribution $P$ and reconstruction functions $g_\mathcal{L}$ such that $(R_i,D_{\mathsf{N}})_{i\in \mathsf{L},N\in \mathcal{L}}\in\mathcal{RD}_{SSC}(P, g_\mathcal{L})$. 
\end{Theorem}

\noindent The closure of the union of all such  achievable vectors is called the 
SSC achievable rate-distortion region and is denoted by $\mathcal{RD}_{SSC}$,
\begin{align*}
 \mathcal{RD}_{SSC}= cl \left( \underset{{P,g_\mathcal{L}}}\bigcup \mathcal{RD}_{SSC}(P, g_{\mathcal{L}}) \right).
\end{align*}
In order to clarify the notation we explain the random variables decoded at each decoder in the three-descriptions problem. When $l=3$, we know $\mathbf{S}_\mathsf{L}$ has 17 elements. In the formulas,  ${\mathbf{M}}_\mathsf{N}$ corresponds to the set of random variables decoded at decoder $\mathsf{N}$,  whereas $\widetilde{{\mathbf{M}}}_\mathsf{N}$ corresponds to the set of random variables which are decodable if we have access to strict subsets of the descriptions received by $\mathsf{N}$.  Here are the random variables decoded at decoders $\{1\}$ and $\{2,3\}$: 
\begin{align*}
 &\text{decoder $\{1\}$: } U_{\{1\},\{2\},\{3\}}, U_{\{1\},\{2\}}, U_{\{1\},\{3\}}, U_{\{1\},\{2,3\}}, U_{\{1\}}\\
&\text{decoder $\{2,3\}$: } U_{\{1\},\{2\},\{3\}}, U_{\{1,2\},\{1,3\},\{2,3\}}, U_{\{1\},\{2\}}, U_{\{1\},\{3\}}, U_{\{2\},\{3\}},\\&\qquad\quad\qquad U_{\{1\},\{2,3\}}, U_{\{2\},\{1,3\}} U_{\{3\},\{1,2\}}, U_{\{1,2\},\{2,3\}}, U_{\{1,3\},\{2,3\}}, U_{\{2\}}, U_{\{3\}}, U_{\{2,3\}}\\
\end{align*}
So as an example ${\mathbf{M}}_{\{1\}}=\Bigg\{{\Big\{\{1\},\{2\},\{3\}}\Big\}, \Big\{\{1\},\{2\}\Big\}, \Big\{\{1\},\{3\}\Big\}, \Big\{\{1\},\{2,3\}\Big\}, \Big\{\{1\}\Big\}\Bigg\}$ which are all the codebooks decoded at decoder $\{1\}$.
 Also $\widetilde{\mathbf{M}}_{\{2,3\}}=\Bigg\{\Big\{\{1\},\{2\},\{3\}\Big\}$, $ \Big\{\{1\},\{2\}\Big\}$, $\Big\{\{1\},\{3\}\Big\}$, $\Big\{\{2\},\{3\}\Big\}$, $\Big\{\{2\},\{1,3\}\Big\}$, $\Big\{\{3\},\{1,2\}\Big\}$, $\Big\{\{2\}\Big\}$, $\Big\{\{3\}\Big\}\Bigg\}$,  and these are all the codebooks which are decoded at decoders $\{2\}$ and $\{3\}$.

\begin{Lemma}\label{lem:SSCconvex}
 The SSC rate-distortion region is convex. 
\end{Lemma}
\begin{proof}
See Section \ref{Ap:SSCconvex} in the appendix.
\end{proof}
\begin{Remark}
For every decoder $\mathsf{N}\in \mathcal{L}$, we have defined the reconstruction as a function of the random variable $U_{\mathsf{N}}$. However, decoder $\mathsf{N}$ decodes all random variables $U_{\mathcal{M}}$ where $\mathcal{M}\in \mathbf{M}_\mathsf{N}$. The following lemma shows that the RD region does not improve if the reconstruction function is defined as a function of $U_{\mathbf{M}_\mathsf{N}}$ instead. 
\end{Remark}
\begin{Lemma}
\label{lem:recfunc}
The RD region in Theorem \ref{thm:SSC} does not change if the reconstruction function at decoder $\mathsf{N}$ is defined as a function of $U_{\mathbf{M}_\mathsf{N}}$.
\end{Lemma}
\begin{proof}
 See Section \ref{Ap:recfunc} in the appendix.
\end{proof}
\begin{Remark}
In the scheme proposed in Theorem \ref{thm:SSC} there are $|\mathbf{S}_\mathsf{L}|$ codebooks. We know that the size of $\mathbf{S}_\mathsf{L}$ is the number of Sperner families on $\mathsf{L}$ minus three. The number of Sperner families is called the Dedekind numbers \cite{Dedekind}. There has been a large body of work in determining the values of Dedekind numbers for different $l$. It is known that these numbers grow exponentially in $l$. As an example the number of codebooks necessary for $l=2,3$ and $4$ are $3$, $17$ and $165$. However in all of the examples in this paper it turns out that many of the codebooks become redundant and only a small subset are used in the scheme.
\end{Remark}
 
\begin{proof}
\nonumber
Before proceeding to a more detailed description of the coding strategy we provide a brief outline. For each family of sets $\mathcal{M}\in \mathbf{S}_\mathsf{L}$ the encoder generates a codebook $C_{\mathcal{M}}$ based on the marginal $P_{U_{\mathcal{M}}}$ independently of the other codebooks. Intuitively, this codebook is the common-component among all the decoders $\mathsf{N}$ such that 
$\mathsf{N}\in \mathcal{M}$, and it is decoded in all decoders $N'\supset N$.
%
Codebook $C_{\mathcal{M}}$ is binned independently and uniformly for each description $i$ if $i\in\widetilde{\mathcal{M}}$. The description will carry the corresponding bin number for the codewords in each of the corresponding codebooks. 
Each decoder reconstructs its corresponding codewords by finding a unique set of jointly typical codevectors in the bins it has received. The existence of the jointly typical set of codewords is ensured at the encoder by the way of satisfaction of \eqref{cov11}, whereas at the decoder unique reconstruction is warranted by \eqref{pack11}.

\noindent\textbf{Codebook Generation:} Fix blocklength $n$ and positive reals $(\rho_{\mathcal{M},i},r_{\mathcal{M}})_{i\in \widetilde{\mathcal{M}},\mathcal{M}\in \mathbf{S}_\mathsf{L}}$. For every $\mathcal{M}\in \mathbf{S}_\mathsf{L}$, generate a codebook $C_{\mathcal{M}}$ based on the marginal $P_{U_\mathcal{M}}$ with size $2^{nr_\mathcal{M}}$. For the $i$th description, if $i\in \widetilde{M}$, bin the codebook $C_{\mathcal{M}}$ randomly and uniformly into $2^{n\rho_{\mathcal{M},i}}$ bins (i.e. randomly and uniformly assign an index $[1,2^{n\rho_{\mathcal{M},i}}]$ to each codeword in $C_{\mathcal{M}}$, and the index is called the \textit{bin-index}.). 

\noindent\textbf{Encoding:} Upon receiving the source vector $X^n$, the encoder finds a jointly-typical set of codewords $u^n_{\mathcal{M}},\mathcal{M}\in \mathbf{S}_\mathsf{L}$. Each description carries the \textit{bin-indices} of all the codewords corresponding to its own binning function. 

\noindent\textbf{Decoding:} Having received the bin-indices from descriptions $i\in \mathsf{N}$, decoder $\mathsf{N}$ tries to reconstruct $C_\mathcal{M}$ if $ \mathcal{M}\in {\mathbf{M}}_\mathsf{N}$. In other words the decoder finds a unique vector $(u_\mathcal{M}^n)_{\mathsf{N}\in \mathcal{M}}$ of jointly typical sequences in the corresponding bins. If the vector does not exist or is not unique, the decoder declares error.

\noindent\textbf{Covering Bounds:} Since codebooks are generated randomly and independently, to find a set of 
vectors $U_\mathcal{M}^n$ that is jointly typical with the source
vector $X^n$, the mutual covering bounds \eqref{cov11} are necessary
based on the mutual covering lemma 
\cite{ElGamalLec}.


\noindent\textbf{Packing Bounds:} For decoder $\mathsf{N}$,
description $i$ is received if $i\in \mathsf{N}$. Since binning is
done independently and uniformly, to find a unique set of jointly
typical sequences $(u_\mathcal{M}^n)_{\mathsf{N}\in \mathcal{M}}$, the
mutual packing bounds \eqref{pack11} are required by the mutual
packing lemma \cite{ElGamalLec}. 

\end{proof}

\begin{Remark}
 There are two main differences between the new scheme and the previous CMSB scheme. First there are additional codebooks present. As an example in Figure \ref{codebooks}, the three codebooks in the right column are not present in the CMSB scheme. Second, description $i$ bins all of the codebooks $\mathcal{M}$ such that $i\in\widetilde{M}$. We will show in the next sections that these additional codebooks contribute to an enlargement of the achievable RD region. In other words we prove that all of the additional codebooks are non-redundant. Also we show that the new binning strategy improves the achievable RD region.
\end{Remark}

\subsection{Improvements Due to additional codebooks}
Consider the general $l$-descriptions problem. In this section we prove that a codebook $C_\mathcal{M}$ is non-redundant if $\mathcal{M}\in \mathbf{S}_\mathsf{L}$.

\begin{Remark}
 It is straightforward to see that addition of a codebook $C_{\mathcal{M}}$ where $\mathcal{M}\notin \mathbf{S}_\mathsf{L}$ is not going to result in a larger achievable RD region. To see this consider the three descriptions problem and assume we add the codebook $C_{\{1\}, \{1,2\}}$. By our definition this new codebook is decoded if we either receive description 1 or both descriptions 1 and 2. In this case the codebook is decoded in exactly those decoders where 
$C_{\{1\}}$ is decoded. This means that merging these two codebooks does not change the packing bounds whereas it may relax the covering bounds. So such a codebook would be redundant. This is the reason why we consider only those codebooks which are associated with Sperner families.
\end{Remark}

\begin{Remark}
  There are three Sperner families for which we do not construct codebooks: $\{\phi, \{\phi\},\{\mathsf{L}\}\}$. It is clear that $U_{\phi}$ and $U_{\{\phi\}}$ are not necessary since they are not decoded at any decoder. Furthermore one can use the proof provided in \cite{Refinement} to show that $U_{\mathsf{L}}$ is also redundant. 
\end{Remark}

The next lemma proves that the random variables considered in Theorem \ref{thm:SSC} are non-redundant.
%
%
\begin{Lemma} 
The random variable $U_{\mathcal{M}}$  is non-redundant for every ${\mathcal{M}\in \mathbf{S}_\mathsf{L}}$. 
\label{theoremrandnonred}
\end{Lemma}
\begin{proof}
\label{redthe}

\begin{figure}[!h]
\centering
\includegraphics[width=3.5in]{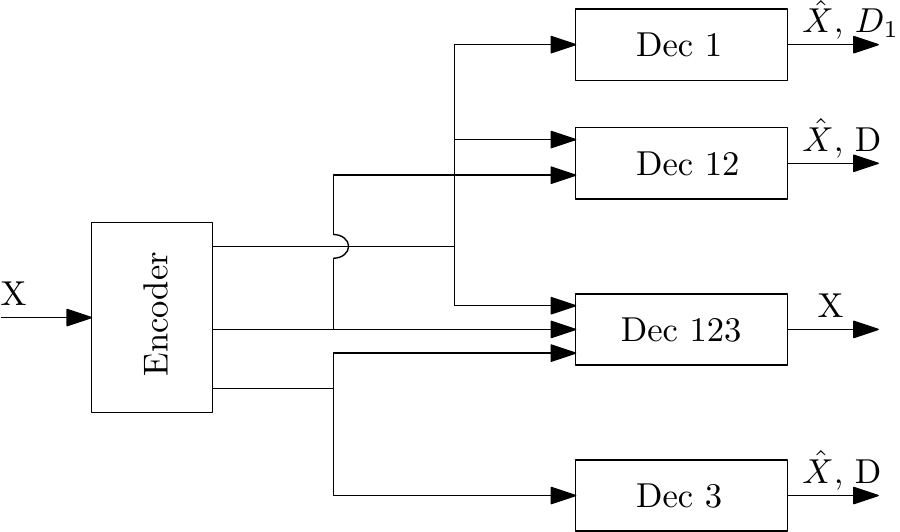}
\caption{Three Descriptions Setup Showing $C_{\{1,2\},\{3\}}$ is not redundant.}
\label{nonred}
\end{figure}

We provide the proof for the $l=3$ case and give an outline of how the proof is generalized for $l>3$. 
The codebooks $C_{\{1\}}$, $C_{\{2\}}$, $C_{\{3\}}$, $C_{\{1,2\}}$, $ C_{\{1,3\}}$, $C_{\{2,3\}}$, $C_{\{1\},\{2\}}$, $C_{\{1\},\{3\}}$, $C_{\{2\},\{3\}}$, $C_{\{1,2\},\{1,3\}}$,  $C_{\{1,2\},\{2,3\}}$, $C_{\{1,3\},\{2,3\}}$, $C_{\{1\},\{2\}, \{3\}}$, $C_{\{1,2\},\{1,3\}, \{2,3\}}$ are all present in the CMSB scheme and it was shown that they are non-redundant. The new codebooks are $C_{\{1,2\},\{3\}}$, $C_{\{1,3\},\{2\}}$ and $C_{\{23\},\{1\}}$. We prove that $C_{\{1,2\},\{3\} }$ is non-redundant using the following example, the two other codebooks are non-redundant by symmetry.

We build on Example \ref{ex:ZB} to construct a three-descriptions example as shown in Figure \ref{nonred}. 
As explained in the previous section, it is known that $U_{\{1\},\{2\}}$ is non-redundant. Let $R_i=0.629, i\in \{1,2\}$, and $D_{\{1,2\}}=0$. Let
\begin{equation}
 D^*=\min_{D} \{D|(0.629,0.629,D,D,0)\in \mathcal{RD}_{ZB}\}.
\end{equation}

Let $\mathsf{P}$ be the set of probability distributions $P_{U_{\{1\},\{2\}},U_{\{1\}},U_{\{2\}},X}$,  such that $(R_1, R_2, D_{\{1\}}, D_{\{2\}},D_{\{1,2\}})$ $= (0.629, $ $0.629, D^*, D^*,0)$  belongs to $\mathcal{RD}_{ZB}(P_{U_{\{1\},\{2\}},U_{\{1\}},U_{\{2\}},X},g_{\mathcal{L}})$ for some $g_{\mathcal{L}}$ as given in Theorem \ref{thm:ZB}.  Define the joint distribution $P^*_{U_{\{1\},\{2\}},U_{\{1\}},U_{\{2\}},X}$ as follows:
\begin{align*}
 P^*_{U_{\{1\},\{2\}},U_{\{1\}},U_{\{2\}},X}\triangleq \operatorname*{\,arginf}_{\substack{P_{U_{\{1\},\{2\}},U_{\{1\}},U_{\{2\}},X}\in \mathsf{P}}} I(U_{\{1\},\{2\}};X).
\end{align*}
Let $P^*_{U_{\{1\},\{2\}},X}$ be the marginal distribution of $U_{\{1\},\{2\}}$ and $X$. Define a random variable $W$ that is correlated with $X$ such that $P_{W,X}=P^*_{U_{\{1\},\{2\}},X}$. Let ${N}_\delta$ be a binary random variable independent of $X$ and $W$ with $P({N}_\delta=1)=\delta, \delta\in (0,0.5)$.  Define $\hat{W}=U_{\{1\},\{2\}}\wedge{} {N}_\delta$ where $\wedge{}$ denotes the logical AND function. Let $P_{\hat{W},X}, P_{X|\hat{W}}$ be the induced joint and conditional distributions, respectively. 

\begin{Example}\label{ex:binning}

We proceed by explaining the new example. The source $X$ is a BSS, decoders $\{1,2\}$ and $\{3\}$ want to reconstruct the source with respect to Hamming distortion and the central decoder wants to reconstruct the source losslessly. Decoder $\{1\}$ wants to reconstruct the source with respect to the distortion function given by: 
\begin{equation*}
d_{\{1\}}(x,\hat{x})=-\log(P_{X|\hat{W}}(x|\hat{x}))
\end{equation*}

\end{Example}


%
%

\begin{Lemma}\label{lemma1}
The following RD vector does not belong to $\mathcal{RD}_{CMSB}$, where $U_{\{1,2\},\{3\}}$ is constant.
The vector belongs to $\mathcal{RD}_{SSC}$  given in Theorem \ref{thm:SSC} which is  
achievable using  the SSC scheme:
\begin{align*}
&(R_1, R_2, R_3, D_{\{1\}}, D_{\{1,2\}}, D_{\{3\}}, D_{\{1,2,3\}})=(I(X;\hat{W}), 0.629- I(X;\hat{W}), 0.629, D\rq{}, D^*, D^*,0 ),
\end{align*}
where $D'=E(d_{\{1\}}(X,\hat{W}))$. 
\end{Lemma}
\begin{proof}
 We provide the intuition behind the proof first. In the coding scheme in Theorem \ref{thm:SSC}, the only codebooks capable of carrying the common-component between decoders $\{1,2\}$ and $\{3\}$ are $C_{\{1\},\{3\}}$, $C_{\{2\},\{3\}}$, $C_{\{1\},\{2\}, \{3\}}$ and $C_{\{1,2\},\{3\}}$. We have set the distortion constraint at decoder $\{1\}$ such that this common message can\rq{}t be carried exclusively on either of the descriptions $1$ and $2$, but rather both descriptions are necessary for the reconstruction of the common codebook. So the codebook $C_{\{1,2\},\{3\}}$ can't be empty. The proof is provided in Section \ref{Ap:lemma1} in the appendix.
\end{proof}
So far we have shown that the additional codebooks are non-redundant when $l=3$. The argument can be extended to the case when $l>3$, an outline of the general argument is provided in appendix \ref{App:addcodegen}.  
 \end{proof}

\subsection{Improvements Due to Binning}
The second factor contributing to the gains in the SSC rate-distortion region is the binning method. In the SSC scheme all descriptions $i\in \widetilde{\mathcal{M}}$ carry independent bin indices of codebook $C_{\mathcal{M}}$. This is different from the CMSB strategy where each codebook is binned by a specific subset of the descriptions based on whether the codebook is a SCEC or an MDS codebook. We prove through a three-descriptions example that the RD region enlarges due to binning in the SSC scheme, even with the three additional codebooks. We show in the following example that the bin indices of $C_{\{1,2\},\{1,3\}}$ should be carried by all descriptions. 
\begin{Example}\label{Ex:Binning}
 The example is generated by modifying Example \ref{ex:binning} and is illustrated in  Figure \ref{binningexample}. The source $X$ is BSS. $d_{\{1\}}(X,\hat{W})$ is defined as in Example \ref{ex:binning}. Decoders $\{1,2\}$ and $\{1,3\}$ want to reconstruct the source with Hamming distortion and decoder $\{1,2,3\}$ wants to reconstruct the source losslessly.
\end{Example}



\begin{Lemma}\label{theorem5}
In order to achieve $(R_1, R_2, R_3, D_{\{1\}}, D_{\{1,2\}}, D_{\{1,3\}},D_{\{1,2,3\}}) $ $=( I(\hat{W};X) , R-I(\hat{W};X),$ $R-I(\hat{W};X),$ $D', D, D,0)$ we must have $\rho_{\{1,2\},\{1,3\},2}+\rho_{\{1,2\},\{1,3\},3}>0$.
\end{Lemma}

\begin{proof}
 See Section \ref{Ap:theorem5} in the appendix.
\end{proof}

\begin{figure}[!h]
\centering
\includegraphics[width=3.5in]{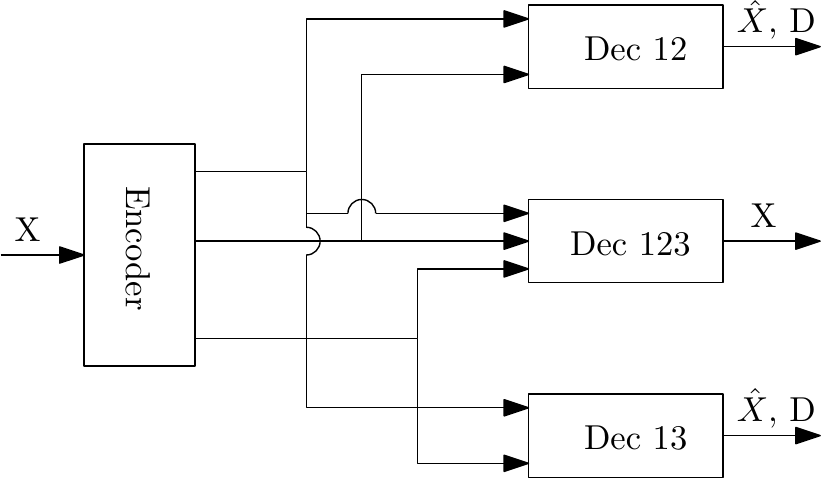}
\caption{Example Showing Improvements Due to Binning}
\label{binningexample}
\end{figure}

%

\section{Linear Coding Examples}
\label{sec:linexamples}
Before providing a unified RD region which uses both unstructured and structured codes (step 2), in this section, for pedagogical reasons, we look at three examples of $l-$descriptions problems and provide example-specific coding schemes based on linear codes that perform strictly better than the SSC scheme 
which is based on unstructured codes. This shows that the SSC region is not complete and a structured coding layer is necessary. These coding schemes are unified and presented in the next section.

\subsection{Gains Due to Linear Quantizers}
We create a three-descriptions setting where reconstructions of bivariate functions are necessary. 
\begin{Example}\label{Ex:Vecsource}
Consider the three-descriptions example in Figure \ref{3D}. Here $X$ and $Z$ are independent BSS. Decoder $\{1\}, \{2\}$ and $\{3\}$ wish to reconstruct $X$, $Z$ and $X+Z$, respectively, with Hamming distortion. Decoders $\{1,2\}, \{1,3\}$, and \{2,3\} wish to reconstruct the pair $(X,Z)$ with distortion function 
\begin{align*}
 d_{XZ}((\hat{X},\hat{Z}),(X,Z))=d_H(\hat{X},X)+d_H(\hat{Z},Z). 
\end{align*}
\end{Example}
We are interested in achieving the following RD vector:
\begin{align}
&R_i=1-h_b(\delta), i \in \{1,2,3\}, D_{\{1\}}=D_{\{2\}}=\delta,\label{RDv3}, D_{\{3\}}=\delta\ast\delta, 
D_{\{1,2\}}=D_{\{1,3\}}=D_{\{2,3\}}=2\delta.
\end{align}

First we argue that in this example, description 3 should carry a bivariate function of descriptions 1 and 2.
Decoders $\{1\}$ and $\{2\}$ operate at the optimal PtP rate-distortion function. So the corresponding descriptions have to allocate all of their rates to satisfy their individual decoder's distortion criteria. Since the distortion constraint at decoder $\{1\}$ only relates to $X$, this description only carries a quantization of $X$, and by the same argument description 2 carries a quantization of $Z$. Then description 3 has to carry the sum of these two quantizations so that the joint decoders' distortion constraints are all satisfied. Since structured codes are efficient for transmitting bivariate summations of random variables, we expect that using structured codes would give gains in this example as opposed to unstructured codes.
%
First, we prove that the RD vector is achievable using linear codes. 
\begin{figure}[!h]
\centering
\includegraphics[width=3.5in]{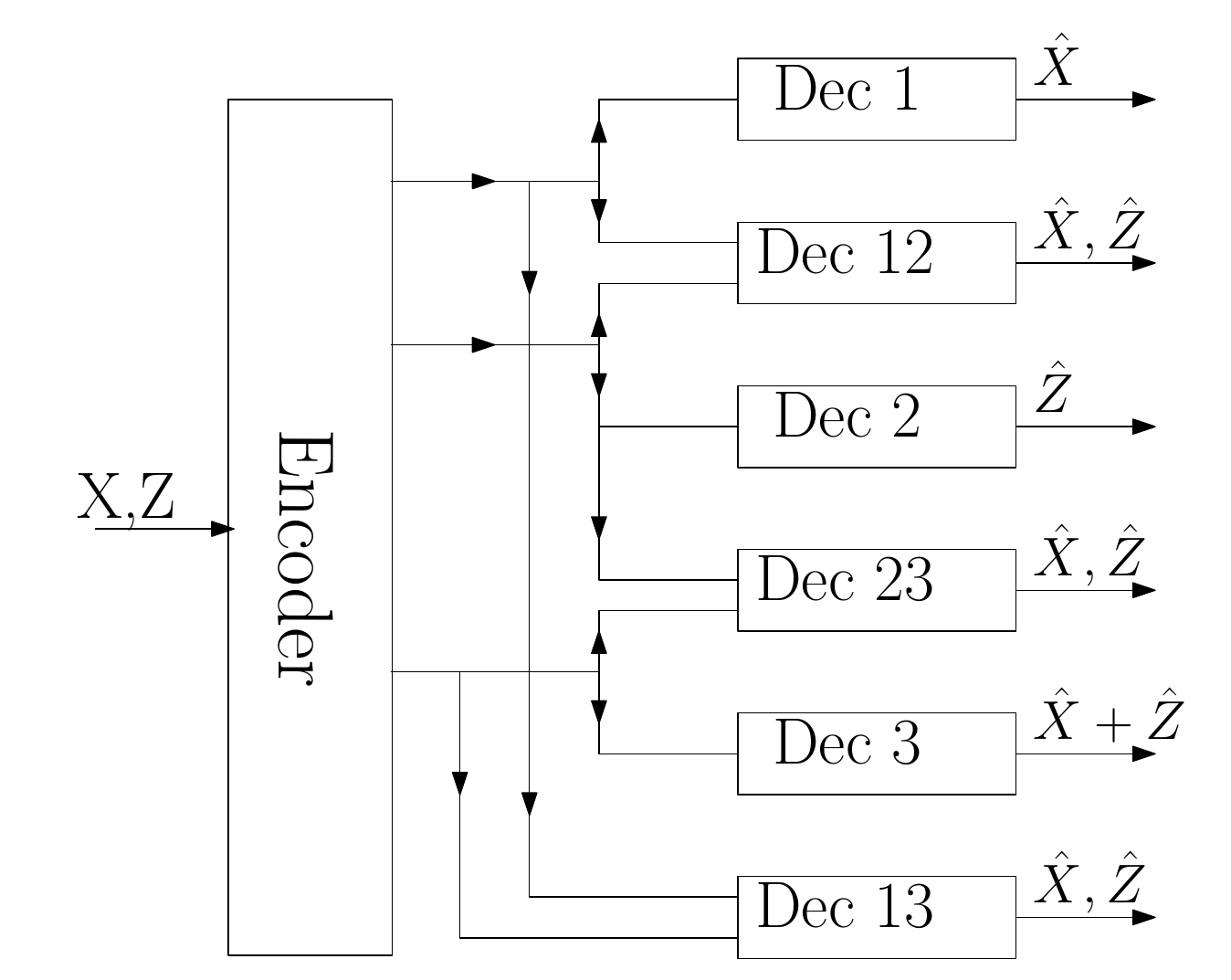}
\caption{Three-Descriptions Example with a Vector Binary Source}
\label{3D}
\end{figure}

\begin{Lemma}
 The RD vector in (\ref{RDv3}) is achievable.
\end{Lemma}
\begin{proof}
 
\textbf{Encoding:}  Construct a sequence of random linear codes $\mathcal{C}^n$ of rate $1-h_b(\delta)+\epsilon_n$, where $\epsilon_n$ is going to 0. It is well known that such a sequence of linear codes can be used to quantize a BSS to Hamming distortion $\delta$. Define the following: 
\begin{align*}
 &\hat{X}^n=argmin_{c^n\in\mathcal{C}^n}{d_H(x^n,c^n)}\\
 &\hat{Z}^n=argmin_{c^n\in\mathcal{C}^n}{d_H(z^n,c^n)}
\end{align*}
 Since $\hat{X}^n$ and $\hat{Z}^n$ are codewords and the codebook is linear, $\hat{X}^n+\hat{Z}^n$ is also a codeword. Description 1 carries the index of $\hat{X}^n$, description 2 carries the index of $\hat{Z}^n$ and description 3 carries the index of $\hat{X}^n+\hat{Z}^n$.

\noindent\textbf{Decoding:} Decoders $\{1\}$ and $\{2\}$, receive $\hat{X}^n$ and $\hat{Z}^n$, respectively, so they satisfy their distortion constraints. Decoder $\{3\}$ reconstructs $\hat{X}^n+\hat{Z}^n$. Lemma \ref{distsum} shows that the distortion criteria at this decoder is satisfied.
\begin{Lemma}\label{lemma2}
In the above setting, we have
$\frac{1}{n}E(d_H(\hat{X}^n+\hat{Z}^n,X^n+Z^n))\to \delta\ast\delta$. 
\label{distsum}
\end{Lemma}

\begin{proof}
See Section \ref{Ap:lemma2} in the appendix.

\end{proof}
Decoder $\{1,2\}$ receives $\hat{X}^n$ and $\hat{Z}^n$, so it satisfies its distortion requirements. Also decoders $\{1,3\}$ and $\{2,3\}$ can recover $\hat{X}^n$ and $\hat{Z}^n$  by adding $\hat{X}^n+\hat{Z}^n$ to $\hat{X}^n$ and $\hat{Z}^n$, respectively. \noindent This shows that the RD vector in (\ref{RDv3}) is achievable using linear codes.

\end{proof}

Next we show that the SSC scheme cannot achieve this RD vector.


\begin{Lemma}\label{theorem7}
The RD vector in (\ref{RDv3}) does not belong to $\mathcal{RD}_{SSC}$, i.e., it is not achievable using the SSC scheme. 
\end{Lemma}
\begin{proof}
 See Section \ref{Ap:theorem7} in the appendix.
  \end{proof}
\subsection{Gains Due to Linear Binning}
In the SSC scheme, there are two stages in the codebook generation phase. In the first stage unstructured codebooks are generated randomly and independently, and in the second stage these codebooks are binned randomly in an unstructured fashion for each description. In the previous example it was shown that in the first stage, it is beneficial to generate codebooks with a linear structure. However in that example there was no need for binning. In the next example, we show that the binning operation needs to be carried out in a structured manner as well. This is analogous to the gains observed in the distributed source coding problem \cite{KM} where the bin structure needs to be linear. Consider the four-descriptions example in Figure \ref{4D}. 

\begin{figure}[!h]
\centering
\includegraphics[width=3.5in]{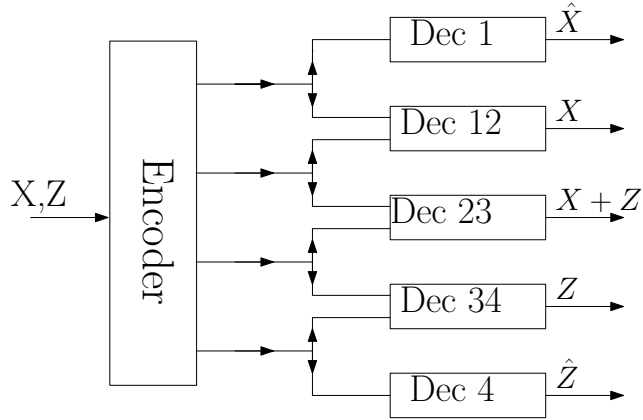}
\caption{An Example Showing the Gains Due to Linear Binning}
\label{4D}
\end{figure}
\begin{Example}\label{Ex:Vecsourcebin}
$X$ and $Z$ are BSS's. $X$ and $Z$ are not independent, and they are related to each other through a binary symmetric channel with bias $p\in (0,\frac{1}{2})$. In other words $X=Z+N_p$  where $N_p\sim Be(p)$ is independent of $X$ and $Z$. Decoders $\{1\}$ and $\{4\}$ wish to decode $X$ and $Z$, respectively, with Hamming distortion. Decoders $\{1,2\}$, $\{3,4\}$ and $\{2,3\}$ require a lossless reconstruction of $X$, $Z$ and $X+Z$, respectively. We are interested in achieving the following RD vector:
\begin{equation}
 R_1=R_4=1-h_b(p), R_2=R_3= h_b(p), D_{\{1\}}=D_{\{4\}}=p 
 \label{RD4}
\end{equation}
\end{Example}
We show that the RD vector in (\ref{RD4}) is achievable using structured codebooks and linear binning in the next lemma.
\begin{Lemma}
 The RD vector in (\ref{RD4}) is achievable. 
\end{Lemma}

\begin{proof}

 \textbf{Codebook Generation:} Take an arbitrary sequence of positive numbers $\epsilon_n$, where $\epsilon_n\to 0$ as $n\to\infty$. For any large $n\in \mathbb{N}$, fix $r_{i,n}= 1-h_b(q)-\epsilon_n$ and $r_{o,n}=1-h_b(q)+\epsilon_n$.   
 Construct a family of nested coset codes $({C}_i^n,{C}_o^n)$ where $  {C}^n_i\subset {C}^n_o$ such that the rate of the outer code is $r_{o,n}$ and the rate of the inner code is $r_{i,n}$. Choose $\mathcal{C}^n_i$ such that it is a good channel code for a BSC($p$), and choose ${C}_o^n$ such that it is a good source code for quantizing a BSS to Hamming distortion $p$. The existence of such nested coset codes is well-known from random coding arguments \cite{Galleger}. 
 Next we bin the space $\mathbb{F}_2^n$ into shifted versions (cosets) of $C_i^n$. Let $\mathcal{P}_i$ be the Voronoi region of the codeword $0^n$ in ${C}_i^n$.  Any vector $x^n\in \mathbb{F}_2^n$ can be written in the form $x^n=v^n+c_i^n, v^n\in \mathcal{P}_i, c_i^n \in {C}$. Define the $i$th bin as  $v^n+{C}_i^n$ . This operation bins the space into $|\mathcal{P}_i|= 2^{n(h_b(p)+\epsilon_n)}$ bins. The bin number associated with an arbitrary vector $x^n$ determines exactly the quantization noise resulting from quantizing the vector using $C_{i}^n$ with the minimum Hamming distortion criterion. We denote the bin number of $x^n$ as $B_i(x^n)$. A similar binning operation can be performed using ${C}_o^n$. Denote the bin number of $x^n$ obtained using shifted versions of ${C}_o^n$ by $B_o(x^n)$.
   
\noindent \textbf{Encoding:} The encoder quantizes $x^n$ and $z^n$ using $\mathcal{C}_o^n$ to $Q_o(x^n)$, and $Q_o(z^n)$, respectively. It also finds the bin number of the two source sequences $B_i(x^n)$ and $B_i(z^n)$. $Q_o(x^n)$ is transmitted on the first description, $B_i(x^n)$ is transmitted on the second description, $B_i(z^n)$ is transmitted on the third description, and $Q_o(z^n)$ is transmitted on the fourth description.
 
\noindent \textbf{Decoding:} Since the outer codes are good source codes, the distortion constraints at decoders $\{1\}$ and $\{4\}$ are satisfied. 
 
 We argue that the Voronoi region of $0^n$ in ${C}_o^n$ is a subset of the one for ${C}_i^n$. This is true since $ {C}^n_i\subset {C}^n_o$. Hence, having $B_i(x^n)$, decoders $\{1,2\}$ and $\{3,4\}$ can calculate $B_o(x^n)$. As mentioned above the bin number determines the quantization noise, so the decoders can reconstruct the source losslessly using the bin number and the quantization vector. Decoder $\{2,3\}$ receives $B_i(x^n)$ and $B_i(z^n)$. We have $x^n=Q_i(x^n)+B_i(x^n)$ and $z^n=Q_i(z^n)+B_i(z^n)$, so $B_i(x^n)+B_i(z^n)=x^n+z^n+Q_i(x^n)+Q_i(z^n)$. Since  ${C}_i^n$ is linear,  $Q_i(x^n)+Q_i(z^n)$ is a codeword, and $x^n+z^n$ can be thought of as the noise vector for a $BSC(p)$. We constructed  ${C}_i^n$  such that it is a good channel code for BSC(p), so the decoder can recover $Q(x^n)+Q(z^n)$ from $x^n+z^n+Q_i(x^n)+Q_i(z^n)$. Then by subtracting the two vectors it can get $x^n+z^n$. 
\end{proof}
Although we have used linear codes for quantization as well as binning, the linearity of the binning codebook ${C}_i^n$ is critical in this example. In fact, it can be similarly shown that one can achieve the RD vector in (\ref{RD4}) with $C_o^n$ chosen to be a union of random cosets of ${C}_i^n$. This is in contrast with the previous example where the quantizing codebook was required to be linear.  
\begin{Lemma}\label{theorem9}
 The RD vector in (\ref{RD4}) is not achievable using the SSC scheme. 
\end{Lemma}

\begin{proof}
See Section \ref{Ap:theorem9} in the appendix.
\end{proof}

\subsection{Correlated Quantizations of a Source}
	It can be noted that in the case of SSC scheme, the unstructured quantizers are generated randomly and independently. As observed in these two examples, in order to efficiently reconstruct the bivariate summation, it is beneficial to use the {same} linear code for quantizing the source. However, in the two examples the source was a vector with two components which were separately quantized using identical linear codes, and the analysis of the coding scheme required only  standard PtP covering and packing bounds for linear codes. In the more general case, evaluation of the performance of identical, and more generally, correlated linear codes for MD quantization, requires new covering and packing bounds. This is illustrated through the following scalar source example
\begin{figure}[!h]
\centering
\includegraphics[width=3.5in]{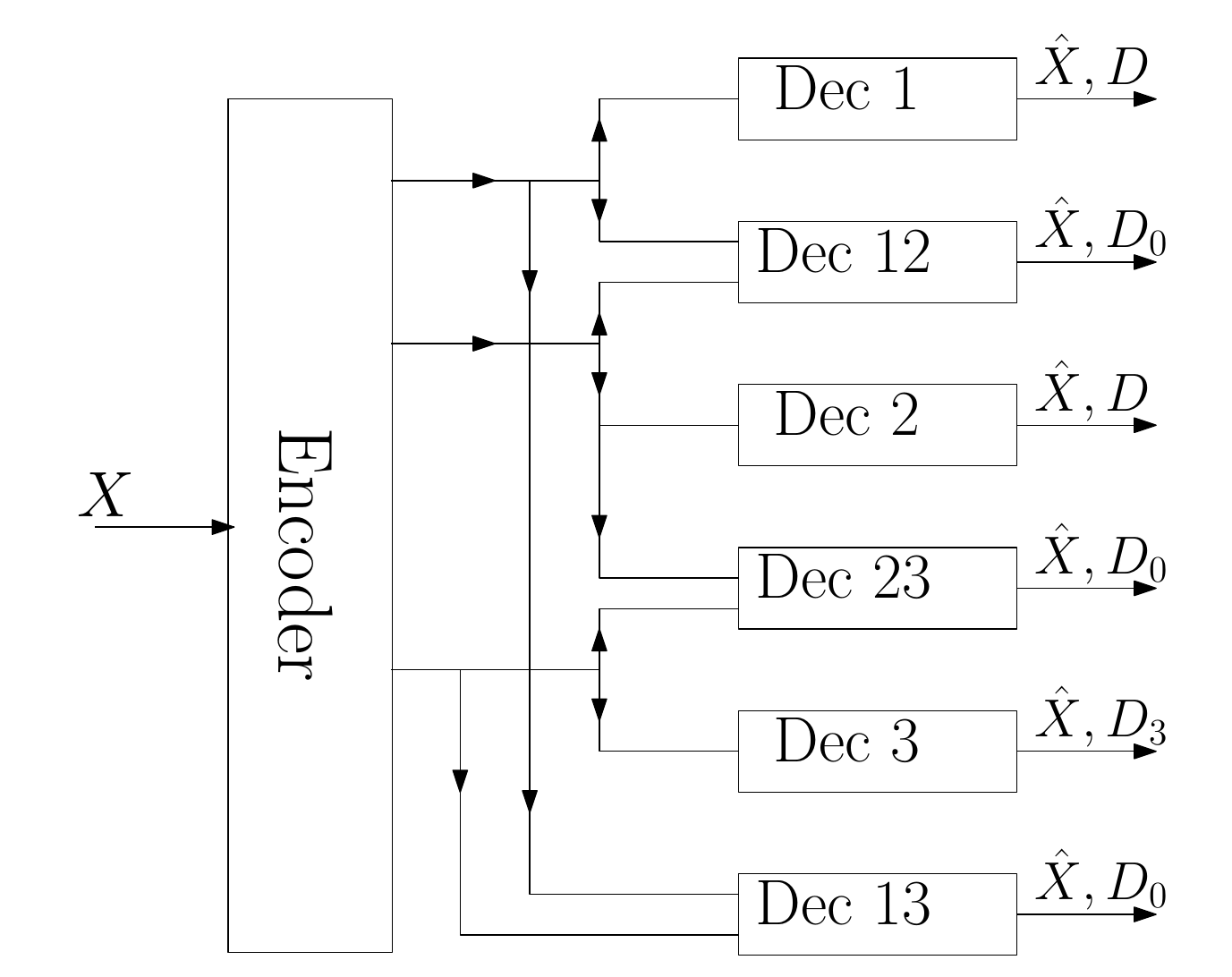}
\caption{Scalar Source  Example with Correlated Quantization}
\label{Scalarfig}
\end{figure}
which is depicted in Figure \ref{Scalarfig}. The setup is constructed based on the no-excess rate example described in \cite{ZB} for the two-descriptions problem. In the two-descriptions example, the source $X$ is BSS, 
and the distortion functions at all decoders is Hamming distortion. For the special case, called no-excess rate regime, when $R_1=R_2=\frac{1-h(D_{0})}{2}$, it was shown that the EGC region is tight. Here $D_0$ is the distortion $D_{\{1,2\}}$ at decoder $\{1,2\}$, and the minimum side distortion $D_{\{1\}}=D_{\{2\}}$ achievable was shown to be $ \frac{1}{2}(1-(1-2D_{0})(2-\sqrt{2}))$. The three-descriptions example is given as follows.

\begin{Example}\label{Ex:Scalar}
 The source $X$ is BSS, the distortion functions at decoders $\{1\}$, $\{2\}$, $\{1,2\}$, $\{1,3\}$ and $\{2,3\}$ are Hamming distortions, and the distortion function at decoder ${\{3\}}$ is the following general distortion function,  
\begin{align*}
 d_{\{3\}}(x,\hat{x}) =
  \begin{cases}
    0       & \quad \text{if } x=\hat{x}\\
    \alpha & \quad \text{if } x=0, \hat{x}=1\\
	\beta & \quad \text{if }  x=1, \hat{x}=0\\
  \end{cases}
\end{align*}
where $\alpha$ and $\beta$ are positive real numbers. 
We are interested in achieving the RD vectors with the following projections:
%
\begin{align}
\nonumber &R_1=R_2=\frac{1-h_b(D_0)}{2}, D_{\{1\}}=D_{\{2\}}= \frac{1}{2}(1-(1-2D_0)(2-\sqrt{2})),\\& D_{\{1,2\}}=D_{\{1,3\}}=D_{\{2,3\}}=D_0,\label{RDscalar}
\end{align}
\end{Example}

Our objective is to evaluate the optimal $(R_3,D_{\{3\}})$ trade-off. 
The following lemma provides the RD vectors achievable using linear codes. 

\begin{Lemma}
 The RD vector in (\ref{RDscalar}) is achievable using linear codes, as long as the following constraints are satisfied:
\begin{align}
 &R_3\geq \frac{1}{2}+h_b\left(\sqrt{2}-1\right)-h_b\left(\frac{\sqrt{2}}{2}\right)-\frac{h_b\left(D_0\right)}{2}\\
 &D_{\{3\}}\geq \alpha\left(\sqrt{2}-1\right)D_0+\beta\left(\left(\frac{3-2\sqrt{2}}{2}\right)\left(1-D_0\right)+\frac{D_0}{2}\right)\\
 & h_b\left(D_0\right)+2h_b\left(\frac{\sqrt{2}}{2}\right)+h_b\left(2\left(\sqrt{2}-1\right)D_0\right)+h_b\left(2\left(\sqrt{2}-1\right)\left(1-D_0\right)\right)\geq1. \label{eq:nonredconstraint}
\end{align}
\label{scalartheorem}
\end{Lemma}
\begin{proof}
%
Consider the following definition.
\begin{Definition}
  Let $\mathbb{F}_q$ be a field. Consider 3 random variables $X$, $U$ and $V$, where $X$ is defined on an arbitrary finite set $\mathsf{X}$, and $U$ and $V$ are defined on $\mathbb{F}_q$. Fix a PMF $P_{X,U,V}$ on $\mathsf{X}\times\mathbb{F}_q\times \mathbb{F}_q$. A sequence of code pairs $(C_1,C_2)$, where $C_j \subset \mathbb{F}_q^n$ for $j=1,2$, 
 is called $P_{XUV}$-covering if $\forall \epsilon>0$,
\begin{align*} 
 P(\{x^n|\exists (u^n,v^n)\in A_{\epsilon}^n(U,V|x^n)\cap C_1\times C_2\})\to 1 \text{ as } n\to \infty.
\end{align*}
\label{P_cov}
\end{Definition}

First, we derive new covering and packing bounds for joint quantization of a general source $X$ (i.e. not necessarily binary), using two pairs of nested coset codes.  Let $(\mathcal{C}_i,\mathcal{C}_o)$ 
and $(\mathcal{C}_i,\mathcal{C}'_o)$ be two pairs of 
nested coset codes with generator matrices $G_1$ and $G_2$ shown in Figure \ref{codebookssum} which share the inner code $\mathcal{C}_i$. 
If $r_i=0$, the two codebooks are generated independently. 
On the other hand, if $r_o=r'_o=r_i$, the two codebooks are the same, 
so this construction generalizes the previous constructions.
 
\begin{figure}[!h]
\centering
\includegraphics[height=2.1in]{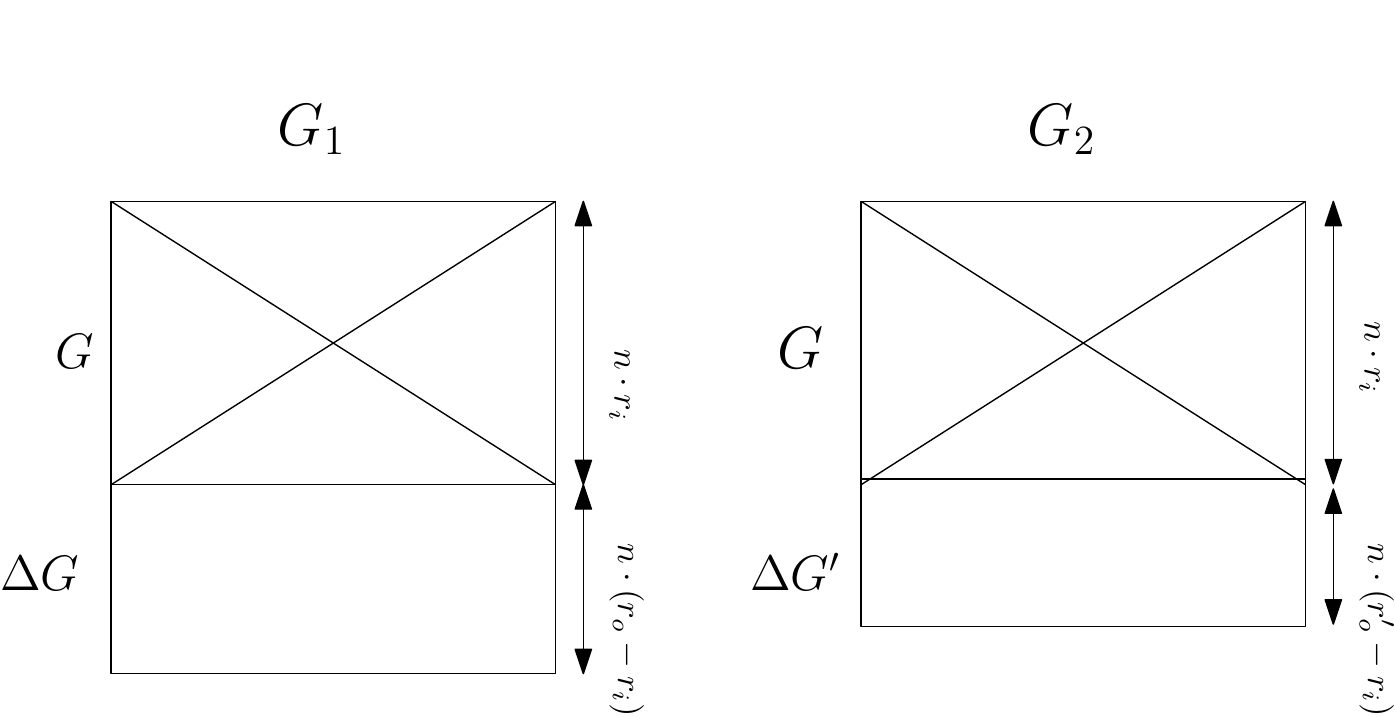}
\caption{Codebook Construction for Lemma \ref{thm: same_code}}
\label{codebookssum}
\end{figure}
\begin{Lemma}[Covering Lemma]\label{thm: same_code}
For any $P_{XUV}$ on $\mathsf{X}\times\mathbb{F}_q\times \mathbb{F}_q$ and rates $r_o,r'_o$ and $r_i$ 
satisfying (\ref{newe})-(\ref{neweq}), there exists a sequence of two pairs of nested coset codes 
$ (\mathcal{C}_o,\mathcal{C}_i)$ and $ (\mathcal{C}'_o,\mathcal{C}_i)$ 
which are $P_{XUV}$-covering.
\begin{align}
&r_o  \geq \log q - H(U|X)\label{newe}\\
&r'_o \geq \log q - H(V|X)\\
&r_o+r'_o \geq 2\log q -H(U,V|X)\\
&r_o+r'_o-r_i \geq \log{q}- H(\alpha U\oplus_q \beta V|X),\forall \alpha,\beta\in \mathbb{F}_q\backslash{\{0\}},
\label{neweq}
\end{align}

\begin{proof}
 See Section \ref{Ap:thm: same_code} in the appendix.
\end{proof}
\end{Lemma}
\begin{Remark}
The only difference between the new mutual covering bounds and the ones for independent codebook generation is the presence of the constraint (\ref{neweq}). If $r_i=0$, (\ref{neweq}) is redundant,
so we recover the mutual covering bounds for independent codebook generation as expected. If $r_i\neq 0$, (\ref{neweq}) is non-redundant. There is an intuitive explanation for this additional bound. Define $\mathcal{C}_3=\alpha\mathcal{C}_1\oplus_q \beta\mathcal{C}_2$. $\mathcal{C}_3$ is a coset code  with generator matrix $G_3=[G^t \ G'^t \ \Delta G^t]^t$, and the size of this codebook is $2^{n(r_o+r'_o-r_i)}$. Suppose there are codevetors $\mathbf{c}_u\in \mathcal{C}_1$ and $\mathbf{c}_v\in \mathcal{C}_2$ jointly typical with $\mathbf{x}$ with respect to $P_{UVX}$, then $\alpha\mathbf{c}_1\oplus_q \beta\mathbf{c}_2\in \mathcal{C}_3$ is jointly typical with $\mathbf{x}$ with respect to $P_{\alpha U\oplus_q \beta V,X}$. This implies that $\mathcal{C}_3$ should have size at least $2^{n (\log q -H(\alpha U\oplus_q \beta V|X))}$ by the converse source coding theorem.  
\end{Remark}
\begin{Definition}
   Let $\mathbb{F}_q, U, V$ and $X$ be as in Definition \ref{P_cov}. A sequence of code pairs $(C_1,C_2)$ 
and bin functions $B_i:\mathcal{C}_i\to [1,2^{n\rho_i}], i\in \{1,2\}$
   is called $P_{XUV}$-packing if for all $\epsilon>0$, 
\begin{align*} 
P \left( \left\{ x^n \left| 
\begin{array}{c}
\exists (c_u^n,c^n_v)\neq ({c'}_u^n,{c'}^n_v), \\
(c_u^n,c^n_v)\in A_{\epsilon}^n(U,V|X) \bigcap C_1 \times C_2, \\
(c'^n_u,c'^n_v)\in A_{\epsilon}^n(U,V) \bigcap C_1 \times C_2, \\
B_1(c_u^n)=B_1(c'^n_u),  \ \ B_2(c_v^n)=B_2(c'^n_v) 
\end{array} \right. \right\} \right)\to 0 \text{ as } n\to \infty.
%
%
\end{align*}
\end{Definition}

\begin{Lemma}[Packing Lemma]\label{thm: same_codepack}
For any $P_{XUV}$ on $\mathsf{X}\times\mathbb{F}_q\times \mathbb{F}_q$, there exists a sequence of two pairs of nested coset codes $(\mathcal{C}_o,\mathcal{C}_i)$ and $ (\mathcal{C}'_o,\mathcal{C}_i)$ and bin function $B_i ,i\in\{1,2\}$ which are $P_{XUV}$-packing, if $r_o, r'_o, \rho_1$ and $\rho_2$ satisfy
\begin{align}
&r_o-\rho_1  \leq \log q - H(U|V),\\
&r'_o-\rho_2 \leq \log q - H(V|U),\\
&(r_o-\rho_1)+(r'_o-\rho_2) \leq 2\log q -H(U,V).\label{sumpack}
\end{align}

\end{Lemma}
\begin{proof}
See Section \ref{Ap:thm: same_codepack} in the appendix.
\end{proof}

%

We proceed with explaining the achievability scheme. Define the joint distribution in Table \ref{tb:jointdist} on random variables $V_{\{1\}},V_{\{2\}}$ and $X$.

\setlength{\extrarowheight}{.3cm}
\begin{table}[h!]
\centering
\begin{tabular}{x{2cm}|x{2cm}|x{2cm}|x{2cm}|x{2cm}|}
\diag{.3em}{2cm}{$\qquad X$}{$V_{\{1\}},V_{\{2\}}$}& \multicolumn{1}{c}{00}&\multicolumn{1}{c}{01}&\multicolumn{1}{c}{10}&\multicolumn{1}{c}{11}\\\hline
0&$\frac{1}{2}(1-D_0)$&$\frac{\sqrt{2}-1}{2}D_0$&$\frac{\sqrt{2}-1}{2}D_0$&$\frac{3-2\sqrt{2}}{2}D_0$\\  \cline{2-5}
1&$\frac{1}{2}D_0$&$\frac{\sqrt{2}-1}{2}(1-D_0)$&$\frac{\sqrt{2}-1}{2}(1-D_0)$&$\frac{3-2\sqrt{2}}{2}(1-D_0)$\\  \cline{2-5}
\end{tabular}
\vspace{0.1in}
\caption{Joint distribution on $X$, $V_{\{1\}}$ and $V_{\{2\}}$.}
\label{tb:jointdist}
\end{table}

 \textbf{Codebook Generation:} Set $r=r_o=r'_o=r_i=1-\frac{H(V_{\{1\}},V_{\{2\}}|X)}{2}+\epsilon$, and $\rho_1=\rho_2=H(V_{\{1\}})-\frac{H(V_{\{1\}},V_{\{2\}}|X)}{2}+\epsilon$ and $\rho_3=H(V_{\{1\}}\oplus V_{\{2\}})-\frac{H(V_{\{1\}},V_{\{2\}}|X)}{2}+\epsilon$. 
 Construct a family of coset codes $\mathcal{C}$ with rate $r$. Also, construct three binning functions $B_i:\mathcal{C}^n\to [1,2^{n\rho_i}], i\in \{1,2,3\}$.

 \textbf{Encoding:} Upon receiving source sequence $x^n$, the encoder finds $\mathbf{c}^n_1$ and $\mathbf{c}^n_2$ in the codebook, such that they are jointly typical with $x^n$ with respect to $P_{V_{\{1\}},V_{\{2\}},X}$. Such a pair of codewords exists as long as the covering bounds in Lemma \ref{thm: same_code} are satisfied. In the case at hand it can be readily checked that $r_o,r'_o$ and $r_i$ satisfy the bounds.  
  Description 1 carries the bin index of $\mathbf{c}^n_1$ using $B_1$, description 2 carries the bin index of $\mathbf{c}^n_2$ using $B_2$ and description 3 carries the bin index of $\mathbf{c}^n_1+\mathbf{c}^n_2$ using $B_3$. 
 
 \textbf{Decoding:} Decoder ${\{1\}}$ receives the bin index carried by description 1, and reconstructs $\mathbf{c}_1^n$ as long as there is a unique codeword in the bin which is typical with respect to $P_{V_{\{1\}}}$. The following packing bound ensures correct decoding with arbitrarily small error:
  \begin{align*}
H(V_{\{1\}})\leq 1-\rho_1+r_i.
\end{align*}

By the same arguments decoder ${\{2\}}$ reconstructs $\mathbf{c}_2^n$ correctly. Decoder ${\{3\}}$ reconstructs $\mathbf{c}^n_1+\mathbf{c}^n_2$ with arbitrarily small error since the following packing bound is satisfied:
 \begin{align*}
&H(V_{\{1\}}+V_{\{2\}})\leq 1-\rho_3+r_i.
\end{align*}

 We conclude that all the decoders which receive two descriptions would have access to $\mathbf{c}^n_1$ and $\mathbf{c}^n_2$. Decoders $\{1\}$, $\{2\}$ and $\{3\}$ announce their decoded codewords as their reconstruction of the source. The reconstruction function at the decoders which receive two descriptions is given as follows:
\[
 \hat{x}_i= \left\{\begin{array}{cc} 0 & \quad c_{1i}=c_{2i}=0\\ 1 & \quad \text{Otherwise }  \end{array}\right.
\]

This implies that the RD vector stated in the lemma is achieved from strong typicality. 
%
\end{proof}


The following lemma shows that some of the RD vectors in Lemma \ref{scalartheorem} are not achievable using the SSC scheme. 

\begin{Lemma}\label{theorem11}
 The RD vector in (\ref{RDscalar}) is not achievable using the SSC scheme for the following values of $\alpha$ and $\beta$ and when the equality holds in (\ref{eq:nonredconstraint}):
 
\begin{align*}
 \alpha=\log_2\frac{1-2(\sqrt{2}-1)D_0}{2(2-\sqrt{2})D_0}, \ \ \beta=-\log_2\frac{1-2(\sqrt{2}-1)(1-D_0)}{2(2-\sqrt{2})(1-D_0)}
\end{align*}
 
\end{Lemma}
For example,   $D_0= 0.035, \alpha= 4.566$ and $\beta= 2.495$ satisfy the above constraints, 
where we have rounded the parameters up to the third decimal place.
\begin{proof}
See Section \ref{Ap:theorem11} in the appendix.
\end{proof}
\section{Achievable RD Region using Structured Codes}
\label{sec:RDregion}
In this section, we provide a new achievable RD region for the general $l-$descriptions problem by enhancing the SSC coding scheme with a structured coding layer. We present this region in four stages. In the first stage, we prove that the SSC region can also be achieved using structured codes. In particular, we use independent nested coset codes for each auxiliary random variable, and exploit the pairwise independence of the codewords to show the achievability of the SSC region. In the subsequent stages, we add coding layers that facilitates the reconstruction of multi-variate functions of the auxiliary random variables. The improvements due to these additional layers comes from exploiting the algebraic structure of the codebooks. In the second stage, we only allow the reconstruction of a bivariate summation of codewords. In the third stage we extend this to a multi-variate summation of the codewords. In the fourth stage, we consider the general case involving the reconstruction of an arbitrary number of multi-variate summations at the decoders.

\subsection{Stage 1: Achievability of the SSC Region Using Nested Coset Codes}

\begin{Definition}
 For a joint distribution $P$ on random variables  $U_{\mathcal{M}},\mathcal{M}\in \mathbf{S}_\mathsf{L}$ and $X$, and a set of reconstruction functions $g_{\mathcal{L}}=\{g_\mathsf{N}:\mathsf{U}_{\mathsf{N}}\to\mathsf{X}, \mathsf{N}\in\mathcal{L}\}$, the set $\mathcal{RD}_{1}(P, g_\mathcal{L}) $ is defined as the set of RD vectors satisfying the following bounds for some non-negative real numbers $(\rho_{\mathcal{M},i},r_{o,\mathcal{M}})_{i\in \widetilde{\mathcal{M}},\mathcal{M}\in \mathbf{S}_\mathsf{L} }$ :  
\begin{align}
&\label{LC1} H(U_{\mathbf{M}}|X)\geq   \sum_{\mathcal{M} \in \mathbf{M}}{(\log q  \!  -  r_{o,\mathcal{M}})}, \forall \mathbf{M}\subset \mathbf{S}_\mathsf{L}\\&\label{LC2}
H(U_{\mathbf{M}_\mathsf{N}}|U_{\mathbf{L}\cup \widetilde{\mathbf{M}}_\mathsf{N}})\leq \!\!\!\!\! \sum_{\mathcal{M}\in \mathbf{M}_\mathsf{N}\backslash{\mathbf{L}\cup \widetilde{\mathbf{M}}_\mathsf{N}}} \!\!\!\! (\log q +\!\! \sum_{j\in [1:L]}\!\! \rho_{\mathcal{M},j}\!\!-r_{o,\mathcal{M}}),\!\!\!\!\!\! \quad  \forall \mathbf{L}\subset \mathbf{M}_\mathsf{N}, \forall \mathsf{N}\in \mathcal{L}\\
&R_i=\sum_{\mathcal{M}} \rho_{\mathcal{M},i}, \qquad D_{\mathsf{N}}=E\big\{d_{\mathsf{N}}(h_{\mathsf{N}}(U_{\mathsf{N}},X))\big\}.\nonumber
\end{align}
where $r_{o,\mathcal{M}}\leq \log{q}, \forall \mathcal{M}\in \mathbf{S}_\mathsf{L}$. 
\label{linind} 
\end{Definition}

\begin{Theorem} \label{theorem12}
The RD vector $(R_i,D_{\mathsf{N}})_{i\in \mathsf{L},\mathsf{N}\in \mathcal{L}}$ is achievable for the $l-$descriptions problem using nested coset codes, if there exists a distribution $P$ and reconstruction functions $g_\mathcal{L}$ such that $(R_i,D_{\mathsf{N}})_{i\in \mathsf{L},N\in \mathcal{L}}\in\mathcal{RD}_{1}(P, g_\mathcal{L})$. 
\end{Theorem}
\begin{proof}
The encoding and decoding steps are exactly the same as the ones in the proof of Theorem \ref{thm:SSC}. The only difference is in the codebook generation phase. In this phase,  for every $\mathcal{M}\in \mathbf{S}_\mathsf{L}$, we generate a coset code $C_{\mathcal{M}}$  with rate $r_\mathcal{M}$, generator matrix $G_\mathcal{M}$, and dither $b_\mathcal{M}$. $G_\mathcal{M}$ and $b_\mathcal{M}$  are generated randomly and uniformly for every $\mathcal{M}$. The bounds in (\ref{LC1}) are the mutual covering bounds for independently generated coset codes. These bounds ensure encoding can be carried out without error. The bounds in (\ref{LC2}) are the mutual packing bounds in each decoder. They ensure errorless decoding.



\end{proof}

\begin{Lemma}\label{lemma5}
 The RD region in Theorem \ref{theorem12} is equal to the SSC RD region.
\end{Lemma}

\begin{proof}
See Section \ref{Ap:lemma5} in the appendix.
\end{proof}

\subsection{Stage 2: Reconstruction of a summation of two codebooks}
In the first stage we constructed one codebook for each subset of the
decoders. However, only the codebooks corresponding
to the Sperner families of sets are shown to be non-redundant. 
We interpret this using the  notion of common-information as
defined by  Gacs, K\"{o}rner, Witsenhausen \cite{ComInf1} \cite{ComInf2}.
Let $K(A_1;A_2)$ denote the common information between any two random
variables $A_1$ and $A_2$. The common information among $m$ random
variables $A_1,A_2,\ldots,A_m$ is a vector of length $(2^m-m-1)$ of information that is
common among every subset of $m$ random variables of size at least
two. When $m=3$, the common information is given by 
\[
[K(A_1;A_2;A_3), K(A_1;A_2), K(A_1;A_3), K(A_2;A_3)].
\]
This was referred to as univariate common information in
\cite{IC},  as each of these components are characterized using
univariate functions. 
We interpret the scheme in the first stage (SSC scheme) as capturing the common-information
components among the random variables associated  with $2^{l}-1$ decoders using univariate functions.

 For $m=3$, this notion of common information was generalized using
bivariate functions to the following seven-dimensional vector in \cite{IC}:
\[
[K(A_1;A_2;A_3), K(A_1;A_2), K(A_1;A_3), K(A_2;A_3), K(A_1;A_2,A_3),
K(A_2;A_1,A_3), K(A_3;A_1,A_2)].
\]
There are seven degrees of freedom in having information common among
$3$ random variables.
The latter three are called bivariate common information components as
they are characterized using bivariate functions of random variables.
In this sense, the addition of the structured coding layers in the next
stages can be thought of as capturing the common-information among
$2^{l}-1$ decoders using bivariate and, more generally,  multivariate
functions. 

We extend the notion of bivariate common information to $m>3$ random
variables as follows. To characterize a bivariate common information
component, we consider three subsets of $\mathsf{N}_1,\mathsf{N}_2$
and $\mathsf{N}_3$ of $\{1,2,3,\ldots,m\}$. Define
$K(\tilde{A}_{\mathsf{N}_1}; \tilde{A}_{\mathsf{N}_2}, \tilde{A}_{\mathsf{N}_3})$ as a
bivariate common information component among $A_1,A_2,\ldots,A_m$,
where $\tilde{A}_{\mathsf{N}_i}$ is the information that is common
among $A_{\mathsf{N}_i}$. For example for $m=4$, let
$\mathsf{N}_1=\{4\}$, $\mathsf{N}_2=\{1,2\}$ and
$\mathsf{N}_3=\{3\}$. This characterizes  the information in $A_4$ that can
be computed by a conference via a bivariate function of (i) the information common between $A_1$
and $A_2$, and (ii) the information in $A_3$.  This concept can be
extended to define multivariate common information among $m$ random
variables. 

We return to our discussion on the achievable RD region for the MD
problem, where $m=2^{l}-1$. In the second stage, we aim to capture the
bivariate common information among random variables associated with
$2^l-1$ decoders. In particular,  we reconstruct a summation of two
codebooks. From the above arguments, instead of one codebook for each
subset of decoders as in the first stage, in this stage we need to
construct one codebook for every triple of subsets of the
decoders. For a given triple of sets of decoders, the third set of
decoders reconstruct a bivariate summation of a random variable
corresponding to the first subset and a random variable corresponding
to the second subset of decoders.  
This is explained in more detail next. We add two new codebooks to the
SSC scheme.  The underlying random variables for these two codebooks
are denoted by $V_{\mathcal{A}_1}$ and $V_{\mathcal{A}_2}$ where
$\mathcal{A}_i\in \mathbf{S}_\mathsf{L}, i\in\{1,2\}$,
$\mathcal{A}_1\neq \mathcal{A}_2$. We construct two pairs of nested
coset codes for these two random variables. The two nested coset codes
have the same inner code. The codebook corresponding to
$V_{\mathcal{A}_i}$ is decoded at decoder $\mathsf{N}$ if
$\mathcal{A}_i\in {{\mathbf{M}}}_\mathsf{N}$, furthermore, the sum of
the two codebooks is decoded at decoder $\mathsf{N}$ if
$\mathcal{A}_3\in {\mathbf{M}}_\mathsf{N}\backslash
\{\mathcal{A}_1,\mathcal{A}_2\}$, where $\mathcal{A}_3$ is an element
of $\mathbf{S}_\mathsf{L}$. For example, let us choose
$\mathcal{A}_i=\{\{i\}\}, i\in \{1,2,3\}$.  In this case the first
codebook is decoded whenever description 1 is received, the second
codebook is decoded if description 2 is received, and the sum is
decoded whenever description 3 is received. This corresponds to the
coding schemes we presented for example \ref{Ex:Vecsource}, where
$V_{\mathcal{A}_1}=X+N_\delta$ and $V_{\mathcal{A}_2}=Z+N'_\delta$. 
The following theorem describes the achievable RD region using this scheme.

\begin{Definition}
 For any three distinct families  $\mathcal{A}_i\in {{\mathbf{S}}}_\mathsf{L},i =1,2,3$, and for a joint distribution $P$ on random variables  $U_{\mathcal{M}},\mathcal{M}\in \mathbf{S}_\mathsf{L}$, $V_{ \mathcal{A}_j}, j\in\{1,2\}$, and $X$, where the underlying alphabet for all auxiliary random variables is the field $\mathbb{F}_q$, and a set of reconstruction functions $g_{\mathcal{L}}=\{g_\mathsf{N}:\mathsf{U}_{\mathsf{N}}\to\mathsf{X}, \mathsf{N}\in\mathcal{L}\}$, the set $\mathcal{RD}_{2}(P, g_\mathcal{L}) $ is defined as the set of RD vectors satisfying the following bounds for some non-negative real numbers  $(\rho_{\mathcal{M},i},r_{o,\mathcal{M}})_{i\in \widetilde{\mathcal{M}},\mathcal{M}\in {\mathbf{S}_\mathsf{L}} }$ and $\rho_{ o,\mathcal{A}_j,i} , r'_{o, \mathcal{A}_j}, i\in \widetilde{\mathcal{A}}_i, j\in\{1,2,3\}$ and $r_{i}$: 
\begin{align}
& H(U_{\mathbf{M}}V_{\mathbf{E}}|X)\geq   \sum_{\mathcal{M} \in \mathbf{M}}{(\log q  \!  -  r_{o,\mathcal{M}})}+\sum_{\mathcal{E} \in \mathbf{E}}{(\log q  \!  -  r'_{o,\mathcal{E}})}, \forall \mathbf{M}\subset \mathbf{S}_\mathsf{L}, \mathbf{E}\subset \mathbf{A}\label{sec1cov1}\\
& H(U_\mathbf{M} , W_{\mathcal{A}_3, \alpha, \beta}|X)\geq   \sum_{\mathcal{M} \in \mathbf{M}}{(\log q  \!  -  r_{o,\mathcal{M}})}+\log{q}-r'_{o,\mathcal{A}_3}, \forall \mathbf{M}\subset \mathbf{S}_\mathsf{L}, \forall \alpha,\beta\in \mathbb{F}_q \backslash \{0\}
\label{sec1cov2}\\
& H([U,V,W]_{\overline{{\mathbf{M}}}_\mathsf{N}}|[U,V,W]_{\widehat{\mathbf{M}}_\mathsf{N}\cup \overline{\mathbf{L}}})\leq \!\!\!\!\!\!\!\!\!\sum_{\mathcal{M}\in {{\mathbf{M}}}_\mathsf{N}\backslash\widetilde{\mathbf{M}}_\mathsf{N}\cup {\mathbf{L}}}\!\!\!\!\!\!\! (\log q +\!\! \sum_{j\in \widetilde{\mathcal{M}}} \rho_{\mathcal{M},j}-r_{o,\mathcal{M}})+\!\!\!\!\!\!\!\!\!\!\!\!\sum_{\substack{\mathcal{M}\in{\mathbf{M}}_\mathsf{N}\backslash\widetilde{\mathbf{M}}_\mathsf{N}\cup \overline{\mathbf{L}}\\\bigcap\{\mathcal{A}_i|i\in[1,3]\} }}\!\!\!\!\!\!\!\!\!\!(\log{q}+\sum_{j\in \widetilde{\mathcal{M}}}\rho_{o,\mathcal{M},j}-r'_{o,\mathcal{M}}), \quad \forall \overline{\mathbf{L}}\subset {\overline{\mathbf{M}}}_\mathsf{N}
\label{sec1pack1}
\\&R_i=\sum_{\mathcal{M}} \rho_{\mathcal{M},i}, \qquad D_{\mathsf{N}}=E\big\{d_{\mathsf{N}}(h_{\mathsf{N}}(U_{\mathsf{N}},X))\big\}.\label{sec1RD}
\end{align}
where (a) $\mathbf{A}\triangleq \{\mathcal{A}_1,\mathcal{A}_2\}$, (b) $\overline{{\mathbf{M}}}_\mathsf{N}\triangleq({\mathbf{M}}_\mathsf{N}, \{ \mathcal{A}_j, j\in \{1,2\}|\mathcal{A}_j\in {\mathbf{M}}_\mathsf{N} \},\{\{\mathcal{A}_3, 1,1\}| \mathcal{A}_3\in {\mathbf{M}}_\mathsf{N}\})$, (c) $\widehat{\mathbf{M}}_\mathsf{N}\triangleq \bigcup_{\mathsf{N}'\subsetneq \mathsf{N}}{\overline{\mathbf{M}}}_{\mathsf{N}'}$, (d) $r'_{o,\mathcal{A}_3}\triangleq r'_{o,\mathcal{A}_1}+r'_{o,\mathcal{A}_2}-r_{i}
$, (e) $r_{o,\mathcal{M}}\leq \log{q}$, and (f) $W_{\mathcal{A}_3, \alpha, \beta}\triangleq\alpha V_{ \mathcal{A}_1}+\beta V_{ \mathcal{A}_2}$\footnote{We have used the script $\mathcal{A}$ to denote subscripts of random variables throughout the paper. However, the collection $\{\mathcal{A}_3, \alpha, \beta\}$ is used as the subscript for $W$ since the random variable is defined using $\alpha$ and $\beta$.}.
\end{Definition}

\begin{Theorem} \label{linoneadd}
The RD vector $(R_i,D_{\mathsf{N}})_{i\in \mathsf{L},\mathsf{N}\in \mathcal{L}}$ is achievable for the $l-$descriptions problem, if there exists a distribution $P$ and reconstruction functions $g_\mathcal{L}$ such that $(R_i,D_{\mathsf{N}})_{i\in \mathsf{L},N\in \mathcal{L}}\in\mathcal{RD}_{2}(P, g_\mathcal{L})$. 
\end{Theorem}

 Before providing the proof we explain the bounds in the new RD region. (\ref{sec1cov1}) and (\ref{sec1pack1})  are the mutual covering and packing bounds which are also present in the Theorem \ref{thm:SSC}, respectively. (\ref{sec1cov2}) is a generalization of the additional covering 
bound derived in the Lemma in \ref{thm: same_code}. 
Note that the common component among decoders $\mathsf{N}\in \mathcal{A}_1$ is the pair $(U_{\mathcal{A}_1},V_{\mathcal{A}_1})$, and similarly for $\mathcal{A}_2$. The common component among decoders $\mathsf{N}\in \mathcal{A}_3$ is the pair $(U_{\mathcal{A}_1},W_{\mathcal{A}_3,1,1})$, and observe that $W_{\mathcal{A}_3,1,1}=V_{\mathcal{A}_1}+V_{\mathcal{A}_2}$. 
\begin{proof}
Given a joint distribution $P_{\mathbf{U},\mathbf{V},X}$, and codebook and binning rates satisfying the bounds in the theorem we prove achievability of the RD vector in (\ref{sec1RD}).

\noindent\textbf{Codebook Generation:} Fix blocklength $n$. For every $\mathcal{M}\in \mathbf{S}_\mathsf{L}$, independently generate a linear code $C_{\mathcal{M}}$ with size $2^{nr_{o,\mathcal{M}}}$. Also generate two nested coset codes $C_{ \mathcal{A}_j}= (C_{i}, C_{o, \mathcal{A}_j}), j\in \{1,2\}$ where the inner code has rate $r_{i}$ and the outer codes have rates $r'_{o,\mathcal{A}_j}$. Define the set of codewords $C_{o, \mathcal{A}_3}\triangleq   C_{o, \mathcal{A}_1}+C_{o,\mathcal{A}_2}$. The size of $C_{o, \mathcal{A}_3}$ is $2^{nr'_{o, \mathcal{A}_3}}$, where $r'_{o,\mathcal{A}_3}=r'_{o,\mathcal{A}_1}+r'_{o,\mathcal{A}_2}-r_{i}
$. For the $i$th description bin the codebook $C_{\mathcal{M}}$ randomly and uniformly with rate $2^{n\rho_{\mathcal{M},i}}$ . 

\noindent\textbf{Encoding:} Upon receiving the source vector $X^n$, the encoder finds a jointly-typical set of codewords $c_{\mathcal{M}}$. Each description carries the \textit{bin-indices} of all of the corresponding codewords. The encoder declares an error if there is no jointly typical set of codewords available. 

\noindent\textbf{Decoding:} Having received the bin-indices from descriptions $i\in \mathsf{N}$, decoder $\mathsf{N}$ tries to find a set of jointly typical codewords $c_\mathcal{M}, \mathcal{M}\in \overline{\mathbf{M}}_\mathsf{N}$.  If the set of codewords is not unique, the decoder declares error.

In order for the encoder to find a set of jointly typical codewords, the mutual covering bounds (\ref{sec1cov1}) and (\ref{sec1cov2}) should hold. This is a generalization of the result in lemma \ref{thm: same_code} and we omit the proof for brevity. The bounds in (\ref{sec1pack1}) are the mutual packing bounds at each decoder.

\end{proof}
\begin{Remark} 
 Here we have considered the general case where $\mathcal{A}_i$ are chosen arbitrarily from $\mathbf{S}_\mathsf{L}$. It turns out that only certain choices of $\mathcal{A}_i$ would give non-redundant codebooks and thus provide improvements over the SSC scheme. One can show that the codebooks are redundant if $\exists N\in \mathcal{A}_1\cup\mathcal{A}_2, N'\in \mathcal{A}_3\text{ such that } N\subset N'$. For example take $\mathcal{A}_1=\{\{1\},\{3\}\}$, $\mathcal{A}_2=\{\{2\}\}$ and $\mathcal{A}_3=\{\{2,3\}\}$. 
   \end{Remark}
 \subsection{Stage 3: Reconstruction of a summation of arbitrary number of codebooks} 
In this section we reconstruct a multi-variate summation of an arbitrary number $m$ of random variables at one decoder where $m\in \mathsf{L}$ and the summation is with respect to a finite field $\mathbb{F}_q$. Following the steps in the previous section, we add $m$ new codebooks to the original SSC scheme. 
Let $\mathsf{M}\triangleq[1,m]$. The underlying random variables for these codebooks are denoted by $V_{\mathcal{A}_k}, k\in \mathsf{M}$. The random variable $V_{ \mathcal{A}_k}$ is decoded at decoder $\mathsf{N}$ if $\mathcal{A}_k\in 
{\mathbf{M}}_\mathsf{N}$. We take the families $\mathcal{A}_k, k\in\mathsf{M}$ to be distinct. The random variable $ \sum_{i\in \mathsf{M}}  V_{ \mathcal{A}_k}$ is decoded at decoder $\mathsf{N}$ if $\mathcal{A}_{m+1}\in {\mathbf{M}}_\mathsf{N}$, where $\mathcal{A}_{m+1}$ is an element of $\mathbf{S}_\mathsf{L}$. The following theorem describes the achievable RD region:

\begin{Definition}
For any $m\in\mathsf{L}$, and $m+1$ distinct families  $\mathcal{A}_i\in {{\mathbf{S}}}_\mathsf{L},i \in [1,m+1]$, and for a joint distribution $P$ on random variables   $U_{\mathcal{M}},\mathcal{M}\in \mathbf{S}_\mathsf{L}$, $V_{\mathcal{A}_k}, k\in \mathsf{M}$ and $X$, where the underlying alphabet for the auxiliary random variables is the field $\mathbb{F}_q$, and a set of reconstruction functions $g_{\mathcal{L}}=\{g_\mathsf{N}:\mathsf{U}_{\mathsf{N}}\to\mathsf{X}, \mathsf{N}\in\mathcal{L}\}$, the set $\mathcal{RD}_{3}(P, g_\mathcal{L}) $ is defined as the set of RD vectors satisfying the following bounds for some non-negative real numbers   $(\rho_{\mathcal{M},i},r_{o,\mathcal{M}})_{i\in \widetilde{\mathcal{M}},\mathcal{M}\in {\mathbf{S}_\mathsf{L}} }$ and $\rho_{o,  \mathcal{A}_k,i}, \rho_{o,\mathcal{A}_{m+1},i} , r'_{o, \mathcal{A}_k}, i\in \widetilde{A}_k, k\in [1,m+1]$ and $r_{i,\alpha_J}, J\subset \mathsf{M}$: 

\begin{align}
& H(U_{\mathbf{M}}V_{\mathbf{E}}|X)\geq   \sum_{\mathcal{M} \in \mathbf{M}}{(\log q  \!  -  r_{o,\mathcal{M}})}+\sum_{\mathcal{E} \in \mathbf{E}}{(\log q  \!  -  r'_{o,\mathcal{E}})}+, \forall \mathbf{M}\subset \mathbf{S}_\mathsf{L}, \mathbf{E}\subset \mathbf{A}, \label{sec2cov1}
\\& H(U_{\mathbf{M}}W_{\mathbf{F}}|X)\geq   \sum_{\mathcal{M} \in \mathbf{M}}{(\log q  \!  -  r_{o,\mathcal{M}})}+\sum_{\mathcal{F} \in \mathbf{F}}{(\log q  \!  -  r'_{o,\mathcal{F}})}, \forall \mathbf{M}\subset \mathbf{S}_\mathsf{L}, \mathbf{F}\subset \mathbf{B},\label{sec2cov2}\\&
 H([U,V,W]_{\overline{{\mathbf{M}}}_\mathsf{N}}|[U,V,W]_{\widehat{\mathbf{M}}_\mathsf{N}\cup \overline{\mathbf{L}}})\leq \!\!\!\!\!\!\!\!\!\sum_{\mathcal{M}\in {{\mathbf{M}}}_\mathsf{N}\backslash\widetilde{\mathbf{M}}_\mathsf{N}\cup {\mathbf{L}}}\!\!\!\!\!\!\! (\log q +\!\! \sum_{j\in \widetilde{\mathcal{M}}} \rho_{\mathcal{M},j}-r_{o,\mathcal{M}})+\!\!\!\!\!\!\!\!\!\!\!\!\sum_{\substack{\mathcal{M}\in{\mathbf{M}}_\mathsf{N}\backslash\widetilde{\mathbf{M}}_\mathsf{N}\cup \overline{\mathbf{L}}\\\bigcap\{\mathcal{A}_i|i\in[1,m+1]\} }}\!\!\!\!\!\!\!\!\!\!(\log{q}+\sum_{j\in \widetilde{\mathcal{M}}}\rho_{o,\mathcal{M},j}-r'_{o,\mathcal{M}}), \quad \forall \overline{\mathbf{L}}\subset {\overline{\mathbf{M}}}_\mathsf{N}
\label{sec2pack1}
\\&R_i=\sum_{\mathcal{M}} \rho_{\mathcal{M},i}, \qquad D_{\mathsf{N}}=E\big\{d_{\mathsf{N}}(h_{\mathsf{N}}(U_{\mathsf{N}},X))\big\}\label{sec2RD}.
\end{align}
where (a) $\mathbf{A}= \{ \mathcal{A}_k,k\in \mathsf{M}\}$, (b) $\mathbf{B}=\{(\mathcal{A}_{m+1},\alpha_\mathsf{M})|{\alpha_\mathsf{M}\in \mathbb{F}^m_q}\}$, (c) $r'_{o,\mathcal{A}_{m+1},\alpha_{\mathsf{M}}}=\sum_{k\in \mathsf{J}} r'_{o, \mathcal{A}_k}- r_{i,\alpha_\mathsf{J}}, \mathsf{J}=\{k|\alpha_k\neq 0\}$, (d)  $\sum_{\mathsf{J}':\mathsf{J}\subset {\mathsf{J}'}}{r_{i,\alpha_{\mathsf{J}'}}}\leq r_{i,\alpha_\mathsf{J}}, \forall \mathsf{J}\subset \mathsf{M}$, (e) $\overline{\mathbf{M}}_\mathsf{N}=({\mathbf{M}}_\mathsf{N},  \{ \mathcal{A}_k|\mathcal{A}_k\in {\mathbf{M}}_\mathsf{N} \}, \{(\mathcal{A}_{m+1},\alpha_{\mathsf{M}})|\mathcal{A}_{m+1}\in {\mathbf{M}}_\mathsf{N}, \alpha_i=1,i\in \mathsf{M}\})$, (f)  $\widehat{\mathbf{M}}_\mathsf{N}= \bigcup_{\mathsf{N}'\subsetneq \mathsf{N}}{\overline{\mathbf{M}}}_{\mathsf{N}'}$, (g) $r_{o,\mathcal{M}}\leq \log{q}$, and (h) $W_{\mathcal{A}_{m+1},\alpha_{\mathsf{M}}}= \sum_{i\in\mathsf{M}} \alpha_i V_{ \mathcal{A}_k}$. 
\end{Definition}

\begin{Theorem} \label{linend}
The RD vector $(R_i,D_{\mathsf{N}})_{i\in \mathsf{L},\mathsf{N}\in \mathcal{L}}$ is achievable for the $l-$descriptions problem, if there exists a distribution $P$ and reconstruction functions $g_\mathcal{L}$ such that $(R_i,D_{\mathsf{N}})_{i\in \mathsf{L},N\in \mathcal{L}}\in\mathcal{RD}_{3}(P, g_\mathcal{L})$. 
\end{Theorem}

Toward proving the theorem we need the following definition.

\begin{Definition}
 A set of $m$ coset codes $C^n_{\mathcal{A}_k}, k\in \mathsf{M}$ is called an ensemble of nested coset codes with parameter $(r_J)_{J\subset\mathsf{M}}$ if the size of the intersection $C_{\mathcal{A}_\mathsf{J}}\triangleq \bigcap_{k\in \mathsf{J}} C_{\mathcal{A}_k}$ is equal to $2^{nr_\mathsf{J}}$ for all $\mathsf{J}\subset \mathsf{M}$.
 \end{Definition}

It is straightforward to show that one can always generate an ensemble of nested coset codes $C_{ \mathcal{A}_k}, k\in \mathsf{M}$ with parameter $(r_{i,\alpha_J})_{J\subset \mathsf{M}}$ as long as $\sum_{J':J\subset {J'}}{r_{i,\alpha_{J'}}}\leq r_{i,\alpha_J}, \forall J\subset \mathsf{M}$. It is enough to choose the rows of the generator matrices of $C_{ \mathcal{A}_k}, k\in J$ such that they have ${nr_{i,\alpha_J}}$ common rows, similar to the case of Figure \ref{codebookssum}.

\begin{proof}
We provide an outline of the proof. The codebook generation for codebooks $C_{\mathcal{M}}, \mathcal{M}\in \mathbf{S}_\mathsf{L}$ is similar to the previous scheme. For random variables $V_{ \mathcal{A}_k}, k\in \mathsf{M}$ we construct an ensemble of nested coset codes $C_{ \mathcal{A}_k}, k\in \mathsf{M}$ with parameter $(r_{i,\alpha_J})_{J\subset \mathsf{M}}$. The encoder chooses a set of codewords from all the codebooks that is jointly typical with the source sequence. The following is a generalized covering lemma which shows that if (\ref{sec2cov1}) and (\ref{sec2cov2}) is satisfied such a set of codewords exists. 

\begin{Definition}
\label{P_covmut}
   Let $\mathbb{F}_q$ be a field and define $\mathcal{M}\triangleq \big\{\{1\},\{2\},\ldots,\{m\}\big\}$. Consider $m+1$ random variables $X$, $V_{\{i\}} ,i \in M$, where $X$ is defined on an arbitrary finite set $\mathsf{X}$ and $V_{\{i\}}$ are defined on $\mathbb{F}_q$. Fix a PMF $P_{X,V_{\mathcal{M}}}$ on $\mathsf{X}\times\mathbb{F}^m_q$. A sequence of m-tuples of codebooks $(C_{\{i\}}^n)_{\{i\}\in \mathcal{M}}$ is called $P_{XV_{\mathcal{M}}}$-covering if:
\begin{align*} 
 &\forall \epsilon>0, P(\{x^n|\exists v_{\mathcal{M}}^n\in A_{\epsilon}^n(V_{\mathcal{M}}|x^n)\cap \Pi_{\{i\}\in \mathcal{M}}C_{\{i\}}\})\to 1 \text{ as } n\to \infty.
\end{align*}
\end{Definition}

\begin{Lemma}[Covering Lemma]\label{thm: same_codemany}
For any $P_{X,V_{\mathcal{M}}}$ on $\mathsf{X}\times\mathbb{F}^\mathcal{M}_q$ and rates $r_{o,\{j\}}, \{j\}\in \mathcal{M}$ satisfying (\ref{sec2newe})-(\ref{sec2neweq}), there exists a sequence of ensemble of nested coset codes $\mathcal{C}^n_\mathcal{M}$  with parameter  $(r_{i,\mathsf{J}})_{\mathsf{J}\subset \mathsf{M}}$ which are $P_{X,V_\mathcal{M}}$-covering.

\begin{align}
& H(V_{\mathcal{J}}|X)\geq   \sum_{\{j\}\in \mathcal{J}}{(\log q  \!  -  r_{o,\{j\}})}, \forall\mathcal{J}\subset \mathcal{M}\label{sec2newe}\\
& H(W_{\mathcal{K}}|X)\geq   \sum_{\alpha_\mathsf{M} \in \mathcal{K}}{(\log q  \!  -  r_{o,\alpha_\mathsf{M}})}, \forall\mathcal{J}\subset \mathcal{M}, \mathcal{K}\subset \mathcal{N} \label{sec2newe2}\\
& \sum_{\mathsf{J}':\mathsf{J}\subset {\mathsf{J}'}}{r_{i,\mathsf{J}'}}\leq r_{i,\mathsf{J}}, \forall \mathsf{J}\subset \mathsf{M},\label{sec2neweq}
\end{align}
where, (a) $\mathcal{N}\triangleq \{\alpha_\mathsf{M}\in \mathbb{F}^m_q\}$, (b) $ W_{\alpha_\mathsf{M}}\triangleq \sum_{j\in \mathsf{M}} \alpha_j V_{\{j\}}$ and (c) $r_{o,\alpha_\mathsf{M}}\triangleq \sum_{j\in \mathsf{J}}r_{o,\{j\}}- r_{i,\mathsf{J}}, \mathsf{J}=\{k|\alpha_k\neq 0\}$.

\end{Lemma}

\begin{proof}

 The proof of the lemma follows the same steps as in lemma \ref{thm: same_code}. We provide the intuition behind the proof. Given that there is a set of codewords in the codebooks $\mathcal{C}_{\{j\}}, j\in \mathsf{M}$ which are jointly typical with the source sequence, for any linear combination $\mathcal{C}\triangleq \sum_{j\in \mathsf{M}}\alpha_j \mathcal{C}_{\{j\}}$ there is a codeword which is jointly typical with the random variables $X,V_{\mathcal{M}}, W_\mathcal{N}, \forall \mathcal{M}, \mathcal{N}$. From a PtP perspective, the rate of codebook $\mathcal{C}$ must satisfy (\ref{sec2newe}) and (\ref{sec2newe2}). This rate can be calculated by counting the number of rows in the generator matrix of $\mathcal{C}$ which is $nr_{o,\alpha_\mathsf{M}}$.
\end{proof}

The packing bounds at each encoder can be written in the same way as in the previous section and are given in (\ref{sec2pack1}). $\mathbf{B}$ is defined such that $W_\mathbf{B}$ is the set of all possible linear combinations of $V_{ \mathcal{A}_k}$'s. 

\end{proof}

\subsection{Stage 4: Reconstruction of an Arbitrary Number of Summations of Arbitrary Lengths}

In this section for completeness, we provide a coding scheme where we reconstruct multi-variate summations of random variables at an arbitrary number of decoders, and these summations each have arbitrary lengths. Of course, due to the large number of random variables the coding scheme becomes extremely complicated.  Let the number of the summations be $s$, and for each summation, let the length of the summation be denoted by $m_i\in L, i\in [1,s]$. Define the sets $\mathsf{S}\triangleq[1,s]$ and $\mathsf{M}_i\triangleq[1,m_i], i\in \mathsf{S}$. Following the steps in the previous sections, we add $m_i$ new codebooks for each summation. The underlying random variables for these codebooks are denoted by $V_{ \mathcal{A}_{k,i}}, k\in \mathsf{M}_i, i\in \mathsf{S}$. The random variable $V_{ \mathcal{A}_{k,i}}$ is decoded at decoder $\mathsf{N}$ if $\mathcal{A}_{k,i}\in {\mathbf{M}}_\mathsf{N}$. Fix the prime number $q_i$. The random variable  $\sum_{j\in \mathsf{M}_i}  V_{ \mathcal{A}_{k,i}}$ is decoded at decoder $\mathsf{N}$ if $\mathcal{A}_{m_i+1,i}\in {\mathbf{M}}_\mathsf{N}$, where the summation is carried out in the finite field $\mathbb{F}_{q_i}$. The following theorem describes the achievable RD region.

\begin{Definition}
 For a joint distribution $P$ on random variables    $U_{\mathcal{M}},\mathcal{M}\in \mathbf{S}_\mathsf{L}$, $V_{ \mathcal{A}_{k,i}}, k\in \mathsf{M}_i, i\in \mathsf{S}$ and $X$, where the underlying alphabet the auxiliary random variables is the field $\mathbb{F}_q$, and a set of reconstruction functions $g_{\mathcal{L}}=\{g_\mathsf{N}:\mathsf{U}_{\mathsf{N}}\to\mathsf{X}, \mathsf{N}\in\mathcal{L}\}$, the set $\mathcal{RD}_{linear}(P, g_\mathcal{L}) $ is defined as the set of RD vectors satisfying the following bounds for some non-negative real numbers   $(\rho_{\mathcal{M},i},r_{o,\mathcal{M}})_{i\in \widetilde{\mathcal{M}},\mathcal{M}\in {\mathbf{S}_\mathsf{L}} }$ and $\rho_{o, \mathcal{A}_{k,i},j_k},\rho_{o, \mathcal{A}_{m_i+1,i},j_k} , r_{o,\mathcal{A}_{k,i}}, j_k\in \widetilde{A}_{k,i}, k\in [1,m_i+1], i\in \mathsf{S}$ and $r_{i,\alpha_{J,i}}, J\subset \mathsf{M}_i, i\in \mathsf{S}$: 
\begin{align}
& H(U_{\mathbf{M}}, V_{\mathbf{E}}|X)\geq   \sum_{\mathcal{M} \in \mathbf{M}}{(\log q  \!  -  r_{o,\mathcal{M}})}+\sum_{\mathcal{E} \in \mathbf{E}}{(\log q  \!  -  r_{o,\mathcal{E}})}, \forall \mathbf{M}\subset \mathbf{S}_\mathsf{L}, \mathbf{E}\subset \mathbf{A}, 
\\& H(U_{\mathbf{M}}, W_{\mathbf{F}}|X)\geq   \sum_{\mathcal{M} \in \mathbf{M}}{(\log q  \!  -  r_{o,\mathcal{M}})}+\sum_{\mathcal{F} \in \mathbf{F}}{(\log q  \!  -  r_{o,\mathcal{F}})}, \forall \mathbf{M}\subset \mathbf{S}_\mathsf{L}, \mathbf{F}\subset \mathbf{B} ,
\\& H([U,V,W]_{\overline{{\mathbf{M}}}_\mathsf{N}}|[U,V,W]_{\widehat{\mathbf{M}}_\mathsf{N}\cup \overline{\mathbf{L}}})\leq \!\!\!\!\!\!\!\!\!\sum_{\mathcal{M}\in {{\mathbf{M}}}_\mathsf{N}\backslash\widetilde{\mathbf{M}}_\mathsf{N}\cup {\mathbf{L}}}\!\!\!\!\!\!\! (\log q +\!\! \sum_{j\in \widetilde{\mathcal{M}}} \rho_{\mathcal{M},j}-r_{o,\mathcal{M}})+\!\!\!\!\!\!\!\!\!\!\!\!\!\!\!\!\sum_{\substack{\mathcal{M}\in{\mathbf{M}}_\mathsf{N}\backslash\widetilde{\mathbf{M}}_\mathsf{N}\cup \overline{\mathbf{L}}\\\bigcap\{\mathcal{A}_{\{i\},s}|i\in[1,m_s+1], s\in\mathsf{S}\} }}\!\!\!\!\!\!\!\!\!\!\!\!\!\!(\log{q}+\sum_{j\in \widetilde{\mathcal{M}}}\rho_{o,\mathcal{M},j}-r'_{o,\mathcal{M}}), \quad \forall \overline{\mathbf{L}}\subset {\overline{\mathbf{M}}}_\mathsf{N}
\label{sec3pack1}
\\&R_i=\sum_{\mathcal{M}} \rho_{\mathcal{M},i} \label{sec3RD}, \quad D_{\mathsf{N}}=E\big\{d_{\mathsf{N}}(h_{\mathsf{N}}(U_{\mathsf{N}},X))\big\}.
\end{align}
where (a) $\mathbf{A}=\bigcup_{i\in \mathsf{S}} \{ \mathcal{A}_{k,i}|{k\in \mathsf{M}_i}\}$,  (b) $\mathbf{B}=\bigcup_{i\in\mathsf{S}}\{(\mathcal{A}_{m_i+1,i},\alpha_{\mathsf{M}_i,i})|{\alpha_{j,i}\in \mathbb{F}_{q_i}}\}$, (c) $r_{o,\mathcal{A}_{m_i+1},\alpha_{\mathsf{M}_i,i}}=\sum_{k\in \mathsf{J}} r_{o, \mathcal{A}_{k,i}}- r_{i,\alpha_{\mathsf{J}_i,i}}, \mathsf{J}_i=\{k|\alpha_{k,i}\neq 0\}, i\in \mathsf{S}$, (d) $\sum_{\mathsf{J}':\mathsf{J}\subset \mathsf{J'}}{r_{i,\alpha_{\mathsf{J}',i}}}\leq r_{i,\alpha_{\mathsf{J},i}}, \forall \mathsf{J}\subset \mathsf{M}_i, i\in \mathsf{S}$, (e) $\overline{\mathbf{M}}_\mathsf{N}=({\mathbf{M}}_\mathsf{N}, \{\mathcal{A}_{k,i}| k\in\mathsf{M}_i,i\in\mathsf{S},\mathcal{A}_{k,i}\in {\mathbf{M}}_\mathsf{N} \}, \{(\mathcal{A}_{m_i+1,i},\alpha_{\mathsf{M}_i,i})|\mathcal{A}_{m_i+1}\in {\mathbf{M}}_\mathsf{N},\alpha_{k,i}=1\})$, (f) $\widehat{\mathbf{M}}_\mathsf{N}= \bigcup_{\mathsf{N}'\subsetneq \mathsf{N}}\overline{\mathbf{M}}_{\mathsf{N}'}$, (g) $r_{o,\mathcal{M}}\leq \log{q}$ and (h) $W_{\mathcal{A}_{m+1,i},\alpha_{\mathsf{M}_i,i}}= \sum_{j=1}^{m_i} \alpha_{j,i} V_{\mathcal{A}_{k,i}}$.
\end{Definition} 

\begin{Theorem} \label{linfinal}
The RD vector $(R_i,D_{\mathsf{N}})_{i\in \mathsf{L},\mathsf{N}\in \mathcal{L}}$ is achievable for the $l-$descriptions problem, if there exists a distribution $P$ and reconstruction functions $g_\mathcal{L}$ such that $(R_i,D_{\mathsf{N}})_{i\in \mathsf{L},N\in \mathcal{L}}\in\mathcal{RD}_{linear}(P, g_\mathcal{L})$. 
\end{Theorem}

\begin{proof}
 This is a straightforward generalization of the previous step, since the proof is similar, it is omitted.
\end{proof}
\begin{Remark} 
 Similar to Theorem \ref{theoremrandnonred} one can identify the non-redundant codebooks in the above scheme. One can show that  a large number of possible codebooks become redundant in this case as well. 
 \end{Remark}

%
\section{Conclusion}
We provided several improvements over previous coding strategies for the MD problem. First, we showed that the CMSB strategy can be enhanced using additional unstructured quantizers and a new unstructured binning approach. We demonstrated these gains using examples involving binary sources and three descriptions. We provided the resulting RD region for the $l$-descriptions problem with arbitrary sources and distortion functions. Additionally, we proved that all of the new codebooks in our scheme are non-redundant for the $l$-descriptions problem. In the second part of the paper, we introduced structure in both the quantizer construction as well as in the binning functions. We showed through several examples that the improvements derived in the first part can be extended if structured quantizers and binning functions are utilized. The RD region in the first part of the paper was further improved upon by introducing additional linear coding layers. Lastly, we combined the ideas in the two parts to provide a new and strictly improved achievable RD region for the $l$-descriptions problem. 
 
\label{sec:conclusion}
\appendices
\section{Proofs for Section \ref{sec:Random}}
\subsection{Proof of Lemma \ref{lem:SSCconvex}} \label{Ap:SSCconvex}
\begin{proof}
Let $(R_{i},D_N)_{i\in \mathsf{L},N\in\mathcal{L}}\in \mathcal{RD}_{SSC}(P_{\mathbf{U},X})$ and $(R'_{i},D'_N)_{i\in \mathsf{L},N\in\mathcal{L}}\in \mathcal{RD}_{SSC}(P'_{\mathbf{U},X})$. Without loss of generality, assume $\mathsf{U}_{\{1\},\{2\},\{3\}}=\mathsf{U}'_{\{1\},\{2\},\{3\}}$. Let $\tilde{U}_{\{1\},\{2\},\{3\}}$ be defined on $\mathsf{U}_{\{1\},\{2\},\{3\}}\times \{0,1\}$. Also let $\tilde{U}_\mathcal{M}=U_{\mathcal{M}}$ if $\mathcal{M}\neq \{{\{1\},\{2\},\{3\}\lambda}\}$. For $\lambda\in[0,1]$, define a new distribution $\tilde{P}_{\tilde{\mathbf{U}},X}$ as follows:
 
 \begin{align*}
\tilde{P}_{\tilde{\mathbf{U}},X}(\tilde{\mathbf{u}},x)= \left\{ 
  \begin{array}{l l}
    \lambda{P}_{{\mathbf{U}},X}({\mathbf{u}},x) & \quad   \tilde{u}_{\{1\},\{2\},\{3\}}=(u_{\{1\},\{2\},\{3\}},0) \\
  (1-\lambda){P}'_{{\mathbf{U}},X}({\mathbf{u}},x) & \quad   \tilde{u}_{\{1\},\{2\},\{3\}}=(u_{\{1\},\{2\},\{3\}},1) \\
  \end{array} \right.
\end{align*}

Then it is straightforward to check that $\lambda(R_{i},D_N)_{i\in \mathsf{L},N\in\mathcal{L}}+(1-\lambda)(R'_{i},D'_N)_{i\in \mathsf{L},N\in\mathcal{L}}\in  \mathcal{RD}_{SSC}(\tilde{P}_{\tilde{\mathbf{U}},X})$.
\end{proof}
\subsection{Proof of lemma \ref{lem:recfunc}} \label{Ap:recfunc}
\begin{proof}
We provide an outline of the proof. Fix $\mathsf{M'}\in \mathcal{L}$. Consider a new scheme where the reconstruction function at decoder $\mathsf{M'}$ is defined as $f_\mathsf{M'}:\prod_{\mathcal{M}\in \mathbf{M}_\mathsf{M'}} U_{\mathcal{M}}\to \mathsf{X}$ with the rest of the reconstruction functions defined as in Theorem \ref{thm:SSC}. Let the RD vector $(R_i,D_{\mathsf{N}})_{i\in \mathsf{L},\mathsf{N}\in \mathcal{L}}$ be achievable in the new scheme using the distribution $P_{U_{\mathbf{S}_\mathsf{L}},X}$ and reconstruction functions $f_{\mathsf{M'}}, g_{\mathsf{M}}, \mathsf{M}\in \mathcal{L}\backslash\{\mathsf{N}\}$. We provide a new probability distribution $P_{U'_{\mathbf{S}_{\mathsf{L}}},X}$and reconstruction functions $g'_{\mathsf{M}}:U_{\mathsf{M}}\to \mathsf{X}, \mathsf{M}\in \mathcal{L}$ to shows that the RD region given in Theorem \ref{thm:SSC} contains $(R_i,D_{\mathsf{N}})_{i\in \mathsf{L},\mathsf{N}\in \mathcal{L}}$. To construct the probability distribution define $U'_{\mathcal{M}}=U_{\mathcal{M}}, \mathcal{M}\in \mathbf{S}_\mathsf{N}\backslash\{\{\mathsf{M'}\}\}$, and $U'_{\mathsf{M'}}=(U_{\mathsf{M'}},f_{\mathsf{M'}}(U_{\mathbf{M}_\mathsf{M'}})) $. As for the reconstruction functions define $g'_{\mathsf{M}}(U_{\mathsf{M}})=g_{\mathsf{M}}(U_\mathsf{M}), \mathsf{M}\in \mathcal{L}\backslash\{\mathsf{N}\}$ and $g'_\mathsf{M'}(U'_{\mathsf{M'}})=f_{\mathsf{M'}}(U_{\mathbf{M}_\mathsf{M'}})$. It is straightforward to check that with these parameters, the RD region in theorem \ref{thm:SSC} contains $(R_i,D_{\mathsf{N}})_{i\in \mathsf{L},\mathsf{N}\in \mathcal{L}}$. Intuitively, since the reconstruction functions are the same, the same distortion is achieved by both schemes.  As for the rates, in the first scheme, wherever $U_{\mathsf{M'}}$ is decoded, all of the random variables $U_{\mathbf{M}_\mathsf{M'}}$ are also decoded. So, adding a function of these random variables to $U_{\mathsf{M'}}$ does not require additional rate. 
\end{proof}
\subsection{Proof of lemma \ref{lemma1}}\label{Ap:lemma1}
\begin{proof}

Let $U_{\{1\}}=\hat{W}$, $U_{\{1,2\},\{3\}}=W$, $U_{\{1,2\}}=\hat{X}_1$, $U_{\{3\}}=\hat{X}_2$, where $\hat{X}_i$ are the reconstructions at decoder $\{i\}$ in the two user problem in Example \ref{ex:ZB}. Then it is straightforward to check that the RD vector is achievable from Theorem \ref{thm:SSC}. Next, assuming the codebook $\mathcal{C}_{\{1,2\},\{3\}}$ is empty, we consider all of the remaining 16 codebooks in the SSC scheme and show that the RD vector is not achievable. 


\noindent\textbf{Step 1:} In this step, we argue that the only non-trivial codebooks are $C_{\{1\}}$, $C_{\{3\}}$, $C_{\{1\}, \{3\}}$,$ C_{\{1,2\}}$ and $C_{\{2\}, \{3\}}$. Due to the structure of the problem, a number of the codebooks are functionally equivalent, meaning they are decoded at exactly the same decoders. So we can merge these codebooks without any loss. For example, description $\{2\}$ is only received by decoders $\{1,2\}$ and $\{1,2,3\}$, hence we can merge $C_{\{2\}}$ into $C_{\{1,2\}}$ without any loss. $C_{\{1,3\},\{23\}}$, $C_{\{1,3\}}$, $C_{\{2,3\}}$ and  $C_{\{1,2,3\}}$ are only decoded at decoder $\{1,2,3\}$ so they are redundant from the results in \cite{Refinement}. $C_{\{1\}, \{2\}}$ can be merged into $C_{\{1\}}$ since decoder $\{2\}$ is not present. $C_{\{1\}, \{2\},\{3\}}$ is equivalent to $C_{\{1\}, \{3\}}$ and can be eliminated. $C_{\{1,2\}, \{1,3\}, \{2,3\}}$, $C_{\{1,2\}, \{1,3\}}$ and $C_{\{1,2\}, \{2,3\}}$ can be merged into $C_{\{1,2\}}$. Finally $C_{\{2,3\},\{1\}}$ can be merged with $C_{\{1\}}$.  
Also $C_{\{1,2\}}$ can be merged with $C_{\{1,2,3\}}$ and is eliminated.
So we are left with four codebooks $C_{\{1\}}$, $C_{\{3\}}$, $C_{\{1\}, \{3\}}$,$ C_{\{1,2\}}$ and $C_{\{2\}, \{3\}}$.

 \noindent\textbf{Step 2:} In this step, we show that if we set $U_{\{1\},\{3\}}=\hat{W}$ and $U_{\{1\}}=\phi$ , there would be no loss in terms of RD function. The codebooks $C_{\{1\}}$ and  $C_{\{1\},\{3\}}$ are decodable using description $1$. Since decoder $\{1\}$ is at PtP optimality, these codebooks only carry $\hat{W}$. To be more precise there is a Markov chain $\left(U_{\{1\}}, U_{\{1\},\{3\}}\right)\leftrightarrow \hat{W}\leftrightarrow X$, which we prove in the following lemma.
 \begin{Lemma}
 In a PtP setup assume the decoder is at optimal PtP RD. It receives variables $U_\mathbf{M}$, and the reconstruction function is $f(U_\mathbf{M})$. Then the following Markov chain holds $U_\mathbf{M}\leftrightarrow f(U_\mathbf{M})\leftrightarrow X$. 
   \end{Lemma}   
 \begin{proof}
 \begin{align*}
& R\geq I(U_\mathbf{M};X)\stackrel{(a)}=I(f(U_\mathbf{M}),U_\mathbf{M};X)=I(f(U_\mathbf{M});X)+I(U_\mathbf{M};X|f(U_\mathbf{M}))\\
&\stackrel{(b)}\geq R+ I(U_\mathbf{M};X|f(U_\mathbf{M}))\Rightarrow I(U_\mathbf{M};X|f(U_\mathbf{M}))=0,
 \end{align*}
 where in (a) we used the fact that $f(U_\mathbf{M})$ is a function of $U_\mathbf{M}$ and in (b) we used the PtP optimality.
 \end{proof}
 Since $\hat{W}$ is decoded both at decoder $\{1\}$ and $\{3\}$, if we replace $U_{\{1\},\{3\}}$ with $(U_{\{1\},\{3\}},\hat{W})$, the decoders decode the same random variables as before, so no extra rate is required. Also, from the lemma $\left(U_{\{1\}}, U_{\{1\},\{3\}}\right)\leftrightarrow \hat{W}\leftrightarrow X$. Hence, we conclude that we can set $U_{\{1\},\{3\}}=\hat{W}$ and $U_{\{1\}}=\phi$ without any loss in terms of distortion. 

\noindent\textbf{Step 3:} Assume there are random variables $U_{\{1\},\{3\}}$ and $U_{\{2\}, \{3\}}=(W, U\rq{}_{\{2\}, \{3\}})$ such that the RD vector is achievable in the SSC scheme. From the Markov chain $\hat{W}\leftrightarrow W\leftrightarrow X$, description 1 is not used in the reconstruction in decoders $\{1,2\}$, $\{3\}$ and $\{1,2,3\}$. If we set $U_{\{1\}}=\phi$, the distortions constraint in decoders $\{1,2\}$, $\{3\}$ and $\{1,2,3\}$ are satisfied.  So we have constructed a scheme to send the descriptions at a lower rate (by setting $U_{\{1\}}=\phi$) without any loss in terms of distortion in these three decoders. This contradicts optimality of the random variables chosen for the two user scheme.

\subsection{Proof of lemma \ref{theoremrandnonred} for l>3}
\label{App:addcodegen}
We have proved that if $C_{\{12\},\{3\}}=\phi$, the RD vector is not achievable but if the constraint is lifted the scheme can achieve this RD vector, so the codebook is non-redundant. For the general $l$-descriptions problem, we provide an outline of the non-redundancy proof for $C_\mathcal{H}, \mathcal{H}\in\mathbf{S}_\mathsf{L}$. Let $\{a_{1,i},a_{2,i},\ldots,a_{n_i,i}\},i\in [1, k]$ be the elements of $\mathcal{H}$. Then to construct an example where $C_\mathcal{H}$ is non-redundant, first consider a set up where for any $i$, each set of three decoders $\{a_{1,i},a_{2,i},\ldots ,a_{n_i,i}\}$ and $\{a_{1,i+1}, a_{2,i+1}, \ldots, a_{n_{i+1},i+1}\}$ and $\{a_{1,i}, a_{2,i},\ldots,  a_{n_i,i}, a_{1,i+1}, a_{2,i}, \ldots, a_{n_{i+1},i+1}\}$ are as in the two user setup in Example \ref{ex:ZB}. Then there should be a common component between each two of the descriptions. It is straightforward to show that the common components must be the same for all of the decoders, otherwise since the codebooks are independent there would be a rate-loss as explained in the previous section. We ensure that the common component can be decoded only when all descriptions $a_{1,i}a_{2,i},\ldots a_{n_i,i}a_{1,i+1}a_{2,i}\ldots a_{n_{i+1},i+1}$ are received and not when a subset of the descriptions is received. This is done by adding decoders $\{a_{1,i}\}$,  $\{a_{1,i},a_{2,i}\}$ through $\{a_{1,i},a_{2,i},\ldots a_{n_i,i},a_{1,i+1},a_{2,i+1}, \ldots a_{n_{i+1},i+1}\}$ such that each of them would be at PtP optimality by receiving a refined version of $W$ (i.e $\{a_{1,i}\}$ would receive $\hat{W}$ and $\{a_{1,i},a_{2,i}\}$ would receive a refinement of $\hat{W}$ and so on). In  this way the only codebook that can carry $W$ without rate-loss is $C_{\mathcal{H}}$. 
 \end{proof}
 \subsection{Proof of Lemma \ref{theorem5}} \label{Ap:theorem5}
 
 \begin{proof}

Let $\rho_{\{1,2\},\{1,3\},2}+\rho_{\{1,2\},\{1,3\},3}>0$, description 1 carries $\hat{W}$ to decoder $\{1\}$ with rate $I(\hat{W};X)$. Descriptions $2$ and $3$ send $W$ to decoders $\{1,2\}$ and $\{1,3\}$ by sending a refinement on $C_{\{1,2\},\{1,3\}}$. In other words $U_{\{1\}}=\hat{W}$, $U_{\{1,2\},\{1,3\}}=W$, $U_{\{1,2\}}=\hat{X}_1$ and $U_{\{1,2\}}=\hat{X}_2$ similar to the proof of Lemma \ref{lemma1}. Then one can check that the RD vector is achievable using the SSC scheme.  Next, assume $\rho_{\{1,2\},\{1,3\},2}+\rho_{\{1,2\},\{1,3\},3}=0$, then $\rho_{\{1,2\},\{1,3\},i}=0,i\in \{2,3\}$. As in the previous section, we begin by eliminating the redundant codebooks for this communications setting.

\noindent\textbf{Step 1:} In this step we argue that only the codebooks $C_{\{1,3\}}$, $C_{\{1,2\}}$ $C_{\{1\}}$ and $C_{\{1,2\},\{1,3\}}$ are non-trivial. Due to the structure of this communications setting many of the codebooks are functionally the same and can be merged together. The codebooks $C_{\{1\}, \{2\}, \{3\}}$, $C_{\{1\}, \{3\}}$, $C_{\{1\}, \{2\}}$, $C_{\{2,3\}, \{1\}}$ are decoded at all four of the decoders and can be merged with $C_{\{1\}}$. $C_{\{1,3\}, \{2,3\}}$ can be merged with $C_{\{1,3\}}$ since decoder $\{2,3\}$ is not present, by the same argument $C_{\{1,2\}, \{1,3\}, \{2,3\}}$ is concatenated with $C_{\{1,2\}, \{1,3\}}$, also $C_{\{1,2\}, \{2,3\}}$ and $C_{\{1,2\}, \{3\}}$ are merged with $C_{\{1,2\}}$.  $C_{\{1,3\}, \{2\}}$ and $C_{\{2\}, \{3\}}$ are combined with $C_{\{1,2\}, \{2,3\}}$. Lastly since decoders $2$ and $3$ are not present, $C_{\{2\}}$ and $C_{\{3\}}$ can be merged into $C_{\{1,2\}}$ and $C_{\{1,3\}}$, respectively. So only the four codebooks $C_{\{1,3\}}$, $C_{\{1,2\}}$ $C_{\{1\}}$ and $C_{\{1,2\},\{1,3\}}$ remain.
 
 \noindent\textbf{Step 2:} By the same arguments as in step 2 of Lemma \ref{lemma1}, we can set $U_{\{1\}}=\hat{W}$. 

\noindent \textbf{Step 3:} By assumption, the codebook $C_{\{1,2\},\{1,3\}}$ is only carried by the first description. However,  the codebook is not decoded at decoder $\{1\}$. Since the decoder is at PtP optimality, $C_{\{1,2\},\{1,3\}}$ can't be sent through the first description either (i.e $\rho_{\{1,2\},\{1,3\},1}=0$ and  $C_{\{1,2\},\{1,3\}}$ can be eliminated.). 

\noindent\textbf{Step 4:} After Fourier-Motzkin elimination, the covering and packing bounds for the remaining three codebooks give the following inequality,
 
\begin{align}
&R_1+R_2+R_3\geq I(U_{\{1,2\}},U_{\{1,3\}};X|\hat{W})+I(U_{\{1,3\}};U_{\{1,2\}}|\hat{W})+I(\hat{W};X)
\label{b2bin}
\end{align}
By the definition of $\hat{W}$ we have $\hat{W}\leftrightarrow W\leftrightarrow X$ and $I(\hat{W};X)<I({W};X)$, so the bound above is strictly larger than the case when $\hat{W}$ is replaced by $W$ (i.e. when $U_{\{1,2\},\{1,3\}}=W$.). This concludes the proof.

\end{proof}
 
 


%
%
%
\section{Proofs for Section \ref{sec:linexamples}}
\subsection{Proof of Lemma \ref{lemma2}}\label{Ap:lemma2}
\begin{proof}

\begin{align*}
 \frac{1}{n}E(d_H(\hat{X}^n\oplus_2\hat{Z}^n,X^n\oplus_2Z^n))&= \frac{1}{n}E(w_H(\hat{X}^n\oplus_2\hat{Z}^n\oplus_2 X^n\oplus_2Z^n))
 \\&= \frac{1}{n}E(w_H(X^n\oplus_2\hat{X}^n\oplus_2 Z^n\oplus_2\hat{Z}^n))
 \\&= \frac{1}{n}E(d_H(X^n\oplus_2\hat{X}^n,Z^n\oplus_2\hat{Z}^n)).
\end{align*}
 Note that $X^n\oplus_2\hat{X}^n$ is the quantization noise of quantizing $X^n$ and $Z^n\oplus_2\hat{Z}^n$ is the quantization noise of quantizing $Z^n$. Since the source vectors are independent, the noise vectors are also independent and the summation converges to $\delta\ast\delta$ (The arguments are similar to the ones given in \cite{FinLen}.).

\end{proof}
\subsection{Proof of Lemma \ref{theorem7}}\label{Ap:theorem7}

\begin{proof}
 We assume that there exists a probability distribution $P$ on $X$ and $U_{\mathbf{S}_\mathsf{L}}$ for which the RD vector is achievable using the SSC scheme and arrive at a contradiction. Since all of the decoders are present in this setup, we need to consider the SSC with all the codebooks present, so the proof is more involved than the proofs in the previous section.

\noindent\textbf{Step 1:} In this step we show that description $i$, where $i=1,2$, does not carry any bin indices for codewords from codebook $C_{\mathcal{M}}$ if $ \mathcal{M}\notin \mathbf{M}_{\{i\}}$.  Descriptions 1 and 2 only carry indices which are used in the reconstruction at decoders $\{1\}$ and $\{2\}$, respectively. This is true since these two decoders are receiving information at optimal PtP rate-distortion. Note that this does not mean the corresponding codebooks are empty, we can only conclude that no bin indices for the codewords are sent through these descriptions. For example if $\mathcal{M}=\{\{2\},\{1,3\}\}$ and $i=1$, then $\rho_{\mathcal{M},i}=0$.

\begin{Lemma} 
For $i\in \{1,2\}$, and ${\mathcal{M}}$ such that $\{i\}\notin {\mathcal{M}}$, $\rho_{{\mathcal{M}},i}=0$.

\end{Lemma}
\begin{proof}
 
 From optimality at decoder $\{1\}$ we have the following equality:
\begin{equation}
R_i= I(U_{\mathbf{M}_{\{i\}}};X,Z)
\label{optimality}
\end{equation}
 
Consider the following covering bound on the random variables $U_{{\mathbf{M}}_{\{i\}}}$:
\begin{align}
H(U_{{\mathbf{M}}_{\{i\}}}|X,Z)\geq  & \sum_{{\mathcal{M}}\in {\mathbf{M}}_{\{i\}}}{   (H(U_{{\mathcal{M}}})  \!  -   \!   r_{{\mathcal{M}}})},
\label{coveringoptimality}
\end{align}
also we have the following packing bound at decoder $\{i\}$:
\begin{align}
H(U_{{\mathbf{M}}_{\{i\}}})\leq \sum_{{\mathcal{M}}\in{{\mathbf{M}}_{\{i\}}}}(H(U_{{\mathcal{M}}})+\rho_{{\mathcal{M}},i}-r_{{\mathcal{M}}}),
\label{packingoptimality}
\end{align}

adding \eqref{coveringoptimality} and \eqref{packingoptimality} we get: 
\begin{align}
\sum_{{\mathcal{M}}\in {\mathbf{M}}_{\{i\}}}\rho_{{\mathcal{M}},i} \geq I(U_{{\mathbf{M}}_{\{i\}}};X,Z),
\end{align}

 $R_i=\sum_{{\mathcal{M}}\in{{\mathbf{S}_\mathsf{L}}}} \rho_{{\mathcal{M}},i}$, comparing this equality with (\ref{optimality}) completes the proof.
\end{proof}

\textbf{Step 2:} In this step, we show that there are no common codebooks decoded at decoders  $\{1\}$ and $\{2\}$. Since decoder $\{1,2\}$ receives descriptions 1 and 2 at optimal RD from a PtP perspective, the random variables decoded at decoder $\{1\}$ must be independent of those decoded at decoder $\{2\}$.  From the next lemma we have that if $ {\mathcal{M}}\in \mathbf{M}_{\{1\}}\bigcap\mathbf{M}_{\{2\}}$ then $r_{{\mathcal{M}}}=0$.

\begin{Lemma}
  Consider the setup in Figure \ref{2Desc}, let
  $(R_1,R_2,D_{1},D_{2},D_{\{1,2\}})$ be such that
  $R_1+R_2=RD_{d_{\{1,2\}}}(D_{\{1,2\}})$, where $RD_d(D)$ is
  Shannon's optimal PtP RD function for distortion function $d$ at
  point $D$. For any distribution
  $P_{U_{\{1\}},U_{\{2\}},U_{\{1,2\}},U_{\{1\},\{2\}}}$ which achieves
  this RD vector, the following conditions must hold: 
\label{indeplem}
\noindent1)$U_{\{1\}}\Perp U_{\{2\}}$ and $C_{\{1\},\{2\}}=\phi$
\\2)If in addition $R_i=RD_{d_{\{i\}}}(D_i),i\in\{1,2\}$ then, $U_{\{1,2\}}\leftrightarrow (U_{\{1\}},U_{\{2\}}) \leftrightarrow X$.
 
\end{Lemma}

\begin{proof}

 Consider the following packing bounds:
\begin{align}
&\textbf{Dec}\hspace{0.03 in} \textbf{$\{1\}$}: H(U_{\{1\},\{2\}},U_{\{1\}})\leq H(U_{\{1\},\{2\}}) \! + \! H(U_{\{1\}}) + \! \rho_{\{1\},\{2\},1} \! + \! \rho_{\{1\},1} \! - \! r_{\{1\},\{2\}} \! - \! r_{\{1\}}
\label{packdec1}\\
&\textbf{Dec}\hspace{0.03 in} \textbf{$\{2\}$}:H(U_{\{1\},\{2\}},U_{\{2\}})\leq H(U_{\{1\},\{2\}}) \! + \! H(U_{\{2\}}) + \! \rho_{\{1\},\{2\},2} \! + \! \rho_{\{2\},2} \! - \! r_{\{1\},\{2\}} \! - \! r_{\{2\}}\\
&\textbf{Dec}\hspace{0.03 in} \textbf{$\{1,2\}$}:H(U_{\{1,2\}}|U_{\{1\},\{2\}},U_{\{1\}},U_{\{2\}})\leq H(U_{\{1,2\}})  + \! \rho_{\{1,2\},1 \! }+ \! \rho_{\{1,2\},2} \! - \! r_{\{1,2\}}
\label{packdec12}
\end{align}

Also the mutual covering bound:
\begin{align}
&H(U_{\{1,2\}},U_{\{1\}},U_{\{2\}},U_{\{1\},\{2\}}|X)\geq H(U_{\{1,2\}})+H(U_{\{1\}})+H(U_{\{2\}})+H(U_{\{1\},\{2\}})-r_{\{1,2\}}-r_{\{1\}}-r_{\{2\}}-r_{\{1\},\{2\}}
\label{covmut}
\end{align}
Now we add inequalities (\ref{packdec1}-\ref{packdec12}) and subtract \eqref{covmut}, we get:
\begin{align*}
&I(U_{\{1,2\}},U_{\{1\}},U_{\{2\}},U_{\{1\},\{2\}};X)+I(U_{\{1\}};U_{\{2\}}|U_{\{1\},\{2\}})
\leq  R_1+R_2-r_{\{1\},\{2\}}
\end{align*}

Using the condition $R_1+R_2=RD_{d_{12}}(D_{\{1,2\}})$ we conclude: 
\begin{align}
&I(U_{\{1\}};U_{\{2\}}|U_{\{1\},\{2\}})+r_{\{1\},\{2\}}\leq 0
\label{indep}
\end{align}
From \eqref{indep} one may deduce $C_{\{1\},\{2\}}=\phi$ and $U_{\{1\}}\Perp U_{\{2\}}$. Furthermore we get:

\begin{align*}
&R_1+R_2= I(U_{\{1,2\}},U_{\{1\}},U_{\{2\}};X)= I(U_{\{1\}};X)\!+\!I(U_{\{2\}};X)\!+\!I(U_{\{1\}};U_{\{2\}}|X)+\!I(U_{\{1,2\}};X|U_{\{1\}},U_{\{2\}}),\\
\end{align*}
where the right-hand side of the second equality is the sum-rate of the two-descriptions problem. Using the conditions $R_i=RD_{d_{i}}(D_i),i\in\{1,2\}$, we have:
\begin{align*}
&\!I(U_{\{1\}};U_{\{2\}}|X)\!+\!I(U_{\{1,2\}};X|U_{\{1\}},U_{\{2\}})=0.
\end{align*}
So $I(U_{\{1,2\}};X|U_{\{1\}},U_{\{2\}})=0$, which gives the desired Markov chain in $(2)$.
\end{proof}

 Assuming the original scheme achieves the RD vector in the theorem, we give a new scheme which also achieves the RD vector. We propose that the encoder operates as before, but decoder $\{1,2\}$ decodes $U_{{\mathcal{M}}}$ only if $ {\mathcal{M}}\in \mathbf{M}_{\{1\}}$ or $ {\mathcal{M}}\in \mathbf{M}_{\{2\}}$. It needs to be shown that the RD vector is the same. First we consider the resulting rates. The covering bounds are not changed. The packing bounds are the same at all decoders other than decoder $\{1,2\}$ since the same variables are being decoded at those decoders. $\mathbf{M_{\{1\}}}\cap\mathbf{M_{\{2\}}}=\phi$. Let ${\tilde{\mathbf{M}}}_{\{1\}}$ and ${\tilde{\mathbf{M}}}_{\{2\}}$ be subsets of ${{\mathbf{M}}}_{\{1\}}$ and ${{\mathbf{M}}}_{\{2\}}$. We need to show that the following packing bound is satisfied:
\begin{align}
&H(U_{{\mathbf{M}}_{\{1\}}},U_{{\mathbf{M}}_{\{2\}}}|U_{\tilde{\mathbf{M}}_{\{1\}}},U_{{\tilde{\mathbf{M}}}_{\{2\}}})\leq \sum_{{\mathcal{M}}\in {{\mathbf{M}}_{\{1\}}}\cup{\mathbf{M}}_{\{2\}}\backslash{\tilde{\mathbf{M}}}_{\{1\}}\cup{\tilde{\mathbf{M}}}_{\{1\}}}
(H(U_{{\mathcal{M}}})+\rho_{{\mathcal{M}},1}+\rho_{{\mathcal{M}},2}-r_{{\mathcal{M}}})
\label{goal}
\end{align}

We have the following two packing bounds from decoders $\{1\}$ and $\{2\}$:
\begin{align}
&H(U_{{\mathbf{M}}_{\{1\}}}|U_{{\tilde{\mathbf{M}}}_{\{1\}}})\leq\!\!\!\!\!\!\! \sum_{{\mathcal{M}}\in {{\mathbf{M}}_{\{1\}}}\backslash \tilde{\mathbf{M}}_1}\!\!\!\!\!
(H(U_{{\mathcal{M}}})+\rho_{{\mathcal{M}},1}-r_{{\mathcal{M}}})
\label{packing1}
\\&H(U_{{\mathbf{M}}_{\{2\}}}|U_{{\tilde{\mathbf{M}}}_{\{2\}}})\leq\!\!\!\!\!\!\! \sum_{{\mathcal{M}}\in {{\mathbf{M}}_{\{2\}}}\backslash \tilde{\mathbf{M}}_1}\!\!\!\!\!
(H(U_{{\mathcal{M}}})+\rho_{{\mathcal{M}},2}-r_{{\mathcal{M}}})
\label{packing2}
\end{align}

Note that from arguments in Lemma \ref{indeplem}, $U_{{\mathbf{M}}_{\{1\}}}$ is independent of $U_{{\mathbf{M}}_{\{2\}}}$. Hence adding \eqref{packing1} and \eqref{packing2}, we get  \eqref{goal}. This proves that the packing bounds are also the same.

 From lemma \ref{indeplem}, we have ${{U}}_{\mathbf{M}_{\{1,2\}}}\leftrightarrow {{U}}_{\mathbf{M}_{\{1\}}},{{U}}_{\mathbf{M}_{\{2\}}} \leftrightarrow X,Z$. Lemma \ref{distmark} shows that the new scheme achieves the same distortions as the previous one. 

\begin{Lemma}
Let the random variables $U,V,X$ be such that $U\leftrightarrow V\leftrightarrow X$. Then for an arbitrary distortion function $f:\mathsf{X}\times\hat{\mathsf{X}}\to \mathsf{R}^+$, there is an optimal reconstruction of $X$ using $U$ and $V$ which is a only function of $V$.

\label{distmark}
\end{Lemma}
\begin{proof}
 We know that the optimal reconstruction function for $X$ given $U$ and $V$ is given by:
 
\begin{align*}
 g(u,v)&= \arg
 \min_{\hat{x}\in\hat{\mathsf{X}}}\mathsf{E}(f(\hat{x},X)|u,v) =\arg \min_{\hat{x}\in\hat{\mathsf{X}}}\mathsf{E}(f(\hat{x},X)|v),
\end{align*}
 which is only a function of $V$. 
\end{proof}

 

 

  By these arguments, codebook $U_{{\mathcal{M}}}$ is eliminated if ${\mathcal{M}}\in \mathbf{M}_{\{1,2\}}\backslash\widetilde{\mathbf{M}}_{\{1,2\}}$.  Also in the new scheme, $U_{\{1,2\},\{1,3\},\{2,3\}}$ and $U_{\{1,3\},\{2,3\}}$ are functionally similar since by the same arguments as in this step $U_{\{1,2\},\{1,3\},\{2,3\}}$ is not used in the reconstruction in decoder $\{1,2\}$ , so we can eliminate $\mathcal{C}_{\{1,2\},\{1,3\},\{2,3\}}$. In summary, thus far we have eliminated 7 codebooks.
\\\textbf{Step 3:} We have the following lemma:
\begin{Lemma}
From optimality of rate and distortion at decoders $\{1,3\}, \{2,3\}$ we have:
\begin{align*}
\rho_{\{2,3\},3}=\rho_{\{1,3\},3}=\rho_{\{2,3\},\{1\},3}=\rho_{\{1,3\},\{2\},3}=0
\end{align*}
\end{Lemma}

\begin{proof}
 
 First we argue that $\rho_{\{2,3\},3}=0$. If this is not true, it contradicts optimality at decoder $\{1,3\}$. $U_{\{2,3\}}$ is not decoded at decoder $\{1,3\}$, but its bin index is carried through description 3. So if the bin index is non-zero, one could reduce $R_3$ by setting the bin index equal to 0 without increasing distortion at decoder $\{1,3\}$, this contradicts optimality at that decoder.  By the same arguments $\rho_{\{1,3\},3}=0$. Now assume $\rho_{\{2,3\},\{1\},3}\neq 0$. We show that this contradicts optimality at decoder $\{1,3\}$. $U_{\{2,3\},\{1\}}$ is decodable using description $1$ (since it is decodable at decoder $\{1\}$).  Hence, if we set $\rho_{\{2,3\},\{1\}\},3}$ to $0$ (i.e. do not send the bin index on description 3), then decoder $\{1,3\}$ can still decode $U_{23,1}$ using description 1. So the distortion is the same at this decoder, but the rate $R_3$ is reduced which contradicts optimality.  By the same arguments, $\rho_{\{13,2\},3}=0$.
\end{proof}

\noindent\textbf{Step 4:} We proceed by showing that $r_{\{1,3\}}=r_{\{2,3\}}=0$. So far we have shown that none of the descriptions carry the bin indices for these codebooks.
 Consider the following packing bounds in decoders $\{1\}$, $\{2,3\}$ and $\{1,3\}$:
\begin{align*}
&H(U_{\{1\}}U_{\{1\},\{3\}}U_{\{2,3\},\{1\}})\leq  H(U_{\{1\}})+H(U_{\{1\},\{3\}})+H(U_{{\{2,3\},\{1\}}})+R_1-r_{\{1\}}-r_{\{1\},\{3\}}-r_{\{2,3\}\{1\}}
\\&H(U_{\{2\}}U_{\{3\}}U_{\{1\},\{3\}}U_{\{2\},\{3\}}U_{\{2,3\}}U_{\{1,3\},\{2\}}U_{\{2,3\},\{1\}}U_{\{1,3\}\{2,3\}})\nonumber\leq H(U_{\{2\}})+H(U_{\{3\}})+H(U_{\{1\},\{3\}})+\\&H(U_{\{2\}\{3\}})+H(U_{\{2,3\}})+H(U_{\{1,3\},\{2\}})+H(U_{{\{2,3\},\{1\}}})\nonumber+H(U_{{\{1,3\},\{2,3\}}})+R_2+R_3-r_{\{2\}}-r_{\{3\}}-r_{\{1\}\{3\}}-\\&r_{\{2\},\{3\}}-r_{\{2,3\}}-r_{\{1,3\}\{2\}}-r_{\{2,3\}\{1\}}-r_{\{1,3\}\{2,3\}}
\\&H(U_{\{1,3\}}|U_{\{1\}}U_{\{3\}}U_{\{1\},\{3\}}U_{\{2\},\{3\}}U_{\{1,3\}\{2\}}U_{\{2,3\},\{1\}}U_{\{1,3\}\{2,3\}})\leq H(U_{\{1,3\}})-r_{\{1,3\}}
\end{align*}
We add the above inequalities and subtract the mutual covering bound on all RV\rq{}s, we get:
\begin{align*}
 &H(U_{\{1\}}U_{\{1\},\{3\}}U_{\{2,3\},\{1\}})+H(U_{\{2\}}U_{\{3\}}U_{\{1\},\{3\}}U_{\{2\},\{3\}}U_{\{2,3\}}U_{\{1,3\},\{2\}}U_{\{2,3\},\{1\}}U_{\{1,3\}\{2,3\}})\\&+H(U_{\{1,3\}}|U_{\{1\}}U_{\{3\}}U_{\{1\},\{3\}}U_{\{2\},\{3\}}U_{\{1,3\}\{2\}}U_{\{2,3\},\{1\}}U_{\{1,3\}\{2,3\}})\\&-H(U_{\{1\}}, U_{\{2\}},U_{\{3\}}, U_{\{1\},\{3\}}, U_{\{2\},\{3\}},U_{\{1,3\},\{2\}}, U_{\{2,3\},\{1\}}, U_{\{2,3\}}, U_{\{1,3\},\{2,3\}}, U_{\{1,3\}}|X,Z)
 \\&\leq H(U_{\{1,3\}})+H(U_{\{2,3\},\{1\}})-r_{\{1,3\}}-r_{\{2,3\},\{1\}}+R_1+R_2+R_3
\\&\Rightarrow I(U_{\{1,3\}}|U_{\{1\}}U_{\{3\}}U_{\{1\},\{3\}}U_{\{2\},\{3\}}U_{\{1,3\}\{2\}}U_{\{2,3\},\{1\}}U_{\{1,3\}\{2,3\}};X,Z)+I(U_{\{1\}}, U_{\{1\},\{3\}}U_{\{2,3\},\{1\}};X,Z)+\\&I(U_{\{1,3\}};X,Z|U_{\{1\}}U_{\{3\}}U_{\{1\},\{3\}}U_{\{2\},\{3\}}U_{\{1,3\}\{2\}}U_{\{2,3\},\{1\}}U_{\{1,3\}\{2,3\}})\leq R_1+R_2+R_3
\\&\Rightarrow I(X,Z;U_{\{1,3\}}|U_{\{1\}}U_{\{3\}}U_{\{1\},\{3\}}U_{\{2\},\{3\}}U_{\{1,3\}\{2\}}U_{\{2,3\},\{1\}}U_{\{1,3\}\{2,3\}})=0.
\end{align*}
 This imposes the Markov chain $U_{\{1,3\}}\leftrightarrow U_{\{1\}}U_{\{3\}}U_{\{1\},\{3\}} $ $U_{\{2\},\{3\}} U_{\{1,3\}\{2\}} U_{\{2,3\},\{1\}} U_{\{1,3\}\{2,3\}} \leftrightarrow X,Z$. Hence by the same arguments as in step 2, we can eliminate $\mathcal{C}_{\{1,3\}}$. Also by the same arguments $C_{\{2,3\}}$ can be eliminated. 
 \\\textbf{Step 5:} In this step we eliminate $\mathcal{C}_{\{1\},\{3\}}$ and $\mathcal{C}_{\{2\},\{3\}}$. 
\begin{Lemma}
The following equality holds:
\begin{equation*}
\rho_{\{1\},\{3\},1}=\rho_{\{1\},\{3\},3}=\rho_{\{2\},\{3\},2}=\rho_{\{2\},\{3\},1}=0
\end{equation*}
\end{Lemma}
\begin{proof}
 
 Assume $\rho_{\{1\},\{3\},1}>0$. We claim this contradicts optimality at decoder $\{1,3\}$, since $U_{\{1\},\{3\}}$ can readily be decoded from the bin number carried by description $3$, so setting $\rho_{\{1\},\{3\},1}$ to $0$ would decease rate without increasing distortion. The rest of the proof follows by the same argument.
\end{proof}

Now consider the following packing bounds at decoders $\{1\}$, $\{3\}$ and $\{1,3\}$ and the mutual covering bound:
\begin{align*}
&H(U_{\{1\}}U_{\{1\},\{3\}}U_{\{2,3\},\{1\}})\leq H(U_{\{1\}})+H(U_{\{1\},\{3\}})+H(U_{{\{2,3\},\{1\}}})+R_1-r_{\{1\}}-r_{\{1\},\{3\}}-r_{\{2,3\}\{1\}}\\
&H(U_{\{3\}}U_{\{1\},\{3\}}U_{\{2\},\{3\}})\leq H(U_{\{3\}})+H(U_{\{1\},\{3\}})+H(U_{\{2\}\{3\}})+R_3-\rho_{\{1,3\}\{2,3\},3}-r_{\{3\}}-r_{\{1\}\{3\}}-r_{\{2\},\{3\}}\\
&H(U_{\{1,3\},\{2,3\}}, U_{\{1,3\}\{2\}} |U_{\{1\}}U_{\{3\}}U_{\{1\},\{3\}}U_{\{2\},\{3\}}U_{\{2,3\},\{1\}})\leq H(U_{\{1,3\},\{2,3\}})+H(U_{\{1,3\},\{2\}})+\rho_{\{1,3\}\{2,3\},3}-r_{\{1,3\},\{2,3\}}\\
&H(U_{\{1\}}U_{\{3\}}U_{\{1\},\{3\}}U_{\{2\},\{3\}}U_{\{1,3\}\{2\}}U_{\{2,3\},\{1\}}U_{\{1,3\},\{2,3\}}|X,Z)\geq H(U_{\{1\}})+H(U_{\{3\}})+ H(U_{\{1\},\{3\}})+ H(U_{\{2\},\{3\}})\\&+ H(U_{\{1,3\}\{2\}}) + H(U_{\{2,3\},\{1\}})+H(U_{\{1,3\},\{2,3\}})-r_{\{1\}}+r_{\{3\}}-r_{\{1\},\{3\}}- r_{\{2\},\{3\}}-r_{\{1,3\},\{,2\}}- r_{\{2,3\},\{1\}}- r_{\{1,3\},\{2,3\}}\\
\end{align*}

Adding the above packing bounds and subtracting the mutual covering bound we get:
\begin{align*}
&H(U_{\{1\}}U_{\{1\},\{3\}}U_{\{2,3\},\{1\}})+H(U_{\{3\}}U_{\{1\},\{3\}}U_{\{2\},\{3\}})+H(U_{\{1,3\},\{2,3\}},U_{\{1,3\}\{2\}}|U_{\{1\}}U_{\{3\}}U_{\{1\},\{3\}}U_{\{2\},\{3\}}U_{\{2,3\},\{1\}})\\&-H(U_{\{1\}}U_{\{3\}}U_{\{1\},\{3\}}U_{\{2\},\{3\}}U_{\{1,3\}\{2\}}U_{\{2,3\},\{1\}}U_{\{1,3\},\{2,3\}}|X,Z)\leq R_1+R_3+H(U_{\{1\}\{3\}}-r_{\{1\},\{3\}}\\
&\Rightarrow I(U_{\{1\}}U_{\{1\},\{3\}}U_{\{2,3\},\{1\}};U_{\{3\}}U_{\{1\},\{3\}}U_{\{2\},\{3\}})+\\&I(U_{\{1\}}U_{\{3\}}U_{\{1\},\{3\}}U_{\{2\},\{3\}}U_{\{1,3\}\{2\}}U_{\{2,3\},\{1\}}U_{\{1,3\},\{2,3\}}|X,Z)-H(U_{\{1\},\{3\}})\leq R_1+R_3-r_{\{1\},\{3\}}\\
&\Rightarrow I(U_{\{1\}}U_{\{2,3\},\{1\}};U_{\{3\}}U_{\{2\},\{3\}}|U_{\{1\},\{3\}})+r_{\{1\},\{3\}}\leq 0
\end{align*}

Particularly $r_{\{1\},\{3\}}=0$, by the same arguments $r_{\{2\},\{3\}}=0$.
\\\textbf{Step 6:} So far we have shown that only $C_{\{1\}},C_{\{2\}},C_{\{3\}}, C_{\{1,3\},\{2\}}, C_{\{2,3\},\{1\}}$ and $C_{\{1,3\}\{2,3\}}$ can be non-trivial. From optimality at decoders $\{1\}$ and $\{1,3\}$ we have the following equalities:

\begin{align}
&R_1=I(U_{\{1\}},U_{\{2,3\},\{1\}};X,Z), R_1+R_3=I(U_{\{1\}},U_{\{2,3\},\{1\}},U_{\{3\}},U_{\{1,3\},\{2\}},U_{\{1,3\},\{2,3\}};X,Z)
\end{align}
Hence we have:
\begin{align}
&R_3=I(U_{\{3\}},U_{\{1,3\},\{2\}},U_{\{1,3\},\{2,3\}};X,Z|U_{\{1\}},U_{\{2,3\},\{1\}})
\label{R3}
\end{align}
Define the following:
\begin{align}
&N_{\delta}^{1}\triangleq X+h_{\{1\}}(U_{\{1\}},U_{\{2,3\},\{1\}})\\
&N_{\delta\ast \delta}^{3}\triangleq X+Z+h_{\{3\}}(U_{\{3\}})\\
&N_{\delta}^{1,3}\triangleq Z+h_{\{1,3\}}(U_{\{1\}},U_{\{2,3\},\{1\}},U_{\{3\}},U_{\{1,3\},\{2\}},U_{\{1,3\},\{2,3\}}),
\end{align}
where $h_{\{1\}}$ is the reconstruction of $X$ at decoder $\{1\}$, $h_{\{3\}}$ is the reconstruction of $X+Z$ at decoder $\{3\}$, and $h_{\{1,3\}}$ is the reconstruction of $Z$ at decoder $\{1,3\}$. Then from \eqref{R3}:
\begin{align*}
&R_3\geq I(h_{\{1,3\}}(
.),h_{\{3\}}(U_{\{3\}});X,Z|U_{\{1\}},U_{\{2,3\},\{1\}}h_{\{1\}}(
.))\\
&\Rightarrow R_3\geq I(Z+N_{\delta}^{1,3},X+Z+N_{\delta\ast \delta}^{3};X,Z|U_{\{1\}},U_{\{2,3\},\{1\}},X+N_{\delta}^{\{1\}}))\\
&\Rightarrow R_3\geq H(Z|U_{\{1\}},U_{\{2,3\},\{1\}},X+N_{\delta}^{\{1\}})-H(Z|Z+N_{\delta}^{1,3},X+Z+N_{\delta\ast \delta}^{3},U_{\{1\}},U_{\{2,3\},\{1\}},X+N_{\delta}^{\{1\}})\\
&\Rightarrow R_3\stackrel{(a)}{\geq} 1-H(Z|Z+N_{\delta}^{1,3},X+Z+N_{\delta\ast \delta}^{3},U_{\{1\}},U_{\{2,3\},\{1\}},X+N_{\delta}^{\{1\}})\\
&\Rightarrow R_3\geq 1-H(Z|Z+N_{\delta}^{1,3})\\
&\Rightarrow R_3\geq 1-H(N_{\delta}^{1,3}|Z+N_{\delta}^{1,3})\\
&\stackrel{(b)}{\Rightarrow} R_3\geq 1-h_b(\delta)
\end{align*}
All the above inequalities must be equality. In particular we have:
\begin{align*}
&(a) \Rightarrow Z\leftrightarrow Z+N_{\delta}^{1,3} \leftrightarrow X+Z+N_{\delta\ast \delta}^{3}, X+N_{\delta}^{\{1\}}\\
&\Rightarrow N_{\delta}^{1,3} \leftrightarrow Z+N_{\delta}^{1,3} \leftrightarrow Z+N_{\delta\ast \delta}^{3}+N_{\delta}^{\{1\}}\\
& \Rightarrow N_{\delta}^{1,3} \leftrightarrow Z+N_{\delta}^{1,3} \leftrightarrow N_{\delta}^{1,3}+N_{\delta\ast \delta}^{3}+N_{\delta}^{\{1\}}
\end{align*}
Note that from (b), we can conclude that $Z$ is independent of $N_{\delta}^{1,3}$, we have $N_{\delta}^{1,3}$ and $N_{\delta}^{1,3}+N_{\delta\ast \delta}^{3}+N_{\delta}^{\{1\}}$ are independent. Define $N\rq{}\triangleq N_{\delta\ast \delta}^{3}+N_{\delta}^{\{1\}}$. We have:
\begin{align*}
&P(N_{\delta}^{\{1,3\}}+N\rq{}=0)\stackrel{(a)}{=} P(N_{\delta}^{\{1,3\}}+N\rq{}=0|N_{\delta}^{\{1,3\}}=0)=P(N\rq{}=0|N_{\delta}^{\{1,3\}}=0)\\
&\stackrel{(b)}{\Rightarrow} P(N\rq{}=0,N_{\delta}^{\{1,3\}}=0)= (1-\delta)(P(N\rq{}=0,N_{\delta}^{\{1,3\}}=0)+P(N\rq{}=1,N_{\delta}^{\{1,3\}}=1))\\
&\Rightarrow P(N\rq{}=0,N_{\delta}^{\{1,3\}}=0)=\frac{1-\delta}{\delta}P(N\rq{}=1,N_{\delta}^{\{1,3\}}=1)
\end{align*}
(a) holds since $N_{\delta}^{1,3}$ and $N_{\delta}^{1,3}+N_{\delta\ast \delta}^{3}+N_{\delta}^{\{1\}}$ are independent. In (b) we have replaced $P(N_{\delta}^{\{1,3\}}+N\rq{}=0)$ by $P(N\rq{}=0,N_{\delta}^{\{1,3\}}=0)+P(N\rq{}=1,N_{\delta}^{\{1,3\}}=1)$. 

Define $a\triangleq P(N\rq{}=1,N_{\delta}^{\{1,3\}}=1)$, then by the same calculations $P(N\rq{}=1,N_{\delta}^{\{1,3\}}=0)=(1-\delta)(1-\frac{1}{\delta}a)$, so $P(N\rq{}=1)=1-\delta+\frac{2\delta-1}{\delta}a$. Note $a=P(N\rq{}=1,N_{\delta}^{\{1,3\}}=1)\leq P(N_{\delta}^{\{1,3\}}=1)=\delta$, hence using $P(N'=1)=1-\delta+\frac{2\delta-1}{\delta}a$, we get $P(N\rq{}=1)\leq \delta$ with equality if and only if $a=\delta$. Also note that $Z+N\rq{}$ is available at decoder $\{1,3\}$ so $P(N'=1)=\delta$ and $a=\delta$, otherwise there is a contradiction with optimality of $h_{\{1,3\}}$. If $a=\delta$, then $N_{\delta}^{\{1,3\}}$ is equal to $N\rq{}$. So by the same arguments we have:
\begin{align*}
N_{\delta\ast \delta}^{\{3\}}=N_{\delta}^{\{1,3\}}+N_{\delta}^{\{1\}}= N_{\delta}^{\{2,3\}}+N_{\delta}^{\{2\}},
\end{align*}
where
\begin{align}
&N_{\delta}^{2}\triangleq Z+h_{\{2\}}(U_{\{2\}},U_{\{1,3\},\{2\}})\\
&N_{\delta}^{2,3}\triangleq Z+h_{\{2,3\}}(U_{\{2\}},U_{\{1,3\},\{2\}},U_{\{3\}},U_{\{2,3\},\{1\}},U_{\{1,3\},\{2,3\}})
\end{align}
Since $N_{\delta}^{\{1\}}\Perp N_{\delta}^{\{2\}}$, $N^{\{1\}}\Perp N^{\{1,3\}}$ and $N^{\{2\}}\Perp N^{\{2,3\}}$, we have: 
\begin{align*}
&N_{\delta}^{\{1,3\}}= N_{\delta}^{2}, N_{\delta}^{\{2, 3\}}= N_{\delta}^{\{1\}}, N_{\delta\ast\delta}^{\{3\}}= N_{\delta}^{1}+N_{\delta}^{2}
\end{align*}
We argue that $C_{\{1,3\},\{2\}}, C_{\{2,3\},\{1\}}$ and $C_{\{1,3\}\{2,3\}}$ can be taken eliminated without any loss in RD. To prove this assume we have a scheme with $P_{U_{\{1,3\},\{2\}}, U_{\{2,3\},\{1\}}, U_{\{1,3\}\{2,3\}}, U_{\{1\}}, U_{\{2\}}, U_{\{3\}}}$. Construct new random variables $\tilde{U}_{\{1\}}=X+N_\delta^{\{1\}}, U_{\{1\}}, U_{\{1\},\{2,3\}}$, $\tilde{U}_{\{2\}}=Z+N_\delta^{\{2\}}, U_{\{2\}}, U_{\{2\},\{1,3\}}$ and $\tilde{U}_3=U_{\{3\}}$ and eliminate the rest of the codebooks. From the independence relations above, the packing bounds would stay the same. Since we have merged codebooks, the covering bounds would loosen, and it is straightforward to see that the reconstructions at each decoder are still the same. 
 We are left with four codebooks, $C_{\{1\}},C_{\{2\}}$ and $C_{\{3\}}$. Note that since decoder $\{1\}$ is only decoding $C_{\{1\}}$ we must have $\rho_{\{1\},1}=r_{\{1\}}=R_1$. This is deduced from the packing bound in decoder $\{1\}$:
\begin{equation*}
H(U_{\{1\}})\leq H(U_{\{1\}})+\rho_{\{1\},1}-r_{\{1\}}\Rightarrow r_{\{1\}}\leq \rho_{\{1\},1}
\end{equation*}
But $\rho_{\{1\},1}\leq r_{\{1\}}$ so they are equal. The same argument gives $\rho_{\{2\},2}=r_{\{2\}}=R_2$, and $\rho_{\{3\},3}=r_{\{3\}}=R_3$. Also, from optimality at the joint decoders and lemma \ref{indeplem},  we have $U_i\Perp U_j, \forall i\neq j$. 
\begin{align}
&H(U_{\{1\}}, U_{\{2\}},U_{\{3\}}|X,Z)\geq H(U_{\{1\}}+H(U_{\{2\}}+H(U_{\{3\}}-R_1-R_2-R-3\nonumber\\
&\Rightarrow I(U_{\{1\}}, U_{\{2\}},U_{\{3\}};X,Z)+I(U_{\{3\}};X,Z,U_{\{1\}},U_{\{2\}})\leq R_1+R_2+R_3\nonumber\\
&\Rightarrow I(U_{\{3\}};X,Z,U_{\{1\}},U_{\{2\}})\leq R_3\label{r3}
\end{align}
Note that $R_1+R_3=I(U_{\{1\}},U_{\{3\}};X,Z)$ and $R_1=I(U_{\{1\}};X)$ from optimality at decoders $\{1\}$ and $\{1,3\}$. So $R_3=I(U_{\{3\}};X,Z|U_{\{1\}})$. Replacing $R_3$ into \eqref{r3}, we get $I(U_{\{3\}};U_{\{2\}}|U_{\{1\}},X,Z)=0$. So we have the Markov chain $U_{\{3\}}\leftrightarrow U_{\{1\}},X,Z\leftrightarrow U_{\{2\}}$. By the same arguments we can derive the Markov chain  $U_{\{3\}}\leftrightarrow U_{\{2\}},X,Z\leftrightarrow U_{\{1\}}$. Using lemma \ref{Markovindeplem} and the previous two Markov chains we get $U_{\{3\}}\leftrightarrow X,Z \leftrightarrow U_{\{1\}},U_{\{2\}}$. Take the Markov chain $U_{\{3\}}\leftrightarrow X,Z \leftrightarrow U_{\{1\}}$, along with $Z\Perp X, U_{\{1\}}$ we get $U_{\{3\}},Z\leftrightarrow X \leftrightarrow U_{\{1\}}$. Also from the optimality of the reconstruction of X at decoders $\{1\}$ and $\{1,3\}$, we have:
\begin{align*}
 I(U_{\{1\}};X)=I(U_{\{1\}},U_{\{3\}};X) \Rightarrow I(U_{\{3\}};X|U_{\{1\}})=0.
\end{align*}
From the above and $Z\Perp X, U_{\{1\}}$, we conclude $U_{\{3\}},Z \leftrightarrow U_{\{1\}}\leftrightarrow X$. Applying Lemma \ref{Markovindeplem} we get $Z, U_{\{3\}}\Perp X,U_{\{1\}}$. 
\begin{Lemma} Let A,B,C and D be RV\rq{}s such that $A\leftrightarrow B,C\leftrightarrow D$ and $A\leftrightarrow B,D\leftrightarrow C$, and also assume there is no $b\in \mathcal{B}$ for which given $B=b$ there are non-constant functions $f_b(C)$ and $g_b(D)$ with $f_b(C)=g_b(D)$ with probability 1. Then $A\leftrightarrow B\leftrightarrow C,D$.\label{Markovindeplem}
\end{Lemma}
\begin{proof} This lemma is a generalization of the one in \cite{wagner}. We need to show that $p(A=a|B=b,C=c,D=d)=p(A=a|B=b,C=c\rq{},D=d\rq{})$ for any $a,b,c,c\rq{},d,d\rq{}$. Note since functions $f_b$ and $g_b$ do not exist, it is straightforward to show that there is a finite sequence of pairs $(c_i,d_i)$ such that $(c_1,D_{\{1\}})=(c,d)$ and $(c_n,d_n)=(c\rq{},d\rq{})$ with the property that either $c_{i}=c_{i+1}$ or $d_i=d_{i+1}$ and that $p(B=b,C=c_i,D=d_i) \neq 0$. Then from the first Markov chain if $d_i=d_{i+1}$, we have $p(A=a|B=b,C=c_i,D=d_i)=p(A=a|B=b,C=c_{i+1},D=d_{i+1})$, also if $c_i=c_{i+1}$ the second Markov chain gives this result. So  $p(A=a|B=b,C=c_i,D=d_i)$ is constant on all of the sequence particularly $p(A=a|B=b,C=c,D=d)=p(A=a|B=b,C=c\rq{},D=d\rq{})$. 
\end{proof}
Let $g(U_{\{1\}},U_{\{3\}})$ be the reconstruction of $Z$ at decoder $\{1,3\}$. We have:
\begin{align*}
&\sum_{z,u_{\{1\}},u_{\{3\}}}p(z,u_{\{1\}},u_{\{3\}})d_H(g(u_{\{1\}},u_{\{3\}}),z)\leq \delta\Rightarrow \sum_{u_{\{1\}}} p(u_{\{1\}})\sum_{z,u_{\{3\}}}p(z,u_{\{3\}})d_H(g(u_{\{1\}},u_{\{3\}}),z)\leq \delta\\
\end{align*}
So there is at least one $u_{\{1\}}\in \mathsf{U}_{\{1\}}$ such that $\sum_{z,u_{\{3\}}}p(z,u_{\{3\}})d_H(g(u_{\{1\}},u_{\{3\}},z)\leq \delta$. Let $g_{u_{\{1\}}}(U_{\{3\}})=g(u_{\{1\}},U_{\{3\}})$ be the reconstruction of Z using $U_{\{3\}}$. By the same argument we can find a reconstruction of X using $U_{\{3\}}$, then $I(U_{\{3\}};X,Z)\geq 2(1-h_b(\delta))$ from a PtP perspective which is a contradiction. 
\end{proof}

\subsection{Proof of Lemma \ref{theorem9}}\label{Ap:theorem9}
\begin{proof}

We provide an outline of the proof here, the arguments are similar to the ones in the previous proofs. 
\\\textbf{Step 1:} I
Any codebook which is not decoded at decoders $\{1\}$, $\{1,2\}$, $\{2,3\}$, $\{3,4\}$ and $\{4\}$ is redundant. This implies that there are at most only  17 codebooks which are non-redundant. These codebooks are $\mathcal{C}_{\{1\}}$ , $\mathcal{C}_{\{1\}, \{2,3\}}$, $\mathcal{C}_{\{1\}, \{3,4\}}$, $\mathcal{C}_{\{1\}, \{4\}}$ ,$\mathcal{C}_{\{1\}, \{2,3\}, \{3,4\}}$, $\mathcal{C}_{\{1\}, \{4\}, \{2,3\}}, \mathcal{C}_{\{4\}}$, $\mathcal{C}_{\{4\}, \{2,3\}}, \mathcal{C}_{\{4\}, \{1,2\}}$, $\mathcal{C}_{\{4\}, \{2,3\}, \{1,2\}}$,$ \mathcal{C}_{\{1,2\}}$,$ \mathcal{C}_{\{2,3\}}$, $\mathcal{C}_{\{3,4\}}$,  $\mathcal{C}_{\{1,2\}, \{2,3\}}$ ,$\mathcal{C}_{\{1,2\}, \{3,4\}}$, $\mathcal{C}_{\{2,3\},\{3,4\}}$ and  $\mathcal{C}_{\{1,2\}, \{2,3\}, \{3,4\}} $.
 \\\textbf{Step 2:} In this step we prove that the only non-trivial codebook decoded at decoder $\{i\}$ is $\mathcal{C}_{\{i\}}$ for $i=1,4$. 
  All possible codebooks decoded at decoder $\{1\}$ are $\mathcal{C}_{\{1\}}$, $\mathcal{C}_{\{1\},\{2,3\}}$, $\mathcal{C}_{\{1\},\{3,4\}}$,$\mathcal{C}_{\{1\},\{4\}}$, $\mathcal{C}_{\{1\},\{2,3\},\{3,4\}}$ and $\mathcal{C}_{\{1\}, \{2,3\},\{4\}}$. From optimality at decoder $\{1,2\}$, $C_{\{1\},\{2,3\}}$ is redundant. The reason is $\rho_{\{1\},\{2,3\},2}=0$ otherwise we can set it to zero without any loss in distortion at decoder $\{1,2\}$ which contradicts optimality, also any random variable that description $\{3\}$ carries must be used in reconstructing $Z$ at decoder $\{3,4\}$ because that decoder is at optimality, which means $\rho_{\{1\},\{2,3\},3}=0$ so the codebook is decoded at decoder $\{2,3\}$ but not sent through either description $\{2\}$ or $\{3\}$, from similar arguments as before the codebook is redundant. Same arguments can be provided to deduce redundancy of $\mathcal{C}_{\{1\},\{3,4\}}$, $\mathcal{C}_{\{2,3\}}$, $\mathcal{C}_{\{1\},\{2,3\},\{3,4\}}$ and $\mathcal{C}_{\{1\}, \{2,3\},\{4\}}$. This implies that only $\mathcal{C}_{\{1\}}$ is decoded at decoder $\{1\}$ and $\mathcal{C}_{\{4\}}$ at decoder $\{4\}$. 
\\\textbf{Step 3:}  We proceed with eliminating $\mathcal{C}_{\{1,2\},\{3,4\}}$ and $\mathcal{C}_{\{1,2\}, \{2,3\}, \{3,4\}}$. Using the PtP optimality of decoder $\{1,2\}$ we have:
\begin{align*}
& I(U_{\{1\}}, U_{\{1,2\}}, U_{\{1,2\},\{2,3\}},U_{\{1,2\},\{3,4\}}, U_{\{1,2\},\{2,3\},\{3,4\}};X)=\\&R_1+R_2\stackrel{(a)}{\geq} I(U_{\{1\}}, U_{\{1,2\}}, U_{\{1,2\},\{2,3\}},U_{\{1,2\},\{3,4\}}, U_{\{1,2\},\{2,3\},\{3,4\}}; X,Z)
\end{align*}
where (a) follows from the usual PtP source coding results. Comparing
the LHS with the RHS we conclude the Markov chain $U_{\{1\}},
U_{\{1,2\}}, U_{\{1,2\},\{2,3\}},U_{\{1,2\},\{3,4\}},
U_{\{1,2\},\{2,3\},\{3,4\}} \leftrightarrow X\leftrightarrow Z$. In
particular we are interested in  $U_{\{1,2\},\{3,4\}},
U_{\{1,2\},\{2,3\},\{3,4\}} \leftrightarrow X\leftrightarrow Z$. By
the same arguments and using the optimality at decoder $\{3,4\}$, we
get  $U_{\{1,2\},\{3,4\}}, U_{\{1,2\},\{2,3\},\{3,4\}} \leftrightarrow
Z\leftrightarrow X$. These two Markov chains along with lemma
\ref{Markovindeplem} prove  $U_{\{1,2\},\{3,4\}},
U_{\{1,2\},\{2,3\},\{3,4\}} \Perp X,Z$. So these two variables are not
used in reconstructing the source and the corresponding codebooks are
eliminated. 
\\\textbf{Step 4:}  The only remaining codebooks are $\mathcal{C}_{\{1\}}$, $\mathcal{C}_{\{4\}}$, $\mathcal{C}_{\{1,2\}}$, $\mathcal{C}_{\{3, 4\}}$, $\mathcal{C}_{\{1,2\}, \{2,3\}}$ and $\mathcal{C}_{\{2,3\}, \{3,4\}}$.  From optimality at decoders $\{1,2\}$ and $\{3,4\}$ we must have $U_{\{1\}}, U_{\{1,2\}}, U_{\{1,2\}, \{2, 3\}}\leftrightarrow (X,Z) \leftrightarrow  U_{\{4\}}, U_{\{3, 4\}}, U_{\{2,3\}, \{3, 4\}}$, also $U_{\{1\}}, U_{\{1,2\}}, U_{\{1,2\}, \{2, 3\}}\leftrightarrow X \leftrightarrow  Z$ and $X\leftrightarrow Z \leftrightarrow  U_{\{4\}}, U_{\{3, 4\}}, U_{\{2,3\}, \{3, 4\}}$. From lemma \ref{MClemma}, we get $U_{\{1\}}, U_{\{1,2\}}, U_{\{1,2\}, \{2, 3\}}\leftrightarrow X\leftrightarrow Z \leftrightarrow  U_{\{4\}}, U_{\{3, 4\}}, U_{\{2,3\}, \{3, 4\}}$. 
\begin{Lemma} For random variables A,B,C,D, the three short Markov chains $A\leftrightarrow (B,C)\leftrightarrow D$, $A\leftrightarrow B \leftrightarrow C$ and $B\leftrightarrow C \leftrightarrow D$ are equivalent to the long Markov chain $A\leftrightarrow B \leftrightarrow C \leftrightarrow D$. 
\label{MClemma}
\end{Lemma}
\begin{proof}
We only need to show that $A\leftrightarrow B \leftrightarrow D$, the rest of the implications of the long Markov chain are either direct results of the three short Markov chains or follow by symmetry. For arbitrary $a,b,d$ we have:
\begin{align*}
P(D=d|B=b,A=a)&\!=\!\sum_{c\in{\mathcal{C}}}\!\!P(C=c|B=b,A=a)P(D=d|A=a,B=b,C=c)\\&\!=\sum_{c\in{\mathcal{C}}}P(C=c|B=b)P(D=d|B=b,C=c)=P(D=d|B=b)
\end{align*}
\end{proof}
We get an inner bound for $R_2+R_3$ at decoder $\{2,3\}$:
\begin{equation*}
R_2+R_3\geq \min  I(U,V;X,Z)= H(X,Z)=1+h_b(p),
\end{equation*}
where the minimum is taken over all $P_{U,V|X,Z}$ for which the long Markov chain $U\leftrightarrow X\leftrightarrow Z \leftrightarrow V$ is satisfied and $(U,V)$ produce a lossless reconstruction of $X+Z$. This resembles the distributed source coding problem in \cite{KM}. So the RD vector can\rq{}t be achieved using random codes.
\end{proof}
\subsection{Proof of Lemma \ref{thm: same_code}}\label{Ap:thm: same_code}

\begin{proof}
In this proof we use bold letters to denote vectors and matrices.
Fix integers $n, l, l'$ and $k$. Choose the elements of the matrices $\mathbf{\Delta G}_{l\times n}$, $\mathbf{\Delta G}'_{k'\times n}$  and $\mathbf{G}_{k\times n}$ and vectors $\mathbf{B}^n$ and  $\mathbf{B}'^n$ randomly and uniformly from $\mathbb{F}_q$. The codebooks $\mathcal{C}^n_o$ and $\mathcal{C}'^n_o$ are defined as follows:
\begin{align*}
&\mathcal{C}_o=\{\mathbf{a}\mathbf{G}+\mathbf{m}\mathbf{\Delta G}+\mathbf{B}|\mathbf{a}\in \mathbb{F}_q^k, \mathbf{m}\in \mathbb{F}^l_q\}\\&
\mathcal{C}'_o=\{\mathbf{ b}\mathbf{G}+\mathbf{m}'\mathbf{\Delta G}'+\mathbf{B}'|\mathbf{b}\in \mathbb{F}_q^{k}, \mathbf{m}'\in \mathbb{F}^{l'}_q\}
\end{align*}
For a typical sequence $\mathbf{x}$ with respect to $P_X$, we define $\theta(\mathbf{x})$ as the function which counts the number of codewords in $\mathcal{C}_o$ and $\mathcal{C}'_o$ jointly typical with respect to $P_{XUV}$:
\begin{align*}
 \theta (\mathbf{x})&= \sum_{\mathbf{u} \in \mathcal{C}'_o,\mathbf{v}\in \mathcal{C}_o } \mathbb{I}\{(\mathbf{u},\mathbf{v})\in A_\epsilon^n(U,V|\mathbf{x})\}\\
 & = \sum_{\mathbf{m},\mathbf{m'}}\sum_{\mathbf{a},\mathbf{b}  \in \mathbb{F}_q^k}\sum_{(\mathbf{u},\mathbf{v})\in A_\epsilon^n(U,V|\mathbf{x})} \mathbb{I}\{\mathbf{a}\mathbf{G}+\mathbf{m}\mathbf{\Delta G}+\mathbf{B}=\mathbf{u},\mathbf{ b}\mathbf{G}+\mathbf{m}'\mathbf{\Delta G}'+\mathbf{B}'=\mathbf{v}\}
\end{align*}

Our goal is to find bounds on $n, l, l'$ and $k$ such that $P(\theta(\mathbf{x})=0)\to 0$ as $n\to\infty$.

 For $\mathbf{a}\in \mathbb{F}_q^k$ and $\mathbf{m} \in \mathbb{F}_q^l$, we denote the corresponding codeword as $g(\mathbf{a},\mathbf{m}):= \mathbf{a}\mathbf{G}+\mathbf{m} \mathbf{\Delta G} +\mathbf{B}$. Similarly define $g'(\mathbf{b},\mathbf{m}'):= \mathbf{b}\mathbf{G}+\mathbf{m}' \mathbf{\Delta G'} +\mathbf{B}'$ for any $\mathbf{b}\in \mathbb{F}_q^{k}$ and $\mathbf{m}' \in \mathbb{F}_q^{l'}$. The following lemma proves several results on the pairwise independence of the codewords.

\begin{Lemma}\label{lem: indep} The following hold:
\begin{enumerate}
\item $g(\mathbf{a},\mathbf{m})$ and $g'(\mathbf{b},\mathbf{m}')$  are distributed uniformly uniform over $\mathbb{F}_q^n$.
\item If $\mathbf{a}\neq \mathbf{\tilde{a}}$, then $g(\mathbf{a},\mathbf{m})$ is independent of $g(\mathbf{\tilde{a}},\mathbf{m})$.
\item If $\mathbf{b}\neq \mathbf{\tilde{b}}$, then $g'(\mathbf{b},\mathbf{m}')$ is independent of $g'(\mathbf{\tilde{b}},\mathbf{m}')$.
\item If $\mathbf{B}$ and $\mathbf{B}\rq{}$ are chosen independently and uniformly over $\mathbb{F}_q^n$, then $g(\mathbf{b},\mathbf{m}')$ and $g'(\mathbf{a},\mathbf{m})$ are independent.
\end{enumerate}
\end{Lemma}
\begin{proof}
Follows from \cite{IC},  and the fact that $\mathbf{B}, \mathbf{B}'$ are independent and uniform.
\end{proof}
We intend to use Chebyshev's inequality to obtain:
\begin{align*}
P\{ \theta(\mathbf{X})=0 \} &\leq \frac{4var\{ \theta(\mathbf{X})\}}{\mathsf{E}\{\theta(\mathbf{X}) \}^2}\to 0\\
\end{align*}

%
%
%

\begin{Lemma}
\label{lem: var_exp}
For $\mathbf{X}\in A_\epsilon^{(n)}(X)$ we have the following bound on $\frac{var\{ \theta(\mathbf{X})\}}{\mathsf{E}\{\theta(\mathbf{X}) \}^2}$:
\begin{align*}
\frac{var\{ \theta(X)\}}{\mathsf{E}\{\theta(X) \}^2}\leq& \frac{q^{2n}}{q^{l+l'}q^{2k}} 2^{-n (H(U,V|X)} + \frac{q^{n}}{q^{l+l'}q^{k}} 2^{-n(H(U|X))} + \frac{q^{n}}{q^{l+l'}q^{k}} 2^{-n(H(V|X))}+ \frac{q^{n}}{q^{l+l'}q^{k}} 2^{-n(H(U,V|X)-\max_{i\neq 0}{H(U,V|X, V+iU)})}\\
& + \frac{q^{n}}{q^{l}q^{k}} 2^{-n(H(U|X))}
 + \frac{q^{n}}{q^{l'}q^{k}} 2^{-n(H(V|X))} + \frac{1}{q^{l}}+\frac{1}{q^{l'}}+\frac{1}{q^{l+l'}}
 + \frac{1}{q^{l+k}}+\frac{1}{q^{l'+k}}
\end{align*}
\end{Lemma}
\begin{proof}
We calculate the expected value of $\theta(\mathbf{X})$ for any $\mathbf{X}\in A_\epsilon^{(n)}(X)$:
\begin{align}\nonumber
\mathsf{E}\{ \theta(\mathbf{X})\}&= \sum _{\mathbf{x} \in A_\epsilon^n(X)} \sum_{\substack{\mathbf{m}\in\mathbb{F}_q^l \\ \mathbf{m}' \in \mathbb{F}_q^{l'}}} \sum_{\mathbf{a} \neq \mathbf{b}} \sum_{(\mathbf{u},\mathbf{v})\in A_\epsilon^n(U,V|\mathbf{x})} P(\mathbf{x})\quad P\{g(\mathbf{a},\mathbf{m})= \mathbf{u}, g'(\mathbf{b},\mathbf{m}')=\mathbf{v}\}\\\nonumber
\label{equ:exp}
&= \sum_{\mathbf{m},\mathbf{m}'} \sum _{\mathbf{x} \in A(X)}
\sum_{\mathbf{a} \neq \mathbf{b}} | A_\epsilon^n(U,V|\mathbf{x})|
P(\mathbf{x}) \frac{1}{q^{2n}}
= \frac{q^{l+l'} q^{2k}}{q^{2n}} 2^{n (H(U,V|X)+O(\epsilon))}
\end{align}
Also:
\begin{align}\nonumber
&\mathsf{E}\{ \theta(\mathbf{X})^2\}= \sum_{\substack{\mathbf{m}, \mathbf{\tilde{m}}\\ \mathbf{m}',\mathbf{\tilde{m}}' }} \sum_{\mathbf{a},\mathbf{ \tilde{a}}} \sum_{\mathbf{b} , \mathbf{\tilde{b}}} \sum_{(\mathbf{u} ,\mathbf{v})}  \sum_{(\mathbf{\tilde{u}},\mathbf{\tilde{v})}\in A_\epsilon^n(U,V|\mathbf{x})}
    P\{g(\mathbf{a},\mathbf{m})=\mathbf{u}, g(\mathbf{\tilde{a}},\mathbf{\tilde{m}})=\mathbf{\tilde{u}}, g'(\mathbf{b},\mathbf{m}')=\mathbf{v}, g'(\mathbf{\tilde{b}},\mathbf{\tilde{m}}')=\mathbf{\tilde{v}}\}\\ 
\end{align}
Using Lemma \ref{lem: indep}:  
\begin{align*}
&P_S \triangleq  P\{g(\mathbf{a},\mathbf{m})=\mathbf{u}, g(\mathbf{\tilde{a}},\mathbf{\tilde{m}})=\tilde{u}, g(\mathbf{b},\mathbf{m}')=\mathbf{v}, g(\mathbf{\tilde{b}},\mathbf{\tilde{m'}})=\mathbf{\tilde{v}}\}\\
&= \frac{1}{q^{2n}}\times P\{g_0(\mathbf{a}-\mathbf{\tilde{a}},\mathbf{m}-\mathbf{\tilde{m}})=\mathbf{u}-\mathbf{\tilde{u}},g'_0(\mathbf{b}-\mathbf{\tilde{b}},\mathbf{m}'-\mathbf{\tilde{m}}')=\mathbf{v}-\mathbf{\tilde{v}}\}
\end{align*}
 At this point we have to consider several different cases for the values of $\mathbf{a},\mathbf{\tilde{a}},\mathbf{ b}, \mathbf{\tilde{b}},\mathbf{ m}, \mathbf{\tilde{m}},\mathbf{m}',\mathbf{ \tilde{m}}'$.

1)\qquad $\mathbf{m} =\mathbf{\tilde{m}},\mathbf{m'} = \mathbf{\tilde{m}'}$

1.1: $\mathbf{a} =\mathbf{ \tilde{a}},\mathbf{ b} =\mathbf{\tilde{b}} \Rightarrow P_s=\frac{1}{q^{2n}}\delta(\mathbf{u}-\mathbf{\tilde{u}})\delta(\mathbf{v}-\mathbf{\tilde{v}})$ 

1.2: $\mathbf{a} =\mathbf{ \tilde{a}},\mathbf{ b} \neq \mathbf{\tilde{b}}$ $\Rightarrow P_s=\frac{1}{q^{3n}}\delta(\mathbf{u}-\mathbf{\tilde{u}})$

1.3: $\mathbf{a} \neq \mathbf{\tilde{a}}, \mathbf{b} =\mathbf{\tilde{b}}$ $\Rightarrow P_s=\frac{1}{q^{3n}}\delta(\mathbf{v}-\mathbf{\tilde{v}})$

1.4: $\mathbf{a} \neq \mathbf{\tilde{a}},\mathbf{ b }\neq \mathbf{\tilde{b}}\Rightarrow P_s=\sum_{\alpha\in \mathbb{F}_q}\frac{1}{q^{3n}}\delta(\mathbf{u}-\mathbf{\tilde{u}}-\alpha(\mathbf{v} -\mathbf{\tilde{v}}))+\frac{1}{q^{4n}}(1-\sum_{\alpha\in \mathbb{F}_q}\delta(\mathbf{u}-\mathbf{\tilde{u}}-\alpha(\mathbf{v} -\mathbf{\tilde{v}})))$

\vspace{0.1in}

 2)\qquad$\mathbf{m} \neq \mathbf{\tilde{m}},\mathbf{m'} = \mathbf{\tilde{m}}'$

2.1: $\mathbf{a} = \mathbf{\tilde{a}},\mathbf{ b} =\mathbf{\tilde{b}}$ $\Rightarrow P_s=\frac{1}{q^{3n}}\delta( \mathbf{v}-\mathbf{\tilde{v}})$

2.2: $\mathbf{a} = \mathbf{\tilde{a}}, \mathbf{b} \neq \mathbf{\tilde{b}}$ $\Rightarrow P_s=\frac{1}{q^{4n}}$

2.3: $\mathbf{a} \neq \mathbf{\tilde{a}},\mathbf{ b}=\mathbf{\tilde{b}}$ $\Rightarrow P_s=\frac{1}{q^{3n}}\delta(\mathbf{v}-\mathbf{\tilde{v}})$

2.4: $\mathbf{a} \neq \mathbf{\tilde{a}},\mathbf{ b }\neq \mathbf{\tilde{b}}$ $\Rightarrow P_s=\frac{1}
{q^{4n}}$

Cases  when $\mathbf{m} = \mathbf{\tilde{m}},\mathbf{ m}' \neq \mathbf{\tilde{m}}'$ and $\mathbf{m} \neq \mathbf{\tilde{m}}, \mathbf{ m'} \neq \mathbf{\tilde{m}}'$ are similarly considered but the derivations are omitted for brevity.  Considering cases $1.1\tiny{-}4$:
\begin{align}
\mathsf{E}\{ \theta(\mathbf{x})^2|\mathbf{m} = \mathbf{\tilde{m}},\mathbf{ m}' = \mathbf{\tilde{m}}'\} 
\label{Part 1}
=&\sum_{\mathbf{m},\mathbf{m}'} \Bigg[\sum_{\mathbf{a} = \mathbf{\tilde{a}}} \sum_{\mathbf{b} =\mathbf{ \tilde{b}}} \sum_{(\mathbf{u},\mathbf{v})\in A_\epsilon^n(U,V|\mathbf{x})}  \frac{1}{q^{2n}} 
+\sum_{\mathbf{a} = \mathbf{\tilde{a}}} \sum_{\substack{ \mathbf{b} \neq \mathbf{\tilde{b}}}} \sum_{(\mathbf{u},\mathbf{v}), (\mathbf{u},\mathbf{\tilde{v}}) }  \frac{1}{q^{3n}}
\\&+ \sum_{\mathbf{a} \neq \mathbf{\tilde{a}}} \sum_{\substack{ \mathbf{b} = \mathbf{\tilde{b}}}} \sum_{(\mathbf{u},\mathbf{v}), (\mathbf{\tilde{u}},\mathbf{v})}  \frac{1}{q^{3n}}
+\sum_{\alpha\in \mathbb{F}_q\backslash\{ 0\}}\sum_{\mathbf{a} \neq \mathbf{\tilde{a}}} \sum_{\small{\substack{ \mathbf{b} \neq \mathbf{\tilde{b}}\\\mathbf{ b}-\mathbf{\tilde{b}}=\alpha(\mathbf{a}-\mathbf{\tilde{a}})\\ \alpha\in \mathbb{F}_q\backslash{0}}}} \sum_{\small{\substack{(\mathbf{u},\mathbf{v}), (\mathbf{\tilde{u}},\mathbf{\tilde{v}})\\\mathbf{v}-\mathbf{\tilde{v}}= \alpha(\mathbf{u}-\mathbf{\tilde{u}})}}}\frac{1}{q^{3n}}\\ 
&+\sum_{\alpha\in \mathbb{F}_q-\{ 0\}}\sum_{\mathbf{a} \neq \mathbf{\tilde{a}}} \sum_{\small{\substack{\mathbf{b} \neq \mathbf{\tilde{b}}\\ \mathbf{b}-\mathbf{\tilde{b}}\neq \alpha(\mathbf{a}-\mathbf{\tilde{a}})\\ \alpha\in \mathbb{F}_q\backslash{0}}}} \sum_{\small{\substack{(\mathbf{u},\mathbf{v}), (\mathbf{\tilde{u}},\mathbf{\tilde{v}})\\\mathbf{v}-\mathbf{\tilde{v}}\neq \alpha(\mathbf{u}-\mathbf{\tilde{u}})}}}  \frac{1}{q^{4n}}\Bigg]
\end{align}
Consequently:
\begin{align*}
\mathsf{E}\{ \theta(\mathbf{X})^2|&\mathbf{m'} = \mathbf{\tilde{m}},\mathbf{m}' = \mathbf{\tilde{m}}'\}\\
&\leq \frac{q^{l+l'} q^{2k}}{q^{2n}} 2^{n (H(U,V|X))} + \frac{q^{l+l'} q^{3k}}{q^{3n}} 2^{n(H(U,V|X)+H(V|X,U))} + \frac{q^{l+l'} q^{3k}}{q^{3n}} 2^{n(H(U,V|X)+H(U|X,V))}+ \\&\frac{q^{l+l'}q^{3k}}{q^{3n}} 2^{{n(H(U,V|X)+\max_{\alpha\neq 0}{H(U,V|X, V+\alpha U)})}} + \frac{q^{l+l'}q^{4k}}{q^{4n}} 2^{2n(H(U,V|X))},
\end{align*}
where we have used Lemma 8 in \cite{allertonsandeep} to get the fourth term.  After considering all the cases, the only non-redundant bounds are the ones mentioned in the lemma. 
\end{proof}
So, the following bounds need to be satisfied:
\begin{align*}
r_o+r'_o &\geq 2\log q - H(U,V|X)\\
r_o+r'_o-r_i &\geq \log q - \min\{H(U|X),H(V|Z)\}\\
r_o+r'_o-r_i &\geq \log q- H(U,V|X) + \max_{\alpha\neq 0}{H(U,V|X, V+\alpha U)}\\
r_o  &\geq \log q - H(U|X))\\
r'_o &\geq \log q - H(V|X))\\
\min\{r_o,r'_o\}  &\geq r_i
\end{align*}
Observe that
\begin{align*}
H(U,V|X, V\!\!+\alpha U)&=H(U,V, V\!\!+\!\alpha U|X)- H(V\!\!+\!\alpha U|X)= H(U,V|X)- H(V+\alpha U|X)
\end{align*}

\end{proof}
 
\subsection{Proof of Lemma \ref{thm: same_codepack}}\label{Ap:thm: same_codepack}

\begin{proof}
 The proof follows the same arguments as that of Lemma \ref{thm: same_code}. We provide an outline of the proof. Define the probability of error $P_e$ as follows:
\begin{align*}
 P_e=P(\{(\mathbf{x},\mathbf{u},\mathbf{v})\in \mathsf{X}^n\times\mathcal{C}_1\times\mathcal{C}_2|\exists (\mathbf{u}',\mathbf{v}')\in A_{\epsilon}^n(U,V)\cap B_2(\mathbf{u})\times B_2(\mathbf{v})\})
\end{align*}
 We define a new conditional probability of error for any triple $\mathbf{x},\mathbf{u},\mathbf{v}\in A_{\epsilon}(X,U,V)$:
\begin{align*}
 P_{e|\mathbf{x},\mathbf{u},\mathbf{v}}=P(\exists (\mathbf{u}',\mathbf{v}')\in A_{\epsilon}^n(U,V)\cap B_2(\mathbf{u})\times B_2(\mathbf{v})|\mathbf{X}=\mathbf{x}, (\mathbf{u},\mathbf{v})\in \mathcal{C}_1\times\mathcal{C}_2)
\end{align*}
 Clearly if $P_{e|\mathbf{x},\mathbf{u},\mathbf{v}}$ goes to 0 for all $\mathbf{x},\mathbf{u},\mathbf{v}\in A_{\epsilon}(X,U,V)$ as $n\to \infty$, then $P_e$ goes to 0. Also define:
$P_{\mathbf{x},\mathbf{u},\mathbf{v}}=
P({(\mathbf{x},\mathbf{u},\mathbf{v})\in \mathsf{X}^n\times
  \mathcal{C}_1\times \mathcal{C}_2)}$, and 
$P_{e,\mathbf{x},\mathbf{u},\mathbf{v}}=P_{e|\mathbf{x},\mathbf{u},\mathbf{v}}P_{\mathbf{x},\mathbf{u},\mathbf{v}}$.
We have:
\begin{align*}
 P_{\mathbf{x},\mathbf{u},\mathbf{v}}&= \sum _{\mathbf{x} \in A_\epsilon^n(X)} \sum_{\substack{\mathbf{m}\in\mathbb{F}_q^l \\ \mathbf{m}' \in \mathbb{F}_q^{l'}}} \sum_{\mathbf{a} \neq \mathbf{b}} \sum_{(\mathbf{u},\mathbf{v})\in A_\epsilon^n(U,V|\mathbf{x})} P(\mathbf{x}) P\{g(\mathbf{a},\mathbf{m})= \mathbf{u}, g'(\mathbf{b},\mathbf{m}')=\mathbf{v}\}\\
&= \sum_{\mathbf{m},\mathbf{m}'} \sum _{\mathbf{x} \in A(X)}
\sum_{\mathbf{a} \neq \mathbf{b}} | A_\epsilon^n(U,V|\mathbf{x})|
P(\mathbf{x}) \frac{1}{q^{2n}}
= \frac{q^{l+l'} q^{2k}}{q^{2n}} 2^{n (H(U,V|X)+O(\epsilon))}
\end{align*}
\begin{align*}
 &P_{e,\mathbf{x},\mathbf{u},\mathbf{v}}= \sum_{\substack{\mathbf{m}, \mathbf{\tilde{m}}\\ \mathbf{m}',\mathbf{\tilde{m}}' }}\sum _{\mathbf{x}} \sum_{\mathbf{a},\mathbf{ \tilde{a}}} \sum_{\mathbf{b} , \mathbf{\tilde{b}}} \sum_{\substack{(\mathbf{u} ,\mathbf{v})\in\\ A_\epsilon^n(U,V|\mathbf{x})}}  \sum_{\substack{(\mathbf{\tilde{u}},\mathbf{\tilde{v}})\in\\ A_{\epsilon}^n(U,V)}}\sum_{b_1\in [1,2^{n\rho_1}]}\sum_{b_2\in [1,2^{n\rho_2}]}P(\mathbf{x}) \\\nonumber
    &P\{g(\mathbf{a},\mathbf{m})=\mathbf{u}, g(\mathbf{\tilde{a}},\mathbf{\tilde{m}})=\mathbf{\tilde{u}}, g'(\mathbf{b},\mathbf{m}')=\mathbf{v}, g'(\mathbf{\tilde{b}},\mathbf{\tilde{m}}')=\mathbf{\tilde{v}}\}P\{B_1(\mathbf{u})=B_1({\mathbf{\tilde{u}}})=b_1, B_2(\mathbf{u})=B_2({\mathbf{\tilde{u}}})=b_2\}\\ 
\end{align*}
 Note that the binning is done independently and uniformly, so $P\{B_1(\mathbf{u})=B_1({\mathbf{\tilde{u}}})=b_1, B_2(\mathbf{u})=B_2({\mathbf{\tilde{u}}})=b_2\}=2^{-2(\rho_1+\rho_2)}$. The rest of the summations are the ones which were present in the proof of Lemma \ref{thm: same_code}. Again we have to do a case by case investigation of the summation. The only new bond comes from the case when $\mathbf{m}=\mathbf{\tilde{m}}$ and $\mathbf{{m}'}=\mathbf{\tilde{m}'}$, $a\neq \tilde{a}, b\neq \tilde{b}$ and $a-\tilde{a}=i(b-\tilde{b})$. We have: 
 \begin{align*}
&A= \sum_{\substack{\mathbf{m}, \mathbf{{\tilde{m}}} }} \sum_{\substack{\mathbf{a},\mathbf{ \tilde{a}}\\\mathbf{a}\neq\mathbf{\tilde{a}}}} \sum_{\substack{\mathbf{b} , \mathbf{\tilde{b}}\\\\\mathbf{a}-\mathbf{\tilde{a}}=i(\mathbf{b}-\mathbf{\tilde{b}})}} \sum_{\substack{(\mathbf{u} ,\mathbf{v})\in\\ A_\epsilon^n(U,V|\mathbf{x})}}  \sum_{\substack{(\mathbf{\tilde{u}},\mathbf{\tilde{v}})\in\\ A_{\epsilon}^n(U,V)\\\mathbf{u}-\mathbf{\tilde{u}}=i(\mathbf{v}-\mathbf{\tilde{v}})}}q^{-3n}2^{-n(\rho_1+\rho_2)}\\
&= \frac{q^{l+l'}}{q^{3n}}q^{3k}2^{nH(U,V|X)}2^{nH(U,V|U+iV)}2^{-n(\rho_1+\rho_2)}
\end{align*}
Dividing this last term by $P_{\mathbf{x},\mathbf{u},\mathbf{v}}$:
\begin{align*}
 \frac{A}{P_{\mathbf{x},\mathbf{u},\mathbf{v}}}=\frac{q^{k}}{q^n}2^{nH(U,V|U+iV)}2^{-n(\rho_1+\rho_2)}
\end{align*}
which goes to $0$ if the following is satisfied: 
\begin{align}
 &r_i-\rho_1-\rho_2 \leq \log{q}- H(U,V| U+i V)
 \label{neweqpackred}
\end{align}
However as shown in the next lemma the new bound in (\ref{neweqpackred}) is redundant. 
\begin{Lemma}\label{lem:neweqred}
 The inequality (\ref{neweqpackred}) in Lemma \ref{thm: same_codepack} is redundant.
\end{Lemma}
\begin{proof}
 Assume there is a distribution $P_{X,U,V}$ for which  (\ref{neweqpackred}) is violated, we show that either (\ref{neweq}) or  (\ref{sumpack}) is also violated. Conversely, as long as  (\ref{neweq}) and  (\ref{sumpack}) are satisfied,  (\ref{neweqpackred}) is also satisfied. Assume we have:
\begin{align*} 
 &(r_o-\rho_1)+(r'_o-\rho_2) \leq 2\log q -H(U,V)\\
&r_i-\rho_1-\rho_2> \log{q}- H(U,V|U+i V),\forall i\in \mathbb{F}_q.
\end{align*}
Adding the two bounds we get:
\begin{align*} 
 r_o+r'_o-r_i &< \log q -H(U,V)+H(U,V| U+i V)\\
 &=\log{q} -H( U +i V) \leq\log{q}- H( U+ i V|X)
\end{align*}
which contradicts (\ref{neweq}). 
 
\end{proof}
\end{proof}

\subsection{Proof of Lemma  \ref{theorem11}}\label{Ap:theorem11}
\begin{proof}
The proof follows the same arguments as in the previous two examples. First we assume there exists a joint distribution $P_{\mathbf{U}X}$ such that the SSC scheme achieves the RD vector, then we arrive at a contradiction by eliminating all codebooks. First note that from our definition of $P_{X,V_{\{1\}},V_{\{2\}}}$, direct calculation shows that $R_1+R_2=I(V_{\{1\}},V_{\{2\}};X)=1-h_b(D_0)$. This means that decoder $\{1,2\}$ is at PtP optimality. Also by the definition of the distortion function $D_{\{3\}}$, decoder $\{3\}$ is at optimal RD.  

\textbf{Step 1:} From the optimality of decoder $\{1,2\}$ and Lemma \ref{indeplem}, there can't be any codebook common between decoders $\{1\}$ and $\{2\}$. So $\mathcal{C}_{\{1\},\{2\}}$ and $\mathcal{C}_{\{1\},\{2\},\{3\}}$ are eliminated. 

 \textbf{Step 2:} From optimality of decoder $\{3\}$, description 3 can't carry the bin number of any codebook which is not decoded at that decoder. Also description 1 and 2 can't carry the bin numbers of codebooks which are not decoded at $\{1,2\}$ because of optimality at this decoder. So codebooks $
 \mathcal{C}_{\{1,3\},\{2,3\}}$, $\mathcal{C}_{\{1,3\}}$ and $\mathcal{C}_{\{2,3\}}$ are not sent on any description and are redundant. 
 
 \textbf{Step 3:} The codebook $\mathcal{C}_{\{1\},\{2,3\}}$ is not binned by description 2 or 3. Description 3 can't bin the codebook since it is not decoded at decoder $\{3\}$, and that decoder is at PtP optimality. Note $\mathcal{C}_{\{1\},\{2,3\}}$ can be decoded using description 1, so any bin information for this codebook that is carried by description 2 is not used at decoder $\{1,2\}$, since decoder $\{1,2\}$ is at PtP optimality we must have $\rho_{\{1\},\{2,3\},2}=0$. The codebook is not sent on description 2 or 3, so by the same arguments as in the previous proofs it can't help in the reconstruction at decoder $\{2,3\}$ and is redundant. By the same arguments $\mathcal{C}_{\{2\},\{1,3\}}$ is redundant.  
  
  \textbf{Step 4:} In this step we show that there is no refinement codebook decoded at decoder $\{1,2\}$. This would eliminate $\mathcal{C}_{\{1,2\}}, \mathcal{C}_{\{1,2\},\{3\}}, \mathcal{C}_{\{1,2\},\{1,3\}}, \mathcal{C}_{\{1,2\},\{2,3\}}$ and $\mathcal{C}_{\{1,2\}, \{1,3\},\{2,3\}}$. More precisely we show that the reconstruction at decoder $\{1,2\}$ is a function of the reconstructions at decoders $\{1\}$ and $\{2\}$. This means that sending a refinement codebook to decoder $\{1,2\}$ will not help in the reconstruction, so the codebook is redundant. 
  
 To prove this claim we consider the two user example depicted in Figure \cite{ZB}. Here all distortions are Hamming distortions. We are interested in achieving the rate distortion vector $(R_1,R_2,D_{\{1\}},D_{\{2\}},D_{\{1,2\}})$ given in (\ref{RDscalar}). Let $P_{X,U_{\{1,2\}},U_{\{1\}}, U_{\{2\}},U_{\{1\},\{2\}}}$ be a distribution on the random variables in the two user SSC achieving this RD vector. Define $\hat{X}_{1}, \hat{X}_2$ and $\hat{X}_{12}$ as the reconstructions at the corresponding codebooks.
  
\begin{Lemma}
There are only two choices for the joint distribution $P_{X,\hat{X}_1,\hat{X}_2,\hat{X}_{12}}$, furthermore in both choices, $\hat{X}_{12}$ is a function of $\hat{X}_1$ and $\hat{X}_2$.
\end{Lemma}

\begin{proof}
As in step 1, from optimality of decoder $\{1,2\}$, $\mathcal{C}_{\{1,2\}}$ is redundant. Also $U_{\{1\}}$ and $U_{\{2\}}$ are independent from Lemma \ref{indeplem}. Note that $\hat{X}_1$ is a function of $U_{\{1\}}$ and $\hat{X}_2$ is a function of $U_{\{2\}}$, so $\hat{X}_1\Perp \hat{X}_2$. We proceed by characterizing $P_{X,\hat{X}_{12}}$. Note that decoder $\{1,2\}$ is at PtP optimality. It is well-known result that when quantizing a BSS to Hamming distortion $D_0$ with rate $1-h_b(D_0)$, the reconstruction is uniquely given by $\hat{X}_{12}=X+\mathsf{N}_0, \mathsf{N}_0\sim Be(D_0)$ where $\mathsf{N}_0\Perp X$. $\hat{X}_1, \hat{X}_2$ and $\hat{X}_{12}$ are available at decoder $\{1,2\}$, from optimality at this decoder we must have:
\begin{align*}
 1-h_b(D_0)=I(\hat{X}_1,\hat{X}_2,\hat{X}_{1,2};X)\geq I(\hat{X}_{12},X)=1-h_b(D_0).
\end{align*}
So the inequality must be equality, which means $I(\hat{X}_1,\hat{X}_2;X|X_{12})=0$. In other words the Markov chain $\hat{X}_1,\hat{X}_2\leftrightarrow \hat{X}_{12}\leftrightarrow X$ must hold.  Using the three facts 1) $\hat{X}_{12}=X\oplus_2\mathsf{N}_0$, 2) $\hat{X}_1\Perp \hat{X}_2$ and 3) $\hat{X}_1,\hat{X}_2\leftrightarrow \hat{X}_{12}\leftrightarrow X$, we can characterize all possible distributions on $P_{X, \hat{X}_{12},\hat{X}_1,\hat{X}_2}$. Let $\hat{X}_1\sim Be(a_1)$ and $\hat{X}_2\sim Be(a_2)$. Then from  $\hat{X}_1\Perp \hat{X}_2$, $P_{\hat{X}_1,\hat{X}_2}$ is fixed. Assume the distribution $P_{\hat{X}_{12},\hat{X}_1,\hat{X}_2}$ is as given below:
\setlength{\extrarowheight}{.3cm}
\begin{table}[h]
\centering
\begin{tabular}{x{2cm}|x{2cm}|x{2cm}|x{2cm}|x{2cm}|x{2cm}}
\diag{.3em}{2cm}{$\qquad \hat{X}_{12}$}{$\hat{X}_1,\hat{X}_2$}& \multicolumn{1}{c}{00}&\multicolumn{1}{c}{01}&\multicolumn{1}{c}{10}&\multicolumn{1}{c}{11}&\multicolumn{1}{c}{Sum}\\\cline{1-6}
0&$P_{000}$&$P_{001}$&$P_{010}$&$P_{011}$&$\frac{1}{2}$\\  \cline{2-5}
1&$P_{100}$&$P_{101}$&$P_{110}$&$P_{111}$&$\frac{1}{2}$\\  \cline{2-5}
Sum& \multicolumn{1}{c}{$(1-a_1)(1-a_2)$}&\multicolumn{1}{c}{$(1-a_1)a_2$}&\multicolumn{1}{c}{$a_1(1-a_2)$}&\multicolumn{1}{c}{$a_1a_2$}&\\
\end{tabular}
\vspace{0.1in}
\caption{}
\end{table}
As shown on the table there are 5 independent linear constraints on $P_{ijk}$'s. We have:
\begin{align*}
 &P_{011}=\frac{1}{2}-P_{000}-P_{001}-P_{010}, \qquad
 P_{100}=(1-a_1)(1-a_2)-P_{000}, \qquad
P_{101}=(1-a_1)a_2-P_{001}, \\
&P_{110}=(1-a_1)a_2-P_{010}, \qquad P_{111}=a_1a_2-\frac{1}{2}+P_{000}+P_{001}+P_{010}\\
 &a_1\in [0,1], a_2\in[0,1], P_{000}\in [0,(1-a_1)(1-a_2)],  P_{001}\in [0,(1-a_1)a_2], P_{010}\in [0,a_1(1-a_2)]\\
 &P_{000}+P_{001}+P_{010}\in [\frac{1}{2}-a_1a_2,\frac{1}{2}]
\end{align*}
Using the Markov chain  $\hat{X}_1,\hat{X}_2\leftrightarrow \hat{X}_{12}\leftrightarrow X$, we have $P_{X,\hat{X}_1,\hat{X}_2}=\sum_{\hat{x}_{12}}P_{X|\hat{X}_{12}}P_{\hat{X}_1,\hat{X}_2,\hat{X}_{12}}$. So $P_{X,\hat{X}_1,\hat{X}_2}$ is as follows:
\setlength{\extrarowheight}{.3cm}
\begin{table}[h]
\centering
\begin{tabular}{x{1.2cm}|x{3cm}|x{3cm}|x{3cm}|x{4cm}|}
\diag{.05em}{1.2cm}{$\ \ \  X$}{$\hat{X}_1,\hat{X}_2$}& \multicolumn{1}{c}{00}&\multicolumn{1}{c}{01}&\multicolumn{1}{c}{10}&\multicolumn{1}{c}{11}\\\cline{1-5}
0&$(1-D_0)P_{000}+D_0((1-a_1)(1-a_2)-P_{000})$&$(1-D_0)P_{001}+D_0((1-a_1)a_2-P_{001})$&$(1-D_0)P_{010}+D_0(a_1(1-a_2)-P_{010})$&$(1-D_0)(\frac{1}{2}-P_{000}-P_{001}-P_{010})+D_0(a_1a_2-\frac{1}{2}+P_{000}+P_{001}+P_{010})$\\  \cline{2-5}
1&$D_0P_{000}+(1-D_0)((1-a_1)(1-a_2)-P_{000})$&$D_0P_{001}+(1-D_0)((1-a_1)a_2-P_{001})$&$D_0P_{010}+(1-D_0)(a_1(1-a_2)-P_{010})$&$D_0(\frac{1}{2}-P_{000}-P_{001}-P_{010})+(1-D_0)(a_1a_2-\frac{1}{2}+P_{000}+P_{001}+P_{010})$\\  \cline{2-5}
\end{tabular}
\vspace{0.1in}
\caption{}
\end{table}
We can minimize the resulting distortion at decoders 1 and 2 by choosing $P_{000}, P_{001}$ and $P_{010}$ optimally. Let $P^*_{X,\hat{X}_1,\hat{X}_2}$ be the optimal joint distribution, we will show that there are two choices for $P^*_{X,\hat{X}_1,\hat{X}_2}$. We have:
\begin{align*}
 &\mathsf{E}(d_H(\hat{X}_1,X))+\mathsf{E}(d_H(\hat{X}_2,X))=P(\hat{X}_1\neq X) +P(\hat{X}_2\neq X)\\
& =(P_{X,\hat{X}_1,\hat{X}_2}(0,0,1)+P_{X,\hat{X}_1,\hat{X}_2}(1,0,1))+(P_{X,\hat{X}_1,\hat{X}_2}(0,1,0)+P_{X,\hat{X}_1,\hat{X}_2}(1,1,0))+2(P_{X,\hat{X}_1,\hat{X}_2}(0,1,1)+P_{X,\hat{X}_1,\hat{X}_2}(1,0,0))\\
&=P_{001}+(1-a_1)a_2-P_{001}+P_{010}+(1-a_2)a_1-P_{010}+2D_0(P_{000}+a_1a_2-\frac{1}{2}+P_{000}+P_{001}+P_{010})\\
&+2(1-D_0)(\frac{1}{2}-P_{000}-P_{001}-P_{010}+(1-a_1)(1-a_2)-P_{000})\\
&=(2D_0-1)a_1+(2D_0-1)a_2+4(2D_0-1)P_{000}+2(2D_0-1)P_{001}+2(2D_0-1)P_{010}-4D_0+3.
\end{align*}
This is an optimization problem on $a_1, a_2, P_{000}, P_{001}, P_{010}$ with respect to the constraints:
\begin{align*}
 &a_1\in [0,1], a_2\in[0,1], P_{000}\in [0,(1-a_1)(1-a_2)],  P_{001}\in [0,(1-a_1)a_2], P_{010}\in [0,a_1(1-a_2)]\\
 &P_{000}+P_{001}+P_{010}\in [\frac{1}{2}-a_1a_2,\frac{1}{2}].
\end{align*}
Also note that for fixed $a_1$ and $a_2$ the problem becomes a linear optimization problem (otherwise the constraints are not linear). So we fix $a_1$ and $a_2$ and optimize $P_{000}, P_{001}$ and $P_{010}$ for each value of $a_1$ and $a_2$. In this case the simplex algorithm provides a straightforward solution. We investigate the solution in several different cases:

\textbf{Case 1:} $(1-a_1)(1-a_2)\geq\frac{1}{2}$:
Note that in the simplex algorithm, the variable with smallest (most negative) coefficient takes its maximum possible value first.Since $D_0<\frac{1}{2}$, $(2D_0-1)<0$, so the algorithm would first maximize the value of $P_{000}$. Since  $(1-a_1)(1-a_2)\geq\frac{1}{2}$, we have $P^*_{000}=\frac{1}{2}$. This along with constraint $P_{000}+P_{001}+P_{010}\in [\frac{1}{2}-a_1a_2,\frac{1}{2}]$ sets $P^*_{001}=0$ and $P^*_{010}=0$. So in this case:
\begin{align*}
 &\mathsf{E}(d_H(\hat{X}_1,X))+\mathsf{E}(d_H(\hat{X}_2,X))
=(2D_0-1)a_1+(2D_0-1)a_2+2(2D_0-1)-4D_0+3\\
&=1+(2D_0-1)(a_1+a_2).
\end{align*}
Now we optimize on $a_1, a_2$ such that $(1-a_1)(1-a_2)\geq\frac{1}{2}$. Increasing $a_1$ or $a_2$ decreases the distortion so the optimal value is achieved when $(1-a_1)(1-a_2)=\frac{1}{2}$, so $a_2=1-\frac{1}{2(1-a_1)}$. We have:
\begin{align*}
 &\mathsf{E}(d_H(\hat{X}_1,X))+\mathsf{E}(d_H(\hat{X}_2,X))
 =1+(2D_0-1)(a_1+1-\frac{1}{2(1-a_1)})
\end{align*}
Optimizing the value of $a_1$, we get $a^*_1=a^*_2=1-\frac{\sqrt{2}}{2}$. These values give $P_{X,\hat{X}_1, \hat{X}_2}=P_{X,V_{\{1\}},V_{\{2\}}}$. Also replacing the values in $P_{\hat{X}_{12},\hat{X}_1,\hat{X}_2}$, we get:
\begin{table}[h]
\centering
\begin{tabular}{x{2cm}|x{2cm}|x{2cm}|x{2cm}|x{2cm}|}
\diag{.2em}{2cm}{$\qquad \hat{X}_{12}$}{$\hat{X}_1,\hat{X}_2$}& \multicolumn{1}{c}{00}&\multicolumn{1}{c}{01}&\multicolumn{1}{c}{10}&\multicolumn{1}{c}{11}\\\cline{1-5}
0&$\frac{1}{2}$&$0$&$0$&$0$\\  \cline{2-5}
1&$0$&$\frac{\sqrt{2}-1}{2}$&$\frac{\sqrt{2}-1}{2}$&$\frac{3-2\sqrt{2}}{2}$\\  \cline{2-5}
\end{tabular}
\vspace{0.1in}
\caption{}
\end{table}
which shows that $\hat{X}_{12}$ is a function of $\hat{X}_1$ and $\hat{X}_2$. 
\textbf{Case 2:} $(1-a_1)(1-a_2)<\frac{1}{2}$, $a_1\leq\frac{1}{2}$:
  In this case the simplex method yields the following set of optimal distributions:
 \begin{align*}
 &P^*_{000}=(1-a_1)(1-a_2), P^*_{001}=\alpha, P^*_{010}=\frac{1}{2}-(1-a_1)(1-a_2)-\alpha, P^*_{011}=0\\
 &P^*_{100}=0, P^*_{101}=(1-a_1)a_2-\alpha, P^*_{010}=(1-a_2)a_1-\frac{1}{2}+(1-a_1)(1-a_2)+\alpha, P^*_{111}=a_1a_2.
\end{align*}
 Where $\alpha\in [a_2-\frac{1}{2},  \frac{1}{2}-(1-a_1)(1-a_2)]$ is an auxiliary variable that does not play a role in the distortion since the coefficients of $P^*_{001}$ and $P^*_{010}$ are equal in the distortion formula. We get:
 
 \begin{align*}
 &\mathsf{E}(d_H(\hat{X}_1,X))+\mathsf{E}(d_H(\hat{X}_2,X))
 =1+(2D_0-1)((1-a_1)(1-a_2)+a_1a_2).
\end{align*}
Note that since $a_1<\frac{1}{2}$, the term $(1-a_1)(1-a_2)+a_1a_2$ is decreasing with $a_2$, so the distortion is increasing with $a_2$ and the optimal values are $a^*_2=max(0, 1-\frac{1}{2(1-a_1)})$, since $a_1\leq\frac{1}{2}$, $a^*_2= 1-\frac{1}{2(1-a_1)}$, replacing $a^*_2$ we have:
 \begin{align*}
 &\mathsf{E}(d_H(\hat{X}_1,X))+\mathsf{E}(d_H(\hat{X}_2,X))
 =1+(2D_{\{1\}}-1)(\frac{1}{2}+a_1(1-\frac{1}{2(1-a_1)})).
\end{align*}
Solving for $a_1$ we get $a_1=1-\frac{1}{\sqrt{2}}$ and in tun $a_2=1-\frac{1}{\sqrt{2}}$ as in the previous case. 

\textbf{Case 3}: $(1-a_1)(1-a_2)<\frac{1}{2}, a_1>\frac{1}{2}, a_1a_2<\frac{1}{2}$:
 The probabilities are as in the last case with $\alpha\in [0,  \frac{1}{2}-(1-a_1)(1-a_2)]$. The distortion is similar to the last case.
Since $a_1>\frac{1}{2}$, the distortion is decreasing in $a_2$. So $a^*_2= \frac{1}{2a_1}$. Which yields:
 \begin{align*}
 &\mathsf{E}(d_H(\hat{X}_1,X))+\mathsf{E}(d_H(\hat{X}_2,X))
 =1+(2D_{\{1\}}-1)((1-a_1)(1-\frac{1}{2a_1})+\frac{1}{2}).
\end{align*}
This would have no solution for optimizing $a_1$ at the given range. 

\textbf{Case 4:} $a_1a_2>\frac{1}{2}$:
By the same arguments the optimal solution is 
\begin{align*}
 &P^*_{000}=(1-a_1)(1-a_2), P^*_{001}=(1-a_1)a_2, P^*_{010}=(1-a_2)a_1, P^*_{011}=0\\
 &P^*_{100}=0, P^*_{101}=0, P^*_{010}=0, P^*_{111}=\frac{1}{2}.
\end{align*}
Then $P^*_{\hat{X}_{12}, \hat{X}_1,\hat{X}_2}$ is:
\begin{table}[h]
\centering
\begin{tabular}{x{2cm}|x{2cm}|x{2cm}|x{2cm}|x{2cm}|}
\diag{.3em}{2cm}{$\qquad \hat{X}_{12}$}{$\hat{X}_1,\hat{X}_2$}& \multicolumn{1}{c}{00}&\multicolumn{1}{c}{01}&\multicolumn{1}{c}{10}&\multicolumn{1}{c}{11}\\\cline{1-5}
0&$\frac{3-2\sqrt{2}}{2}$&$\frac{\sqrt{2}-1}{2}$&$\frac{\sqrt{2}-1}{2}$&$0$\\  \cline{2-5}
1&$0$&$0$&$0$&$\frac{1}{2}$\\  \cline{2-5}
\end{tabular}
\vspace{0.1in}
\caption{}
\end{table}
which is the second choice for the optimal joint distribution. Note that again $\hat{X}_{12}$ is a function of $\hat{X}_{1}$ and $\hat{X}_2$.
\end{proof}
 
 \textbf{Step 5: }  We are left with $\mathcal{C}_{\{1\},\{3\}}$, $\mathcal{C}_{\{2\},\{3\}}$, $\mathcal{C}_{\{1\}}$ ,$\mathcal{C}_{\{2\}}$ and $\mathcal{C}_{\{3\}}$. Let $X_i$ be the reconstruction at decoder $\{i\}$ for $i\in \{1,2,3\}$.
\begin{Lemma}
 The following Markov chains hold:
\begin{align}
&U_{\{1\},\{3\}}, U_{\{1\}}, X_1 \Perp U_{\{2\},\{3\}} ,U_{\{2\}}, X_2 \label{MC50}\\
&U_{\{1\}},U_{\{2\}},U_{\{1\},\{3\}},U_{\{2\},\{3\}}\leftrightarrow X_1,X_2 \leftrightarrow X\label{MC51}\\
 &U_{\{1\},\{3\}},U_{\{1\}} \leftrightarrow X_1 \leftrightarrow X, U_{\{2\},\{3\}}, U_{\{2\}}\label{MC52}\\
 &U_{\{2\},\{3\}},U_{\{2\}} \leftrightarrow X_2 \leftrightarrow X,U_{\{1\},\{3\}}, U_{\{1\}}\label{MC53}\\
 &U_{\{1\},\{3\}},U_{\{2\},\{3\}},U_{\{3\}} \leftrightarrow X_3 \leftrightarrow X\label{MCmiss}\\
 &X_1, X_2, U_{\{1\}}, U_{\{2\}}\leftrightarrow U_{\{1\},\{3\}}U_{\{2\},\{3\}},X\leftrightarrow U_{\{3\}}, X_3\label{MC54}\\
 &U_{\{1\}}\leftrightarrow U_{\{1\},\{3\}}U_{\{2\},\{3\}},X_1,U_{\{3\}}\leftrightarrow X\label{MC55}\\
 &U_{\{2\}}\leftrightarrow U_{\{1\},\{3\}}U_{\{2\},\{3\}},X_2,U_{\{3\}}\leftrightarrow X\label{MC56}\\
 &U_{\{3\}}\leftrightarrow U_{\{1\},\{3\}},U_{\{2\},\{3\}},X_3,X_1 \leftrightarrow X\label{MC57}\\
 &U_{\{3\}}\leftrightarrow U_{\{1\},\{3\}},U_{\{2\},\{3\}},X_3,X_2 \leftrightarrow X\label{MC58}
\end{align}
  \end{Lemma}
\begin{proof}
(\ref{MC50}) holds from Lemma \ref{indeplem}. From the optimality at decoder $\{1,2\}$ and step 4 we have:
\begin{equation*}
I(U_{\{1\},\{3\}},U_{\{2\},\{3\}},U_{\{1\}},U_{\{2\}},X_1,X_2;X)=I(X_1,X_2;X)=1-h_b(D_0),
\end{equation*}
which proves (\ref{MC51}). Next we prove (\ref{MC52}):
\begin{align*}
 P({U_{\{1\},\{3\}},U_{\{1\}},U_{\{2\},\{3\}},U_{\{2\}},X_1,X})&\stackrel{(a)}{=}\sum_{X_2}P({U_{\{1\},\{3\}},U_{\{1\}},X_1})P({X_2,U_{\{2\},\{3\}},U_{\{2\}}})P({X|X_1,X_2})\\
 &=P({U_{\{1\},\{3\}},U_{\{1\}},X_1})P(U_{\{2\},\{3\}},U_{\{2\}})\sum_{X_2}P(X_2|U_{\{2\}},U_{\{2\},\{3\}})P({X|X_1,X_2})\\
 &=P({U_{\{1\},\{3\}},U_{\{1\}},X_1})P(U_{\{2\},\{3\}},U_{\{2\}})P({X|X_1,U_{\{2\}},U_{\{2\},\{3\}}})\\
 &\stackrel{b}{=}P({U_{\{1\},\{3\}},U_{\{1\}},X_1})P(U_{\{2\}},U_{\{2\},\{3\}},X|X_1)\\
\end{align*}
In $(a)$ we have used (\ref{MC50}) and the Markov chain (\ref{MC51}), in $(b)$, we have used (\ref{MC50}). (\ref{MC53}) follows by symmetry. (\ref{MCmiss}) can be proved using optimality at decoder $\{3\}$ and the argument given in the proof of (\ref{MC50}). We proceed with the proof of (\ref{MC54}). Consider the following packing bounds at decoder $\{1,2\}$ and $\{3\}$:
\begin{align*}
 &H(U_{\{1\}},U_{\{2\}},U_{\{1\},\{3\}},U_{\{2\},\{3\}})\leq H(U_{\{1\}})+H(U_{\{2\}})+H(U_{\{1\},\{3\}})+H(U_{\{2\},\{3\}})-r_1-r_2-r_{1,3}-r_{2,3}+R_1+R_2\\
 &H(U_{\{1\},\{3\}},U_{\{2\},\{3\}},U_{\{3\}})\leq H(U_{\{1\},\{3\}})+H(U_{\{2\},\{3\}})+H(U_{\{3\}})-r_{1,3}-r_{2,3}-r_3+R_3\\
\end{align*}
And the following covering bounds:
\begin{align*}
 &H(U_{\{1\}},U_{\{2\}},U_{\{3\}},U_{\{1\},\{3\}},U_{\{2\},\{3\}}|X)\geq H(U_{\{1\}})+H(U_{\{2\}})+H(U_{\{3\}})\\&+H(U_{\{1\},\{3\}})+H(U_{\{2\},\{3\}})-r_1-r_2-r_3-r_{1,3}-r_{2,3}\\
 &H(U_{\{1\},\{3\}},U_{\{2\},\{3\}}|X)\geq H(U_{\{1\},\{3\}})+H(U_{\{2\},\{3\}})-r_{1,3}-r_{2,3}
\end{align*}
Adding all the bounds and simplifying we get:
\begin{align*}
 R_1+R_2+R_3\geq I(U_{\{1\}},U_{\{2\}},U_{\{1\},\{3\}},U_{\{2\},\{3\}};X)+I(U_{\{1\},\{3\}},U_{\{2\},\{3\}},U_{\{3\}};X)+I(U_{\{1\}},U_{\{2\}};U_{\{3\}}|U_{\{1\},\{3\}},U_{\{2\},\{3\}},X)
 \end{align*}
 This resembles the two user sum-rate bound when the first user is sending descriptions 1 and 2 while the second user transmits description 3. From optimality at decoder $\{1\}$2, $R_1+R_2= I(U_{\{1\}},U_{\{2\}},U_{\{1\},\{3\}},U_{\{2\},\{3\}};X)$ and optimality at decoder $\{3\}$ yields $R_3=I(U_{\{1\},\{3\}},U_{\{2\},\{3\}},U_{\{3\}};X)$. So $I(U_{\{1\}},U_{\{2\}};U_{\{3\}}|U_{\{1\},\{3\}},U_{\{2\},\{3\}},X)=0$. This proves (\ref{MC54}). We have:
 
\begin{align*}
&P({U_{\{1\},\{3\}},U_{\{1\}},X_1,U_{\{2\},\{3\}},U_{\{3\}},X})\\&{=} P(U_{\{1\},\{3\}},U_{\{2\},\{3\}},X_1)P(U_{1}|U_{\{1\},\{3\}},U_{\{2\},\{3\}},X_1)P(X|U_{\{1\},\{3\}},U_{\{2\},\{3\}},X_1,U_{\{1\}})P(U_{\{3\}}|U_{\{1\},\{3\}},U_{\{2\},\{3\}},X_1,U_{\{1\}},X)\\
&\stackrel{(a)}{=}P(U_{\{1\},\{3\}},U_{\{2\},\{3\}},X_1)P(U_{1}|U_{\{1\},\{3\}},U_{\{2\},\{3\}},X_1)P(X|U_{\{1\},\{3\}},U_{\{2\},\{3\}},X_1)P(U_{\{3\}}|U_{\{1\},\{3\}},U_{\{2\},\{3\}},X_1,U_{\{1\}},X)\\
&\stackrel{(b)}{=}P(U_{\{1\},\{3\}},U_{\{2\},\{3\}},X_1)P(U_{1}|U_{\{1\},\{3\}},U_{\{2\},\{3\}},X_1)P(X|U_{\{1\},\{3\}},U_{\{2\},\{3\}},X_1)P(U_{\{3\}}|U_{\{1\},\{3\}},U_{\{2\},\{3\}},X)\\
\end{align*}
where $(a)$ follows from (\ref{MC52}) and Lemma \ref{MCmove} given below. $(b)$ follows from (\ref{MC54}). So we have shown that $U_{\{1\}}\leftrightarrow U_{\{1\},\{3\}},U_{\{2\},\{3\}},X_1\leftrightarrow X,U_{\{3\}}$, using Lemma \ref{MCmove} we conclude (\ref{MC55}). (\ref{MC56}) follows by symmetry.  Lastly we prove (\ref{MC57}):
 \begin{align*}
 &P(X,X_1,U_{\{3\}}|U_{\{1\},\{3\}}, U_{\{2\},\{3\}},X_3)\\&{=} P(X|U_{\{1\},\{3\}},U_{\{2\},\{3\}},X_3)P(U_{\{3\}}|U_{\{1\},\{3\}},U_{\{2\},\{3\}},X_3,X)P(X_1|U_{\{1\},\{3\}},U_{\{2\},\{3\}},X,U_{\{3\}},X_3)\\
&\stackrel{a}{=}P(X|U_{\{1\},\{3\}},U_{\{2\},\{3\}},X_3)P(U_{\{3\}}|U_{\{1\},\{3\}},U_{\{2\},\{3\}},X_3,X)P(X_1|U_{\{1\},\{3\}},U_{\{2\},\{3\}},X,X_3)\\
&{=}P(U_{\{3\}}|U_{\{1\},\{3\}},U_{\{2\},\{3\}},X_3,X)P(X,X_1|U_{\{1\},\{3\}},U_{\{2\},\{3\}},X_3)\\
&\stackrel{b}{=}P(U_{\{3\}}|U_{\{1\},\{3\}},U_{\{2\},\{3\}},X_3)P(X,X_1|U_{\{1\},\{3\}},U_{\{2\},\{3\}},X_3)
\end{align*}
where $(a)$ follows form \ref{MC54}. $(b)$ holds because of (\ref{MCmiss}). (\ref{MC57}) follows from lemma \ref{MCmove}.
\begin{Lemma}
 For random variables $A, B, C$ and $D$ if we have $A,B\leftrightarrow C \leftrightarrow D$ then $A\leftrightarrow B,C \leftrightarrow D$. \label{MCmove}
\end{Lemma}
\begin{proof}
We have:
\begin{align*}
 P(A,D|B,C)=\frac{P(A,B,C,D)}{P(B,C)}=\frac{P(C)P(A,B|C)P(D|C)}{P(C)P(B|C)}=P(A|BC)P(D|C)=P(A|BC)P(D|BC)
\end{align*}
\end{proof}
\end{proof}
Next we argue that if we set $U_{\{1\}}$ to be equal to $X_1$ there would be no change in distortion and the rate does not increase. First consider decoder $\{1,3\}$. The optimal reconstruction function is given by \\$argmax_x(P_{X|U_{\{1\},\{3\}} U_{\{2\},\{3\}}U_{\{1\}}U_{\{3\}}} (x|u_{1,3},u_{23},u_{\{1\}},u_{\{3\}}))$. We have:
\begin{align*}
 &argmax_x(P_{X|U_{\{1\},\{3\}}U_{\{2\},\{3\}}U_{\{1\}}U_{\{3\}}}(x|u_{1,3},u_{23},u_{\{1\}},u_{\{3\}}))\\
 &\stackrel{(a)}{=}argmax_x(P_{X|U_{\{1\},\{3\}}U_{\{2\},\{3\}}U_{\{1\}}U_{\{3\}}X_1}(x|u_{1,3},u_{23},u_{\{1\}},u_{\{3\}},x_1))\\
 &\stackrel{(b)}{=}argmax_x(P_{X|U_{\{1\},\{3\}}U_{\{2\},\{3\}}U_{\{3\}}X_1}(x|u_{1,3},u_{23},u_{\{3\}},x_1))
\end{align*}
where in $(a)$ we used the fact that $x_1$ is a function of
$U_{\{1\}},U_{\{1\},\{3\}}$ and in $(b)$ we use (\ref{MC55}). So the
distortion won't change at decoder $\{1,3\}$.  Also the reconstruction
at decoder $\{1\}$ is $X_1$ so setting $U_{\{1\}}=X_1$ won't change
the reconstruction at this decoder. At decoder $\{1,2\}$ we showed in
step 4 that $X_{12}$ is a function of $X_1,X_2$ where $X_2$ is a
function of $U_{\{2\},\{3\}},U_{\{2\}}$, so setting $U_{\{1\}}=X_1$
does not change the distortion at this decoder either. The rest of the
decoders do not receive $U_{\{1\}}$. As for the rate, note that $X_1$
was reconstructed at all decoders reconstructing $U_{\{1\}}$. So
replacing $U_{\{1\}}$ with $X_1$ does not require sending any extra
information. So we set $U_{\{1\}}=X_1$ without any loss in distortion
and with a potential gain in rate. The same argument combined with the
Markov chains (\ref{MC56}) sets $U_{\{2\}}=X_2$, also using Markov
chains (\ref{MC57}) and (\ref{MC58}) we set $U_{\{3\}}=X_3$.  
\begin{Lemma}
 The following constraints hold:\label{IMlemma}
\begin{align}
&P_{XX_1X_2} \text{ is fixed and equal to $P_{X,V_{\{1\}},V_{\{2\}}}$ in the previous step.}\label{PXX1X2}\\
&P_{XX_3} \text{ is fixed and equal to $P_{X,V_3}$ which is the optimizing distribution for decoder $\{3\}$.} \label{PXX3}\\
&U_{\{1\},\{3\}}\leftrightarrow X_1 \leftrightarrow U_{\{2\},\{3\}}, X, X_2\label{X1Mark}\\
&U_{\{2\},\{3\}}\leftrightarrow X_2 \leftrightarrow U_{\{1\},\{3\}}, X, X_1\label{X2Mark}\\
&U_{\{1\},\{3\}}, U_{\{2\},\{3\}}\leftrightarrow X_3\leftrightarrow X\label{X3Mark}\\
&X_1,X_2 \leftrightarrow X,U_{\{1\},\{3\}},U_{\{2\},\{3\}}\leftrightarrow X_3\label{123Mark}
\end{align}
\end{Lemma}
 
\begin{proof}
 (\ref{PXX1X2}) was proved in the step 4. (\ref{PXX3}) follows from PtP optimality at decoder $\{3\}$. (\ref{X1Mark}) follows from (\ref{MC52}), (\ref{X2Mark}) follows from (\ref{MC53}). (\ref{X3Mark}) follows from (\ref{MCmiss}). (\ref{123Mark}) follows from (\ref{MC54}). 
\end{proof}
 We proceed by bounding the cardinality of $\mathsf{U}_{\{1\},\{3\}}$ and $\mathsf{U}_{\{2\},\{3\}}$. Using Lemma \ref{IMlemma}, the joint distribution between the random variables is given as follows:
\begin{align}
& P(U_{\{1\},\{3\}},U_{\{2\},\{3\}},X_1,X_2,X_3,X)=P(U_{\{1\},\{3\}},X_1)P(U_{\{2\},\{3\}},X_2)P(X|X_1,X_2)P(X_3|U_{\{1\},\{3\}},U_{\{2\},\{3\}},X)\nonumber\\
&=P(U_{\{1\},\{3\}},X_1)P(U_{\{2\},\{3\}},X_2)P(X|X_1,X_2)\frac{P(U_{\{1\},\{3\}}U_{\{2\},\{3\}}X_3X)}{P(U_{\{1\},\{3\}}U_{\{2\},\{3\}}X)}\nonumber\\
&=P(U_{\{1\},\{3\}},X_1)P(U_{\{2\},\{3\}},X_2)P(X|X_1,X_2)\frac{P(U_{\{1\},\{3\}}U_{\{2\},\{3\}}|X_3)P(X_3X)}{\sum_{X_1,X_2}P(U_{\{1\},\{3\}},X_1)P(U_{\{2\},\{3\}},X_2)P(X|X_1,X_2)}\label{joint}
\end{align}
Also note that we have the following equality:
\begin{equation*}
 P(U_{\{1\},\{3\}},U_{\{2\},\{3\}},X)=\sum_{X_3}P(X,X_3)P(U_{\{1\},\{3\}}U_{\{2\},\{3\}}|X_3)=\sum_{X_1,X_2}P(U_{\{1\},\{3\}}|X_1)P(U_{\{2\},\{3\}}|X_2)P(XX_1,X_2)
\end{equation*}
 Denote $P(X,X_1,X_2)=P_{xx_1x_2}$ and $P(U_{\{1\},\{3\}}=\theta|X_1=i)=\alpha_i(\theta), \theta\in \mathsf{U}_{1,3}, i\in \{0,1\}$ and $P(U_{\{2\},\{3\}}=\gamma|X_2=i)=\beta_i(\gamma), \gamma\in \mathsf{U}_{2,3},i\in\{0,1\}$. We have:
 \begin{align*}
 P_{U_{\{1\},\{3\}}U_{\{2\},\{3\}}|X_3}(\theta, \gamma|0)P_{X_3,X}(0,0)+& P_{U_{\{1\},\{3\}}U_{\{2\},\{3\}}|X_3}(\theta, \gamma|1)P_{X_3,X}(1,0)\\&=\alpha_{0}(\theta)\beta_{0}(\gamma)P_{000}+\alpha_{0}(\theta)\beta_{1}(\gamma)P_{001}+
\alpha_{1}(\theta)\beta_{0}(\gamma)P_{010}+\alpha_{1}(\theta)\beta_{1}(\gamma)P_{011}
\end{align*}
\begin{align*}
 P_{U_{\{1\},\{3\}}U_{\{2\},\{3\}}|X_3}(\theta, \gamma|0)P_{X_3,X}(0,1)+& P_{U_{\{1\},\{3\}}U_{\{2\},\{3\}}|X_3}(\theta, \gamma|1)P_{X_3,X}(1,1)\\&=\alpha_{0}(\theta)\beta_{0}(\gamma)P_{100}+\alpha_{0}(\theta)\beta_{1}(\gamma)P_{101}+
\alpha_{1}(\theta)\beta_{0}(\gamma)P_{110}+\alpha_{1}(\theta)\beta_{1}(\gamma)P_{111}
\end{align*}
Using the values given in Table (\ref{tb:jointdist}), we solve the system of equations:
\begin{align*}
 &P_{U_{\{1\},\{3\}}U_{\{2\},\{3\}}|X_3}(\theta,\gamma|0)=\alpha_0(\theta)\beta_0(\gamma)\\
 &P_{U_{\{1\},\{3\}}U_{\{2\},\{3\}}|X_3}(\theta,\gamma|1)=\frac{\sqrt{2}-1}{2}(\alpha_1(\theta)\beta_1(\gamma)-\alpha_0(\theta)\beta_0(\gamma))+\frac{1}{2}(\alpha_0(\theta)\beta_1(\gamma)+\alpha_1(\theta)\beta_0(\gamma))\\
\end{align*}
Hence the distribution in \ref{joint} is completely determined by $\alpha_i$ and $\beta_i, i\in \{0,1\}$. 
\begin{Lemma}
 Assume there exists $\alpha_i$ and $\beta_i$, such that $D_{\{1,3\}}\leq D_0$, then $I(U_{\{1\},\{3\}}U_{\{2\},\{3\}}X_1X_3;X)\geq 1-h_b(D_0)$. 
\end{Lemma}
\begin{proof}
 The proof follows from Shannon's rate distortion function for PtP source coding. 
\end{proof}
Based on the previous lemma it is enough to show that for every $\alpha_i$ and $\beta_i$, $I(U_{\{1\},\{3\}}U_{\{2\},\{3\}}X_1X_3;X)< 1-h_b(D_0)$, in that case we have a contradiction. We need to maximize $I(U_{\{1\},\{3\}}U_{\{2\},\{3\}}X_1X_3;X)$ as a function of $\alpha_i$ and $\beta_i$. We use the following lemma:
\begin{Lemma} \cite{ElGamalLec}
 Let $\mathsf{X}$ be a finite set and $\mathsf{U}$ be an arbitrary set. Let $\mathcal{P}(\mathsf{X})$ be a set of pmfs on $\mathsf{X}$
 and $p(x|u)$ be a collection of pmfs on $\mathsf{X}$ for every $u\in \mathsf{U}$. Let $g_j, j\in [1,d]$
 be real-valued continuous functions on $\mathcal{P}(\mathsf{X})$. Then for every $U\sim F(u)$ defined
 on $\mathsf{U}$, there exists random variable $U'\sim p(u')$ with cardinality $|\mathsf{U}'|\leq d$
 and a collection of conditional pmfs $p(u'|x)$ on $\mathsf{X}$ for every $u'\in \mathsf{U}'$ such that
 for every $j\in [1,d]$:
\begin{equation*}
 \int_{\mathsf{U}}g_j(p_{X|U}(x|u))dF(u)=\sum_{u'}g_j(p_{X|U'}(x|u'))p(u')
\end{equation*}
\end{Lemma}
\begin{figure}[!h]
\captionsetup{justification=centering}

\includegraphics[width=5in]{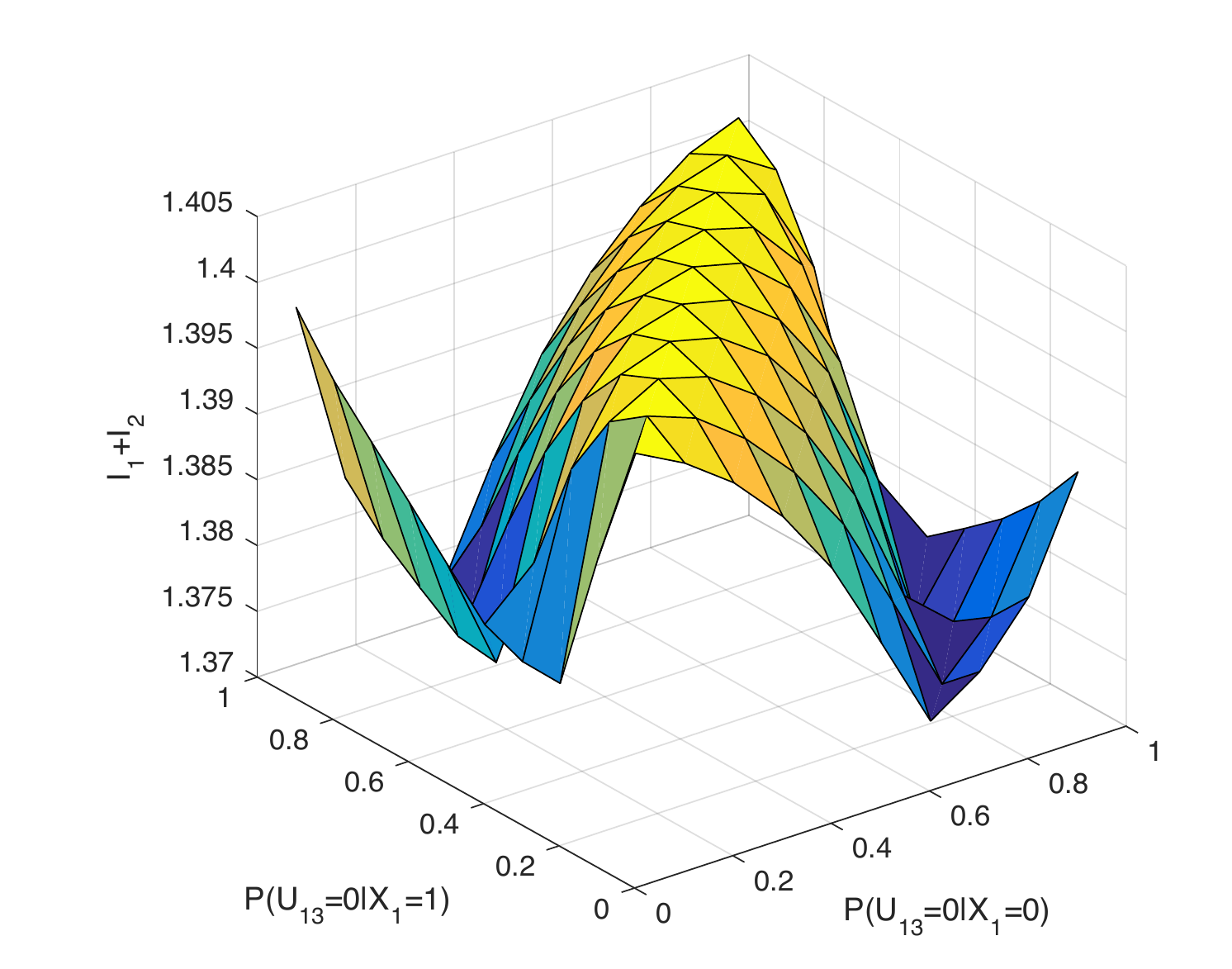}
\caption{Plot of maximum value of $I(U_{\{1\},\{3\}},U_{\{2\},\{3\}},X_1,X_3;X)+I(U_{\{1\},\{3\}},U_{\{2\},\{3\}},X_2,X_3;X)$}
\label{figD}
\end{figure}
We want to use the lemma to bound cardinality of
$\mathsf{U}_{1,3}$. Take
$g_1(p_{U_{\{2\},\{3\}},X_1,X_2,X_3,X}|U_{\{1\},\{3\}})=p_{X_1|U_{\{1\},\{3\}}}(1|u_{13})$
and
$g_{2}(p_{U_{\{2\},\{3\}},X_1,X_2,X_3,X}|U_{\{1\},\{3\}})=H(X|U_{\{2\},\{3\}},X_1,X_3,X,U_{\{1\},\{3\}}=u_{1,3})$. Note
that fixing the expectation on $g_1$ fixes the joint distribution in
(\ref{joint}) and fixing the expectation of $g_2$ fixes the term we
want to minimize. So for any $U_{\{1\},\{3\}}$ minimizing
$I(U_{\{1\},\{3\}}U_{\{2\},\{3\}}X_1X_3;X)$ , there exists $U'_{1,3}$
with cardinality at most 2, such that the joint distribution and
$I(U_{\{1\},\{3\}}U_{\{2\},\{3\}}X_1X_3;X)$ are the same. So it is
enough to search over $U_{\{1\},\{3\}}$ with cardinality 2. The same
arguments hold for bounding the cardinality of $\mathsf{U}_{2,3}$. For
this size of random variables, computer-assisted calculation shows
that
$I(U_{\{1\},\{3\}},U_{\{2\},\{3\}},X_1,X_3;X)+I(U_{\{1\},\{3\}},U_{\{2\},\{3\}},X_2,X_3;X)<
1.42<2(1-h_b(D_0))=1.58$ as shown in Figure \ref{figD}. So we have a
contradiction and the SSC does not achieve the RD vector.  
\end{proof}
\section{Proofs for Section \ref{sec:RDregion}}
\subsection{Proof of lemma \ref{lemma5}}\label{Ap:lemma5}
\begin{proof}
Index the inequalities in the SSC from 1 to $K$. For every inequality in the linear coding region (LCR), there exists a unique inequality in the SSC with the same left hand side, index this inequality with the same index used in the RCR. Let $I_1>R$ be a bound resulting from applying FME on the SSC. Assume the bound results from adding inequalities indexed $i_1,i_2,\ldots,i_k$, it is straightforward to show that adding inequalities with the same indices in the LCR gives the same bound. The reason is that by our construction, the left-hand sides would be the same. In the right-hand side, due to the FME, the terms involving $r_A$ would be eliminated. Define $r'_{A}=r_A-H(U_A)$ and $r'_{o,A}=r_{o,A}-\log(q)$, eliminating $r_A$ is equivalent to eliminating $r'_A$ or $r'_{o,A}$.
\end{proof}

{\Large {\bf Acknowledgment:} }
The authors would like to thank Prof. Kenneth Rose of UC Santa Barbara and 
Mohsen Heidari Khoozani of Univ. of Michigan for helpful discussions.



%

%
 \end{document}